\documentclass[10pt]{article}
\usepackage{geometry,url,graphicx}
\usepackage{amsmath,amssymb,amsthm,enumerate}
\usepackage{epstopdf}
\usepackage{newcent}
\usepackage{color}

\definecolor{dartmouthgreen}{rgb}{0.05, 0.5, 0.06}
\definecolor{ferngreen}{rgb}{0.31, 0.47, 0.26}
\definecolor{shamrockgreen}{rgb}{0.0, 0.62, 0.38}

\usepackage{hyperref}
\hypersetup{colorlinks=true,allcolors=blue}

\def\today{\number\day\ \ifcase\month\or
	January\or February\or March\or April\or May\or June\or
	July\or August\or September\or October\or November\or December\fi
	\space \number\year}
\def\gobble#1#2{}
\def\shortdate{\number\day/\number\month/\expandafter\gobble\number\year}

\usepackage{fancyhdr}

\def\rmd{\mathrm{d}}
\def\rme{\mathrm{e}}
\def\rmi{\mathrm{i}}
\def\ds{\displaystyle}
\def\ss{\scriptstyle}

\def\PI{\mbox{\rm P$_{\rm I}$}}
\def\PII{\mbox{\rm P$_{\rm II}$}}
\def\PIII{\mbox{\rm P$_{\rm III}$}}
\def\PIV{\mbox{\rm P$_{\rm IV}$}}
\def\PV{\mbox{\rm P$_{\rm V}$}}
\def\PVI{\mbox{\rm P$_{\rm VI}$}}

\def\dPI{{\rm d\PI}}
\def\dPII{{\rm d\PII}}
\def\dPIII{{\rm d\PIII}}

\def\a{\alpha}
\def\b{\beta}
\def\ga{\gamma}
\def\de{\delta}
\def\ep{\varepsilon}
\def\ph{\varphi}
\def\k{\kappa}

\def\la{\lambda}
\def\th{\theta}
\def\vth{\vartheta}

\def\p{Painlev\'{e}}

\def\peq{\p\ equation}
\def\peqs{\p\ equations}
\def\dpeq{discrete \p\ equation}
\def\dpeqs{discrete \p\ equations}
\def\bk{B\"ack\-lund}
\def\bt{B\"ack\-lund transformation}
\def\bts{B\"ack\-lund transformations}

\def\cdot{{\,\scriptscriptstyle\bullet\,}}
\def\hide#1{}
\def\ifrac#1#2{{#1/#2}}
\def\hf{\tfrac{1}{2}}
\def\hfz{\tfrac{1}{2}z}

\newcommand{\LaguerreL}[2]{L^{(#1)}_{#2}}

\newcommand{\Integer}{\mathbb{Z}}

\newcommand{\Com}{\mathbb{C}}

\newcommand{\deriv}[3][]{\frac{\rmd^{#1}{#2}}{{\rmd{#3}}^{#1}}}

\def\Matrix#1{\begin{matrix}#1\end{matrix}}

\def\wz{\deriv{w}{z}}
\def\wzz{\deriv[2]{w}{z}}
\def\Lag#1#2{L_{#1}^{(#2)}}

\newtheorem{theorem}{Theorem}[section]

\newtheorem{lemma}[theorem]{Lemma}
\newtheorem{corollary}[theorem]{Corollary}
\theoremstyle{definition}
\newtheorem{definition}[theorem]{Definition}
\newtheorem{example}[theorem]{Example}
\newtheorem{remark}[theorem]{Remark}
\newtheorem{remarks}[theorem]{Remarks}

\numberwithin{figure}{section}
\numberwithin{equation}{section}
\numberwithin{table}{section}

\newcommand{\comment}[1]{}
\def\beq{\begin{equation}}
\def\eeq{\end{equation}}
\def\wh#1{\widehat{#1}}
\def\wt#1{\widetilde{#1}}
\def\pms#1#2#3{(\a,\b,\ga)=\big(#1,#2,#3\big)}

\def\pmsn#1#2#3#4{(\a_{#1},\b_{#1},\ga_{#1})=\big(#2,#3,#4\big)}

\def\D{{\rm D}}

\def\Lag#1#2{L_{#1}^{(#2)}}
\def\Tmn#1#2{T_{#1}^{(#2)}}
\def\Umn#1#2{U_{#1}^{(#2)}}

\usepackage{tikz} 
\usepackage{youngtab}

\def\T{\mathcal{T}}
\DeclareMathOperator{\Wr}{Wr}

\def\=#1{\overline{#1}}

\def\etal{\textit{et al.}}

\def\rarrow{-\!\!\!-\!\!\!-\!\!\!\longrightarrow}
\def\larrow{\longleftarrow\!\!\!-\!\!\!-\!\!\!-}
\def\arr#1{\enskip\begin{array}{c}#1\\[-8pt] \rarrow\end{array} \enskip}
\def\darrow#1#2{\enskip\begin{array}{c}
{\ss{#1}}\\[-8pt] \rarrow\\[-8pt] \larrow\\[-7pt] {\ss{#2}}
\end{array} \enskip}

\def\Hir#1#2{\D_{z}\big(#1\cdot#2\big)}
\def\dln{\deriv{}{z}\ln}
\def\Ln#1#2{\deriv{}{z}\ln\frac{#1}{#2}}

\def\A#1{\widetilde{A}_{#1}^{(1)}}
\def\spv{s\PV}
\def\bn#1#2{\Big[\substack{#1\\[2pt]#2}\Big]}
\newcommand{\Zn}[2][]{\mathcal{Z}_{#2}^{#1}(\hf z)}

\begin{document}

\title{Discrete equations from \bts\\ of the fifth \peq}

\author{Peter A.\ Clarkson$^{1}$\footnote{Corresponding author}, Clare Dunning$^{2}$ and Ben Mitchell$^{1}$\\[5pt]
$^{1}$School of Engineering, Mathematics and Physics,\\ University of Kent, Canterbury, CT2 7NF, UK\\[5pt] 
$^{2}$School of Mathematics, University of Leeds, Leeds, LS2 9JT, UK\\[5pt]
Email: {P.A.Clarkson@kent.ac.uk, T.C.Dunning@leeds.ac.uk, bsjm2@kent.ac.uk}
}
\maketitle
 
\begin{abstract} In this paper discrete equations are derived from B\"{a}cklund transformations of the fifth Painlev\'{e} equation, including a new discrete equation which has ternary symmetry. There are two classes of rational solutions of the fifth Painlev\'{e} equation, one expressed in terms of the generalised Laguerre polynomials and the other in terms of the generalised Umemura polynomials, both of which can be expressed as Wronskians of Laguerre polynomials. Hierarchies of rational solutions of the discrete equations are derived in terms of the generalised Laguerre and generalised Umemura polynomials. It is known that there is non-uniqueness of some rational solutions of the fifth Painlev\'{e} equation. Pairs of non-unique rational solutions are used to derive distinct hierarchies of rational solutions which satisfy the same discrete equation.
\end{abstract}



\tableofcontents

\section{Introduction}
In this paper we are concerned with \dpeqs\ associated with the fifth \peq 
\beq\wzz=\left(\frac{1}{2w}+\frac{1}{w-1}\right)\left(\wz\right)^{\!\!2}-\frac{1}{z}\wz+\frac{(w-1)^2(\a w^2+{\b})}{z^2w}+\frac{\ga w}{z} 
+\frac{\de w(w+1)}{w-1},\label{eq:p5gen}\eeq
with $\a$, $\b$, $\ga$ and $\de$ constants. We consider the generic case of \eqref{eq:p5gen} when $\de\not=0$, so we set
$\de=-\hf$, without loss of generality (by rescaling $z$ if necessary) and obtain
\beq\wzz=\left(\frac{1}{2w}+\frac{1}{w-1}\right)\left(\wz\right)^{\!\!2}-\frac{1}{z}\wz+\frac{(w-1)^2(\a w^2+{\b})}{z^2w}+\frac{\ga w}{z}-\frac{w(w+1)}{2(w-1)},\label{eq:pv}\eeq
which we will refer to as \PV. In the case when $\de=0$ in \eqref{eq:p5gen}
\beq\wzz=\left(\frac{1}{2w}+\frac{1}{w-1}\right)\left(\wz\right)^{\!\!2}-\frac{1}{z}\wz+\frac{(w-1)^2(\a w^2+{\b})}{z^2w}+\frac{\ga w}{z},\label{eq:p5deg}\eeq
which is known as the \textit{degenerate fifth \peq}\ (deg-\PV).

The six \peqs\ were discovered by \p, Gambier and their colleagues in the late 19th and early 20th centuries whilst studying second order ordinary differential equations of the form
\begin{equation} \label{gen-ode}
\wzz=F\left(z,w,\wz\right), 
\end{equation}
where $F$ is rational in $\rmd w/\rmd z$ and $w$ and analytic in $z$. The \p\ transcendents, i.e.\ the solutions of the \peqs, can be thought of as nonlinear analogues of the classical special functions. 
The \peqs\ are regarded as being ``completely integrable" and solvable using Riemann-Hilbert techniques \cite{refFIKN}.
The general solutions of the \peqs\ are transcendental in the sense that they cannot be expressed in terms of known elementary functions and so require the introduction of a new transcendental function to describe their solution. However, it is well known that all the \peqs, except \PI, possess rational solutions, algebraic solutions and solutions expressed in terms of the classical special functions 
for special values of the parameters, see, e.g.\ \cite{refPAC06review,refGLS}
and the references therein. These hierarchies are usually generated from ``seed solutions'' using the associated \bk\ transformations and frequently can be expressed in the form of determinants. 

Discrete \peqs\ also have a long history. The discrete \p\ I (\dPI) equation is 
\beq\label{eq:dPI} x_{n+1}+x_{n}+x_{n-1}=1+\frac{\la n+\rho}{x_{n}}, \eeq 
with $\la$ and $\rho$ constants. Equation \eqref{eq:dPI} was studied in 1939 by Shohat \cite{refShohat}, developing ideas due to Laguerre \cite{refLag} dating from 1885, in the study of orthogonal polynomials associated with the quartic weight $\exp(-\tfrac{1}{4}x^4)$ on the real line.
Subsequently this problem was studied by Freud \cite{refFreud}, see also \cite{refBNevai}, though it was not recognised at the time as a \dpeq. 
Equation \eqref{eq:dPI} was then rediscovered in the 1980s in models of gauge ﬁeld theory by Bessis, Itzykson and Zuber \cite{refBIZ} and in quantum gravity by Br\'{e}zin and Kazakov \cite{refBreKaz}, who showed that its continuum limit is the first \peq\ (\PI)
\beq\wzz=6w^2+z,\label{eq:p1}\eeq
see also \cite{refFIK91}.

The discrete \p\ II (\dPII) equation is 
\beq x_{n+1}+x_{n-1}=\frac{\zeta[(n+\la) x_{n}+\rho]}{1-x_{n}^2},\label{eq:dPII}\eeq
with $\zeta$, $\la$ and $\rho$ constants, which was derived by Periwal and Shevitz \cite{refPS} in a study of unitary matrix models in a double scaling limit using orthogonal polynomials on the unit circle.
Around the same time, Nijhoff and Papageorgiou \cite{refNP91} independently derived \dPII\ \eqref{eq:dPII} and showed that its continuum limit is the second \peq\ (\PII)
\beq\wzz=2w^3+z w+\a,\label{eq:p2}\eeq
with $\a$ a constant; Periwal and Shevitz \cite{refPS} gave a continuum limit to \PII\ \eqref{eq:p2} with $\a=0$.

One method for deriving \dpeqs\ is using \bts\ of (continuous) \peqs. For example, it is well known that \PII\ \eqref{eq:p2} has the \bts\ 
\begin{subequations}\begin{align}
w(z;\a+1)&=-w(z;\a)-\frac{2\a+1}{2w^2(z;\a)+2w'(z;\a)+z},\\
w(z;\a-1)&=-w(z;\a)-\frac{2\a-1}{2w^2(z;\a)-2w'(z;\a)+z}.
\end{align}\end{subequations}
with $'\equiv \rmd/\rmd z$ \cite{refFA82,refGambier10}.
Hence letting $x_{n}=w(z;n+\la-\hf)$ and $\a=n+\la-\hf$, 
then by eliminating $w'(z;\a)$ we obtain another discrete equation, commonly known as alt-\dPI\ 
\beq \frac{n+\la}{x_{n+1}+x_{n}}+\frac{n+\la-1}{x_{n}+x_{n-1}}+2x_{n}^2+z=0,\label{alt-dPI}\eeq
which has \PI\ \eqref{eq:p1} as its continuum limit \cite{refFGR}.
The alt-\dPI\ equation \eqref{alt-dPI} appeared in 1981 in the work of Jimbo and Miwa \cite[equation (5.3)]{refJM}, though it was not recognised to be a \dpeq\ until the work of Fokas, Grammaticos and Ramani \cite{refFGR}. 

In some respects the \dpeqs\ are much richer than the \peqs. For example, they can have solutions which have no continuum limit and other properties which are lost in the continuous limit. Perhaps the most fundamental difference is that, whilst the \peqs\ are unique up to a M\"obius
transformation, this is not true for the \dpeqs. For each \dpeq\ there exists a number of different forms, as illustrated by \dPI\ \eqref{eq:dPI} and alt-\dPI\ \eqref{alt-dPI} which both have \PI\ \eqref{eq:p1} as a continuum limit yet are not equivalent as there is no transformation relating the solutions \eqref{eq:dPI} and \eqref{alt-dPI}. 
Further their exact solutions are quite different since \dPI\ \eqref{eq:dPI} arises from \bts\ of the fourth \peq\ (\PIV)
\beq \wzz=\frac{1}{2w}\left(\wz\right)^{\!2}+\tfrac{3}{2}w^3+4z w^2+2(z^2-\a)w+\frac{\b}{w},\label{eq:piv}\eeq
with $\a$ and $\b$ constants \cite{refGRP97}, whereas alt-\dPI\ \eqref{alt-dPI} arises from \bts\ of \PII\ \eqref{eq:p2}, as shown above.

The nomenclature for \dpeqs\ can be misleading. The original naming of \dpeqs\ came from their continuum limits but this can be ambiguous. For example, there are generalisations of \dPI\ \eqref{eq:dPI} and \dPII\ \eqref{eq:dPII}, see equations \eqref{eq:adPI} and \eqref{eq:adPII}, respectively, which have continuum limits to other \peqs, as discussed below. Another approach is to classify the \dpeqs\ by the type of affine Weyl group and associated affine Dynkin diagram they are derived from, which is now commonly used.
Sakai \cite{refSakai}, extending the earlier work of Okamoto \cite{refOkamoto79} for the (continuous) \peqs, approached the \dpeqs\ from the point of view of algebraic geometry. Sakai clarified the algebraic nature of \dpeqs\ and gave a classification scheme. A comprehensive review of the geometric aspects of the continuous and \dpeqs\ is given by Kajiwara, Noumi and Yamada \cite{refKNY}. There is a recent discussion of the symmetry group of \dPII\ \eqref{eq:dPII} in \cite{refDSSW}. Discrete \peqs\ arise in various applications such as random matrix theory and orthogonal polynomials, cf.~\cite{refWVAbk,refWVA2020}.

There are three distinct types of \dpeqs: (i), those for which the discrete variable is additive, which is the case for the discrete equations we discuss; (ii), the $q$-\dpeqs\ for which the discrete variable is multiplicative; and (iii), elliptic \dpeqs\ whose coefficients are elliptic functions of the independent variable $n$.

From the works of Okamoto \cite{refOkamotoP2P4,refOkamotoPV,refOkamotoPVI,refOkamotoPIII}, it is known that the parameter spaces of \PII--\PVI\ all admit the action of an
extended affine Weyl group; the group acts as a group of \bts. 
In a series of papers, Noumi and Yamada \cite{refNoumi,refNY98i,refNY98iii,refNY04}
have implemented this idea to derive a hierarchy of dynamical systems
associated to the affine Weyl group of type $\A{N}$, often known as ``\textit{symmetric forms of the \peqs}". The behaviour of each dynamical
system varies depending on whether $N$ is even or odd. 
The first member of the $\A{2n+1}$ hierarchy, i.e.\ $\A{3}$, usually known as symmetric-\PV\ (\spv), is equivalent to \PV\ \eqref{eq:pv}
and given by
\begin{subequations}\label{eq:spv}\begin{align}
z\deriv{f_{0}}{z}&=f_{0}f_{2}(f_{1}-f_{3})+(\hf-\vth_{2})f_{0}+\vth_{0}f_{2},\\
z\deriv{f_{1}}{z}&=f_{1}f_{3}(f_{2}-f_{0})+(\hf-\vth_{3})f_{1}+\vth_{1}f_{3},\\
z\deriv{f_{2}}{z}&=f_{2}f_{0}(f_{3}-f_{1})+(\hf-\vth_{0})f_{2}+\vth_{2}f_{0},\\
z\deriv{f_{3}}{z}&=f_{3}f_{1}(f_{0}-f_{2})+(\hf-\vth_{1})f_{3}+\vth_{3}f_{1},
\end{align}
with the normalisations
\beq\label{eq:sPVa}
f_{0}(z)+f_{2}(z)=\sqrt{z},\qquad f_{1}(z)+f_{3}(z)=\sqrt{z},
\eeq
and $\vth_{0}$, $\vth_{1}$, $\vth_{2}$ and $\vth_{3}$ are constants such that
\beq\label{eq:sPVk}
\vth_{0}+\vth_{1}+\vth_{2}+\vth_{3}=1.
\eeq\end{subequations}
The symmetric system \spv\ \eqref{eq:spv} was found by Adler \cite{refAdler} in the context of periodic chains of \bts\ and has applications in random matrix theory, cf.~\cite{refFW02}. 

The outline of this paper is as follows.
In \S\ref{sec:bts} we describe \bts\ of \PV\ \eqref{eq:pv}.
In \S\ref{sec:deqn} we derive some \dpeqs\ from these \bts, including a new discrete equation which has a ternary symmetry. There have been earlier studies deriving \dpeqs\ from \bts\ of \PV\ \eqref{eq:pv}, see \cite{refFGR,refTGR02}. Our study explicitly gives the \bts\ which generate the \dpeqs. These then can be used to obtain exact solutions of the discrete equations in terms of known solutions of \PV\ \eqref{eq:pv}.
In \S\ref{sec:PVrats} we describe the rational solutions of \PV\ \eqref{eq:pv} which are expressed in terms of Wronskians of Laguerre polynomials. 
In \S\ref{sec:rats} we derive hierarchies of rational solutions of the \dpeqs\ derived in \S\ref{sec:deqn} in terms of the rational solutions of \PV\ \eqref{eq:pv} given in \S\ref{sec:PVrats}. Further pairs of rational solutions which although different, which satisfy \PV\ \eqref{eq:pv} for the same parameters, are used to derive distinct hierarchies of rational solutions that satisfy the same discrete equation. Finally in \S\ref{sec:dis} we discuss our results.

\section{\label{sec:bts}\bts\ of the fifth \peq}
In this section we describe the \bts\ of \PV\ \eqref{eq:pv}.

\begin{definition} A \textit{\bt} maps a solution of a given \peq\ to another solution of the same \peq,
with different values of the parameters, or another equation.
\end{definition}

\begin{definition}
Let $w=w(z;\a,\b,\ga)$ be a solution of \PV\ \eqref{eq:pv} such that
\beq\label{eq:PT.BTs.P5a}
\Phi(w;z)=z\wz-\ep_{2}aw^2+(\ep_{2}a-\ep_{3}b+\ep_{1}z)w+\ep_{3}b \not=0,\eeq
and \beq\label{eq:PT.BTs.P5a1}\Phi(w;z)-2\ep_{1}z w\not=0,\eeq
where $a=\sqrt{2\a}$ and $b=\sqrt{-2\b}$ and $c=\ga$. Then the \textit{\bt}\ $\T_{\ep_{1},\ep_{2},\ep_{3}}$ is defined by
\begin{equation}\label{eq:PT.BTs.P5b}
\T_{\ep_{1},\ep_{2},\ep_{3}}:\quad
W(z;A,B,C)
=1-\frac{2\ep_{1}z w}{\Phi(w;z)}
=1-\frac{2\ep_{1}z w}{z w'-\ep_{2}aw^2+(\ep_{2}a-\ep_{3}b+\ep_{1}z)w+\ep_{3}b},
\end{equation} 
with $'\equiv\rmd/\rmd z$, $\ep_j=\pm1$, $j=1,2,3$, independently, and for the parameters
\beq\label{eq:PT.BTs.P5c}
\T_{\ep_{1},\ep_{2},\ep_{3}}(a,b,c)=\Big(\hf\big[c+\ep_{1}(1-\ep_{3}b-\ep_{2}a)\big],\hf\big[c-\ep_{1}(1-\ep_{3}b-\ep_{2}a)\big],\ep_{1}(\ep_{3}b-\ep_{2}a)\Big),\eeq
so
\[(A,B,C)=\Big(\tfrac{1}{8}\big[c+\ep_{1}(1-\ep_{3}b-\ep_{2}a)\big]^2,-\tfrac{1}{8}\big[c-\ep_{1}(1-\ep_{3}b-\ep_{2}a)\big]^2,\ep_{1}(\ep_{3}b-\ep_{2}a)\Big),\]
cf.~\cite{refGromak,refGF01a,refGF01b}, \cite[\S39]{refGLS} or \cite[\S32.7(v)]{refDLMF}.
\end{definition}

Note that we are using the same notation as in \cite[\S32.7(v)]{refDLMF} which slightly differs from that used in \cite{refGromak,refGF01a,refGF01b} or 
\cite[\S39]{refGLS}.

\begin{remarks}{\rm
\begin{enumerate}[(i)]\item[]
\item Since $\ep_j=\pm1$, $j=1,2,3$, independently, there are eight distinct transformations $\T_{\ep_{1},\ep_{2},\ep_{3}}$. The parameter $\ep_{1}$ reflects that the transformation $(z,\ga)\to(-z,-\ga)$ is a discrete symmetry of \PV. The parameters $\ep_{2}$ and $\ep_{3}$ reflect that $\a=\hf a^2$ and $\b=-\hf b^2$, respectively, so $a$ and $b$ can take either sign.
\item The \bt\ $\T_{\ep_{1},\ep_{2},\ep_{3}}$ is usually given for the parameters $\a$, $\b$ and $\ga$, in terms of $a$, $b$ and $c$. However in order to apply a sequence of \bts, it's preferable to define the action of \bt\ on the parameters $a$, $b$ and $c$, rather than $\a$, $\b$ and $\ga$, to avoid any ambiguity in taking a square root.
\item We note that the Hamiltonian for \PV\ \eqref{eq:pv} in \cite{refJM} is written in terms of the monodromy parameters $\th_{0}$, $\th_{1}$ and $\th_{\infty}$, which are related to the parameters $\a$, $\b$ and $\ga$ through
\beq \a=\tfrac{1}{8}(\th_{0}-\th_{1}+\th_{\infty})^2,\qquad \b=-\tfrac{1}{8}(\th_{0}-\th_{1}-\th_{\infty})^2,\qquad\ga=1-\th_{0}-\th_{1},\label{SchlesParams}\eeq
see also \cite{refOkamotoPV}, as are the Schlesinger transformations, which relate two solutions of the associated monodromy problem, for \PV\ \eqref{eq:pv}, see \cite{refFMA,refMF}. 
Further the discrete symmetries of \PV\ \eqref{eq:pv}, corresponding to the extended affine Weyl group of type $A_3^{(1)}$, act naturally on the root variables, see, for example, \cite{refCDHM}. Hence it is natural to define the \bts\ for \PV\ \eqref{eq:pv} in terms of the parameters $a=\sqrt{2\a}$, $b=\sqrt{-2\b}$ and $c=\ga$, rather than $\a$, $\b$ and $\ga$.
\item When applying the \bt\ $\T_{\ep_{1},\ep_{2},\ep_{3}}$ followed by the \bt\ $\T_{\sigma_{1},\sigma_{2},\sigma_{3}}$ we use the notation $\T_{\sigma_{1},\sigma_{2},\sigma_{3}}\circ \T_{\ep_{1},\ep_{2},\ep_{3}}$, i.e.
\[\T_{\sigma_{1},\sigma_{2},\sigma_{3}}\big(\T_{\ep_{1},\ep_{2},\ep_{3}}(w;a,b,c)\big)=\T_{\sigma_{1},\sigma_{2},\sigma_{3}}\circ \T_{\ep_{1},\ep_{2},\ep_{3}}(w;a,b,c).\]
{Further if 
\[ \T_{\ep_{1},\ep_{2},\ep_{3}}(w;a,b,c)= \Big(W(w,w',a,b,c);A(a,b,c),B(a,b,c),C(a,b,c)),\]
then
\[
\T_{\sigma_{1},\sigma_{2},\sigma_{3}}\circ \T_{\ep_{1},\ep_{2},\ep_{3}}(w;a,b,c)
= \T_{\sigma_{1},\sigma_{2},\sigma_{3}}\Big(W(w,w',a,b,c);A(a,b,c),B(a,b,c),C(a,b,c)).
\]}
\end{enumerate}
}\end{remarks}

\section{\label{sec:deqn}Deriving \dpeqs}
In this section we derive four \dpeqs\ which are associated with \PV\ \eqref{eq:pv} using the \bts\ described in \S\ref{sec:bts}.

To derive a \dpeq\ from \bts\ of \PV\ \eqref{eq:pv}, we use the \bt\ $\mathcal{R}_+$, which relates a solution $w(z;a,b,c)$ to another solution $w_+(z;a_+,b_+,c_+)$ and its inverse transformation $\mathcal{R}_-=\mathcal{R}_+^{-1}$, which relates $w(z;a,b,c)$ to a third solution $w_-(z;a_-,b_-,c_-)$. Eliminating the derivative between the two \bts\ gives an algebraic relation between $w(z;a,b,c)$, $w_+(z;a_+,b_+,c_+)$ and $w_-(z;a_-,b_-,c_-)$, which is the discrete equation. 
Then setting $w(z;a,b,c)=w_{n}$, ${w_{\pm}}(z;a_{\pm},b_{\pm},c_{\pm})=w_{n\pm1}$, with
$(a,b,c)=(a_{n},b_{n},c_{n})$ and $(a_{\pm},b_{\pm},c_{\pm})=(a_{n\pm1},b_{n\pm1},c_{n\pm1})$, gives the chain
$$\{w_{n-1},{a}_{n-1},{b}_{n-1},{c}_{n-1}\}\darrow{\ds{\mathcal{R}_+}}{\ds\mathcal{R}_-} 
\{w_{n},a_{n},b_{n},c_{n}\}\darrow{\ds{\mathcal{R}_+}}{\ds{\mathcal{R}_-}}\{w_{n+1},{a}_{n+1},{b}_{n+1},{c}_{n+1}\}.$$
The parameters satisfy a linear discrete system derived from \eqref{eq:PT.BTs.P5c}.
See, for example, \cite{refCMW,refFGR} for more details of this procedure.

\begin{remarks}{\rm 
\begin{enumerate}[(i)]\item[] \item 
It is essential that $\mathcal{R}_-$ is the inverse transformation of $\mathcal{R}_+$, {see Example \ref{exam32}}. It is possible to create an algebraic relation between three solutions of \PV\ \eqref{eq:pv} using two \bts\ $\mathcal{R}_1$ and $\mathcal{R}_2$ when $\mathcal{R}_2\not=\mathcal{R}_1^{-1}$, but this won't necessarily result in a \dpeq. For example, in \cite[\S39]{refGLS}, see also \cite{refGF01a,refGF01b}, there are algebraic relations of three (and four) solutions of \PV, but these are not \dpeqs.
\item The Bessel function $J_{n}(z)$ satisfies the differential equation
\beq
z^2\deriv[2]{J_{n}}{z}+ z\deriv{J_{n}}{z}+ (z^2-n^2)J_{n}=0,\label{eqbessel1}
\end{equation}
and the recurrence relations
\begin{subequations}\begin{align}
&z\deriv{J_{n}}{z}(z) = zJ_{n-1}(z) - nJ_{n}(z),\label{eqbessel2}\\
&z\deriv{J_{n}}{z}(z) = nJ_{n}(z) - zJ_{n+1}(z),\label{eqbessel3}\\
&zJ_{n+1}(z) -2n J_{n}(z)+ zJ_{n-1}(z) =0,\label{eqbessel4}
\end{align}\end{subequations}
cf.~\cite[\S10]{refDLMF}. Eliminating $\ds\deriv{J_{n}}{z}(z)$ between \eqref{eqbessel2} and
\eqref{eqbessel3} yields the recurrence relation \eqref{eqbessel4}. Letting $n\to n+1$ in \eqref{eqbessel2} and eliminating $J_{n+1}(z)$ using
\eqref{eqbessel3} gives the differential equation \eqref{eqbessel1}. 
Thus the Bessel function $J_{n}(z)$ satisfies both the differential equation \eqref{eqbessel1} and the difference equation \eqref{eqbessel4}. An analogous situation occurs for the \peqs, except \PI\ \eqref{eq:p1}.

\end{enumerate}
}\end{remarks}

In the next three subsections, we derive discrete equations using the \bts
\[\begin{array}{l@{\qquad}l}
\mathcal{R}_{1}=\T_{1,-1,1}, & \mathcal{R}_{1}^{-1}=\T_{-1,1,-1},\\
\mathcal{R}_{2}=\T_{-1,-1,-1}\circ\T_{-1,1,-1}, &\mathcal{R}_{2}^{-1}=\T_{1,1,1}\circ\T_{1,-1,1},\\
\mathcal{R}_{3}=\T_{-1,1,-1}\circ\T_{1,-1,-1}\circ\T_{1,-1,1}, & \mathcal{R}_{3}^{-1}=\T_{-1,1,1}, 
\end{array}\]
as well as a discrete equation which has a ternary symmetry using the \bt
\[\mathcal{R}_{4}=\mathcal{R}_{3}^2=\T_{-1,1,-1}\circ \T_{1,-1,-1}^2\circ\T_{1,-1,1},\qquad \mathcal{R}_{4}^{-1}=\T_{-1,1,1}^2.\]
{These \bts\ are the only combinations which we have found that give rise to a discrete equation.
\begin{example}\label{exam32}
Consider the \bts\ $\T_{-1,1,1}$ and $\T_{1,-1,1}$. First applying $\T_{-1,1,1}$ then $\T_{1,-1,1}$ gives 
\begin{align*}
\T_{-1,1,1}(w)&=1+\frac{2zw}{zw'-aw^2+(a-b-z)w+b},\qquad
\T_{1,-1,1}\circ\T_{-1,1,1}(w)=w,
\end{align*} 
\comment{Note: changed a + to - in the first equation. The following was the original:
\begin{align*}
\T_{-1,1,1}(w)&=1+\frac{2zw}{zw'+aw^2+(a-b-z)w+b},\qquad
\T_{1,-1,1}\circ\T_{-1,1,1}(w)=w,
\end{align*} }
with $'\equiv \rmd/\rmd z$, where we have used the fact that $w$ satisfies \PV\ \eqref{eq:pv} with 
$\pms{\hf a^2}{-\hf b^2}{c}$.
However applying $\T_{1,-1,1}$ then $\T_{-1,1,1}$ gives
\begin{align*}
\T_{1,-1,1}(w)&=1-\frac{2zw}{zw'+aw^2-(a+b-z)w+b},\\
\T_{-1,1,1}\circ\T_{1,-1,1}(w)&=w+\frac{(a-b-c+1)w(w-1)^2}{zw'+(b+c-1)w^2-(2b+c-1+z)w+b}.
\end{align*}
Hence we conclude that $\T_{1,-1,1}\not=\T_{-1,1,1}^{-1}$ and so these \bts\ can not be used to derive a discrete equation. We remark that the parameters transform as follows
\begin{align*} \T_{1,-1,1}\circ\T_{-1,1,1}(a,b,c)&=(a,-b,c),\\ \T_{-1,1,1}\circ\T_{1,-1,1}(a,b,c)&=\big(\hf(a+b+c-1),\hf(a+b-c+1),a-b+1\big),
\end{align*}
confirming that $\T_{1,-1,1}\not=\T_{-1,1,1}^{-1}$.
\end{example}}

\subsection{\label{ssec:adPII}First discrete equation: asymmetric discrete \PII}
Consider the \bt\ $\mathcal{R}_{1}=\T_{1,-1,1}$ which has an inverse $\mathcal{R}_{1}^{-1}=\T_{-1,1,-1}$, then
\begin{subequations}\label{bt:R1}\begin{align} 
w_{n+1}&=\mathcal{R}_{1}(w_{n})=1-\frac{2z w_{n}}{\ds z w_{n}'+a_{n}w_{n}^2-(a_{n}+b_{n}-z)w_{n}+b_{n}},\label{bt:R1a}\\
w_{n-1}&=\mathcal{R}_{1}^{-1}(w_{n})=1+\frac{2z w_{n}}{\ds z w_{n}'-a_{n}w_{n}^2+(a_{n}+b_{n}-z)w_{n}-b_{n}},\label{bt:R1b}
\end{align}\end{subequations}
with $'\equiv \rmd/\rmd z$, 
where $a_{n}$, $b_{n}$ and $c_{n}$ satisfy 
\begin{align*}
(a_{n+1},b_{n+1},c_{n+1})&=\mathcal{R}_{1}(a_{n},b_{n},c_{n})=\left(\hf\big(a_{n}-b_{n}+c_{n}+1\big),\hf\big(-a_{n}+b_{n}+c_{n}-1\big),a_{n}+b_{n}\right),\\
(a_{n-1},b_{n-1},c_{n-1})&=\mathcal{R}_{1}^{-1}(a_{n},b_{n},c_{n})=\left(\hf\big(a_{n}-b_{n}+c_{n}-1\big),\hf\big(-a_{n}+b_{n}+c_{n}+1\big),a_{n}+b_{n}\right),
\end{align*}
 which have solution
\begin{align}\label{dP2:abc} a_{n}&=\hf\big[n+\la+\rho+(-1)^n\ph\big],\qquad
b_{n}=-\hf\big[n+\la-\rho-(-1)^n\ph\big],\qquad
c_{n}=\rho-(-1)^n\ph,
\end{align}
with $\la$, $\rho$ and $\ph$ arbitrary constants.
From \eqref{bt:R1} we obtain
\begin{subequations}\label{dP2:btsys}\begin{align}
\frac{1}{w_{n+1}-1}&+\frac{1}{2w_{n}}\deriv{w_{n}}{z}+\frac{a_{n}w_{n}^2-(a_{n}+b_{n}-z)w_{n}+b_{n}}{2z w_{n}}=0,\\
\frac{1}{w_{n-1}-1}&-\frac{1}{2w_{n}}\deriv{w_{n}}{z}+\frac{a_{n}w_{n}^2-(a_{n}+b_{n}-z)w_{n}+b_{n}}{2z w_{n}}=0,
\end{align}\end{subequations}
so adding gives
\[\frac{1}{w_{n+1}-1}+\frac{1}{w_{n-1}-1}+\frac{a_{n}w_{n}^2-(a_{n}+b_{n}-z)w_{n}+b_{n}}{z w_{n}}=0,\]
which is equation (3.4) in \cite{refGT96}.

Setting $w_{n}=(q_{n}+1)/(q_{n}-1)$ and $w_{n\pm1}=(q_{n\pm1}+1)/(q_{n\pm1}-1)$ in \eqref{bt:R1} gives 
\begin{subequations}\label{dP2:btsys2}\begin{align} 
q_{n+1}&+\frac{2}{1-q_{n}^2}\deriv{q_{n}}{z}-\frac{2}{z}\,\frac{(n+\la) q_{n}+\rho+(-1)^n\ph}{1-q_{n}^2}=0,\\
q_{n-1}&-\frac{2}{1-q_{n}^2}\deriv{q_{n}}{z}-\frac{2}{z}\,\frac{(n+\la) q_{n}+\rho+(-1)^n\ph}{1-q_{n}^2}=0.
\end{align}\end{subequations}
Hence by adding we obtain 
\beq q_{n+1}+q_{n-1}=\frac{4}{z}\frac{(n+\la) q_{n}+\rho+(-1)^n\ph}{1-q_{n}^2},\label{eq:adPII}\eeq
which is known as the \textit{asymmetric \dPII\ equation}
\cite{refSKGHR} and is a generalisation of \dPII\ \eqref{eq:dPII}.
{Further $q_{n}$ satisfies the second-order equation
\beq \deriv[2]{q_{n}}{z}=\frac{q_{n}}{q_{n}^2-1}\left(\deriv{q_{n}}{z}\right)^{\!2}-\frac{1}{z}\deriv{q_{n}}{z}-\frac{a_{n}^2(q_{n}+1)^2-b_{n}^2(q_{n}-1)^2}{z^2(q_{n}^2-1)}-\frac{c_{n}(q_{n}^2-1)}{2z}+\tfrac{1}{4}q_{n}(q_{n}^2-1),\label{eq2:qn}
\eeq
with $a_{n}$, $b_{n}$ and $c_{n}$ given by \eqref{dP2:abc}.}

In order to consider the alternating term in \eqref{eq:adPII}, i.e.\ the binary symmetry, we set $q_{2n}=x_{n}$ and $q_{2n+1}=y_{n}$ in \eqref{eq:adPII} and obtain the discrete system
\begin{align} \label {sys:dPII}
x_{n+1}+x_{n}&=\frac{4}{z}\frac{(2n+\la+1)y_{n}+\rho-\ph}{1-y_{n}^2},\qquad
y_{n}+y_{n-1}=\frac{4}{z}\frac{(2n+\la)x_{n}+\rho+\ph}{1-x_{n}^2},
\end{align}
which is regarded as a discrete \PIII\ equation (\dPIII) since it admits a continuous limit to the third \peq\ (\PIII) 
\beq \wzz=\frac{1}{w}\left(\wz\right)^{\!2}-\frac{1}{z}\wz+\frac{\a w^2+\b}{w}+\ga w^3+\frac{\de}{w},\label{eq:p3} \eeq
with $\a$, $\b$, $\ga$ and $\de$ constants \cite{refGNPRS,refROSG}.

We note that
\[\mathcal{R}_{1}^2(a,b,c)=(a+1,b-1,c),\] and so \[\mathcal{R}_{1}^{2n}(a,b,c)=(n+a,b-n,c).\]
Therefore since
\begin{align*} 
(a_{0},b_{0},c_{0})&=\big(\hf(\la+\rho+\ph),-\hf(\la-\rho-\ph),\rho-\ph\big),\\
(a_{1},b_{1},c_{1})&=\big(\hf(\la+\rho-\ph+1),-\hf(\la-\rho+\ph+1),\rho+\ph\big),
\end{align*}
then
\begin{align*} 
(a_{2n},b_{2n},c_{2n})&=\Big(\hf\big(2n+\la+\rho+\ph\big),-\hf\big(2n+\la-\rho-\ph\big),\rho-\ph\Big),\\
(a_{2n+1},b_{2n+1},c_{2n+1})&=\Big(\hf\big(2n+1+\la+\rho-\ph\big),-\hf\big(2n+1+\la-\rho+\ph\big),\rho+\ph\Big).
\end{align*}
If $w_{0}(z)$ satisfies \PV\ \eqref{eq:pv} with $\a=\tfrac{1}{8}(\la+\rho+\ph)^2$, $\b=-\tfrac{1}{8}(\la-\rho-\ph)^2$ and $\ga=\rho-\ph$ then
\beq x_{n}=q_{2n}=\frac{\mathcal{R}_{1}^{2n}(w_{0})+1}{\mathcal{R}_{1}^{2n}(w_{0})-1},\qquad 
y_{n}=q_{2n+1}=\frac{\mathcal{R}_{1}^{2n+1}(w_{0})+1}{\mathcal{R}_{1}^{2n+1}(w_{0})-1},\eeq
satisfy \eqref {sys:dPII}.

\begin{remarks}
\begin{enumerate}[(i)]\item[]
\item In the case when $\rho=\ph=0$, equation \eqref{eq2:qn} becomes
\beq \deriv[2]{q_{n}}{z}=\frac{q_{n}}{q_{n}^2-1}\left\{\left(\deriv{q_{n}}{z}\right)^{\!2} -\frac{(4n+\la)^2}{z^2} \right\}-\frac{1}{z}\deriv{q_{n}}{z} 
+\tfrac{1}{4}q_{n}(q_{n}^2-1),\label{eq2:qn0}\eeq 
which is equivalent to \PV\ \eqref{eq:pv} with $\a=-\b=\hf(n+\la)^2$ and $\ga=0$. 
Making the transformation $q_{n}=1/\sqrt{1-u_{n}}$ with $z=2\sqrt{x}$ in equation \eqref{eq2:qn0} gives deg-\PV\ \eqref{eq:p5deg} with $\a=0$, $\b=-2(n+\la)^2$ and $\ga=\hf$.
Solutions of deg-\PV\ \eqref{eq:p5deg} are related to those of \PIII\ \eqref{eq:p3}, cf.~\cite[Theorem 34.2]{refGLS}. 
Solutions of equation\ \eqref{eq2:qn0}, in terms of Bessel functions, are discussed in \cite{refPAC23,refHisakado}.

\item Equation \eqref{eq2:qn0} arises as a reduction of the two-dimensional complex sine-Gordon equation, also known as the {Pohlmeyer-Lund-Regge model}\ \cite{refLund77,refLR,refPohl},
in quantum systems \cite{refCFH05,refCH05,refCH08}, in relativity \cite{refGMS} and in coefficients of the three-term recurrence relation for orthogonal polynomials with respect to the weight $w(\theta)=\exp({t\cos\theta})$ on the unit circle, see \cite[equation (3.13)]{refWVAbk}. The orthogonal polynomials for this weight are related to unitary random matrices \cite{refPS}. 

\end{enumerate}
\end{remarks}

\subsection{\label{ssec:adPn}Second discrete equation}
Consider the \bt\ $\mathcal{R}_{2}=\T_{-1,-1,-1}\circ\T_{-1,1,-1}$ which has the inverse $\mathcal{R}_{2}^{-1}=\T_{1,1,1}\circ\T_{1,-1,1}$. {Then we have
\begin{align*}
\T_{-1,1,-1}(w) &=1+\frac{2zw}{zw'-aw^2+(a+b-z)w-b},\nonumber \\
\T_{-1,-1,-1}\circ\T_{-1,1,-1}(w) 
&=\frac{1}{w}\left\{1+\frac{(a-b-c-1)(w-1)^2}{\ds z w'-aw^2+(2a-c-1+z)w-(a-c-1)}\right\},\\
\T_{1,-1,1}(w) &=1-\frac{2zw}{zw'+aw^2-(a+b-z)w+b},\nonumber\\
\T_{1,1,1}\circ\T_{1,-1,1}(w) &=\frac{1}{w}\left\{1-\frac{(a-b-c+1)(w-1)^2}{\ds z w'+aw^2-(2a-c+1+z)w+(a-c+1)}\right\},
\end{align*}
with $'\equiv \rmd/\rmd z$, where we have used the fact that $w$ satisfies \PV\ \eqref{eq:pv} with 
$\pms{\hf a^2}{-\hf b^2}{c}$ to eliminate $w''$.
Also for the parameters $(a,b,c)$ we have
\begin{align*}
\T_{-1,1,-1}(a,b,c) &=\Big(\hf(a-b+c-1),-\hf(a-b-c-1),a+b\Big),\nonumber\\
\T_{-1,-1,-1}\circ\T_{-1,1,-1}(a,b,c) &= \Big(\hf\big(a+b-c-1\big),\hf\big(a+b+c+1\big),-a+b+1\Big),\\
\T_{1,-1,1}(a,b,c) &=\Big(\hf(a-b+c+1),-\hf(a-b-c+1),a+b\Big),\nonumber\\
\T_{1,1,1}\circ\T_{1,-1,1}(a,b,c) &=\Big(\hf\big(a+b-c+1\big),\hf\big(a+b+c-1\big),-a+b-1\Big).
\end{align*}
Consequently, if $w_{n}=w\big(z;\hf a_{n}^2,-\hf b_{n}^2,c_{n}\big)$ and $w_{n\pm1}=w\big(z;\hf a_{n\pm1}^2,-\hf b_{n\pm1}^2,c_{n\pm1}\big)$ then}
\begin{subequations}\label{bt:R2}\begin{align} 
w_{n+1}=\mathcal{R}_{2}(w_{n})&=\frac{1}{w_{n}}\left\{1+\frac{(a_{n}-b_{n}-c_{n}-1)(w_{n}-1)^2}{\ds z w_{n}'-a_{n}w_{n}^2+(2a_{n}-c_{n}-1+z)w_{n}-(a_{n}-c_{n}-1)}\right\},\label{bt:R2a}\\
w_{n-1}=\mathcal{R}_{2}^{-1}(w_{n})&=\frac{1}{w_{n}}\left\{1-\frac{(a_{n}-b_{n}-c_{n}+1)(w_{n}-1)^2}{\ds z w_{n}'+a_{n}w_{n}^2-(2a_{n}-c_{n}+1+z)w_{n}+(a_{n}-c_{n}+1)}\right\},\label{bt:R2b}
\end{align}\end{subequations}
where $a_{n}$, $b_{n}$ and $c_{n}$ satisfy 
\begin{align*} 
(a_{n+1},b_{n+1},c_{n+1})&=\mathcal{R}_{2}(a_{n},b_{n},c_{n}) 
=\Big(\hf\big(a_{n}+b_{n}-c_{n}-1\big),\hf\big(a_{n}+b_{n}+c_{n}+1\big),-a_{n}+b_{n}+1\Big),\\
(a_{n-1},b_{n-1},c_{n-1})&=\mathcal{R}_{2}^{-1}(a_{n},b_{n},c_{n}) 
=\Big(\hf\big(a_{n}+b_{n}-c_{n}+1\big),\hf\big(a_{n}+b_{n}+c_{n}-1\big),-a_{n}+b_{n}-1\Big),
\end{align*}
which have solution
\begin{align} \label{eq:abcn2}a_{n}&=-\hf\big[n+\la+\rho+(-1)^n\ph\big],\qquad
b_{n}=\hf\big[n+\la-\rho+(-1)^n\ph\big],\qquad
c_{n}=n+\la-(-1)^n\ph,
\end{align}
with $\la$, $\rho$ and $\ph$ arbitrary constants.
From \eqref{bt:R2} we obtain
\begin{subequations}\label{dPn:btsys}\begin{align}
\frac{a_{n}-b_{n}-c_{n}-1}{w_{n} w_{n+1}-1}&-\frac{z}{(w_{n}-1)^2}\deriv{w_{n}}{z}+a_{n}-\frac{z-c_{n}-1}{w_{n}-1}-\frac{z}{\left(w_{n}-1\right)^{2}}=0,\\
\frac{a_{n}-b_{n}-c_{n}+1}{w_{n} w_{n-1}-1}&+\frac{z}{(w_{n}-1)^2}\deriv{w_{n}}{z}+a_{n}-\frac{z-c_{n}+1}{w_{n}-1}-\frac{z}{\left(w_{n}-1\right)^{2}}=0,
\end{align}\end{subequations}
Hence adding
and using \eqref{eq:abcn2} 
gives the difference equation
\beq
\frac{2n+2\la+1}{w_{n}w_{n+1}-1}+\frac{2n+2\la-1}{w_{n}w_{n-1}-1}+n+\la+\rho+(-1)^n\ph-\frac{2\big[n+\la-(-1)^n\ph\big]}{w_{n}-1}+\frac{2z w_{n}}{(w_{n}-1)^2}=0,
\eeq
which is equivalent to equation (2.8) in \cite{refROG}.

Setting $w_{n}=(q_{n}+1)/(q_{n}-1)$ and $w_{n\pm1}=(q_{n\pm1}+1)/(q_{n\pm1}-1)$ in \eqref{bt:R2}
gives
\begin{subequations}\label{dPn:btsys2}\begin{align} 
\frac{2 n+2\la+1}{q_{n+1}+q_{n}}&-\frac{z}{1-q_{n}^2}\deriv{q_{n}}{z}-\hfz+\frac{\big[n+\la+(-1)^{n}\ph \big] q_{n}+\rho}{1-q_{n}^{2}}=0, \\
\frac{2 n+2\la-1}{q_{n}+q_{n-1}}&+\frac{z}{1-q_{n}^2}\deriv{q_{n}}{z}-\hfz+\frac{\big[n+\la+(-1)^{n}\ph \big] q_{n}+\rho}{1-q_{n}^{2}}=0.
\end{align}\end{subequations}
Hence by adding we obtain
\beq\frac{2 n+2\la+1}{q_{n+1}+q_{n}}+\frac{2 n+2\la-1}{q_{n}+q_{n-1}}=z-\frac{2\big[n+\la+(-1)^{n}\ph \big] q_{n}+2\rho}{1-q_{n}^{2}},
\label{eq:dPn}\eeq
and then letting $\xi_{n}=n+\la$ gives
\beq \frac{\xi_{n+1}+\xi_{n}}{q_{n+1}+q_{n}}+\frac{\xi_{n}+\xi_{n-1}}{q_{n}+q_{n-1}}=z-\frac{2\big[\xi_{n}+(-1)^{n}\ph \big] q_{n}+2\rho}{1-q_{n}^{2}}.\label{eq:dPn1}\eeq
This is equation (2.9) in \cite{refROG}, equation (3.25) in \cite{refTGR02} and equation (13) in \cite{refGRTTS}.
Following \cite{refROG}, we introduce a new variable $p_{n}$ through
\[ p_{n}=1-\frac{2}{z}\,\frac{\xi_{n+1}+\xi_{n}}{q_{n+1}+q_{n}},\]
and then \eqref{eq:dPn1} becomes
\begin{align}\label{eq:dPn2}
p_{n}+p_{n-1}&=\frac{4}{z}\,\frac{[\xi_{n}+(-1)^n\ph]q_{n}+\rho}{1-q_{n}^2},\qquad
q_{n+1}+q_{n}=\frac{2}{z}\,\frac{\xi_{n+1}+\xi_{n}}{1-p_{n}},
\end{align}
which is system (2.10) in \cite{refROG}. 
Ramani, Ohta and Grammaticos \cite{refROG} state that ``the system \eqref{eq:dPn2} describes a motion in the same parameter space as asymmetric \dPII\ which is associated with the $A_3$ afﬁne Weyl group. While asymmetric \dPII\ describes a motion along the axes, the geometry of \eqref{eq:dPn2} is more complicated: the associated motion is a `staircase’ one".
To our knowledge\footnote{This was confirmed in an email communication with Basil Grammaticos.}, there is no discussion in the literature as to which (continuous) \peq\ arises in the continuum limit for either \eqref{eq:dPn1} or \eqref{eq:dPn2}.

In order to consider the alternating term in \eqref{eq:dPn}, i.e.\ the binary symmetry, 
setting $q_{2n}=x_{n}$ and $q_{2n+1}=y_{n}$ yields the discrete system
\begin{subequations}\label {sys:dPn}
\begin{align} 
\frac{4n+2\la+3}{x_{n+1}+y_{n}}+\frac{4n+2\la+1}{x_{n}+y_{n}}&=z-\frac{2(2n+\la-\ph+1)y_{n}+2\rho}{1-y_{n}^2},\\
\frac{4n+2\la+1}{x_{n}+y_{n}}+\frac{4n+2\la-1}{x_{n}+y_{n-1}}&=z-\frac{2(2n+\la+\ph)x_{n}+2\rho}{1-x_{n}^2},
\end{align}\end{subequations}
and so if $\xi_{n}=n+\la$ then
\begin{align*} 
\frac{\xi_{2n+2}+\xi_{2n+1}}{x_{n+1}+y_{n}}+\frac{\xi_{2n+1}+\xi_{2n}}{x_{n}+y_{n}}&=z-\frac{2(\xi_{2n+1}-\ph)y_{n}+2\rho}{1-y_{n}^2},\\
\frac{\xi_{2n+1}+\xi_{2n}}{x_{n}+y_{n}}+\frac{\xi_{2n}+\xi_{2n-1}}{x_{n}+y_{n-1}}&=z-\frac{2(\xi_{2n}+\ph)x_{n}+2\rho}{1-x_{n}^2}.
\end{align*} 

{Further $q_{n}$ satisfies the second-order equation
\beq \deriv[2]{q_{n}}{z}=\frac{q_{n}}{q_{n}^2-1}\left(\deriv{q_{n}}{z}\right)^{\!2}-\frac{1}{z}\deriv{q_{n}}{z}-\frac{a_{n}^2(q_{n}+1)^2-b_{n}^2(q_{n}-1)^2}{z^2(q_{n}^2-1)}
-\frac{c_{n}(q_{n}^2-1)}{2z}+\tfrac{1}{4}q_{n}(q_{n}^2-1), \label{eq2:qn2}
\eeq
with $a_{n}$, $b_{n}$ and $c_{n}$ given by \eqref{eq:abcn2}.}

We note that
\[\mathcal{R}_{2}^2(a,b,c)=(a-1,b+1,c+2),\] and so \[\mathcal{R}_{2}^{2n}(a,b,c)=(a-n,b+n,c+2n).\]
Therefore since
\begin{align*}
(a_{0},b_{0},c_{0})&=\big(-\hf(\la+\rho+\ph),\hf(\la-\rho+\ph),\la-\ph\big),\\
(a_{1},b_{1},c_{1})&=\big(-\hf(\la+\rho-\ph+1),\hf(\la-\rho-\ph+1), \la+\ph+1\big).
\end{align*}
then
\begin{align*}
(a_{2n},b_{2n},c_{2n})&=\big(-\hf(2n+\la+\rho+\ph),\hf(2n+\la-\rho+\ph),2n+\la-\ph\big),\\
(a_{2n+1},b_{2n+1},c_{2n+1})&=\big(-\hf(2n+\la+\rho-\ph+1),\hf(2n+\la-\rho-\ph+1),2n+\la+\ph+1\big).
\end{align*}
If $w_{0}(z)$ satisfies \PV\ \eqref{eq:pv} with $\a=\tfrac{1}{8}(\la+\rho+\ph)^2$, $\b=-\tfrac{1}{8}(\la-\rho+\ph)^2$ and $\ga=\la-\ph$ then
\beq x_{n}=q_{2n}=\frac{\mathcal{R}_{2}^{2n}(w_{0})+1}{\mathcal{R}_{2}^{2n}(w_{0})-1},\qquad 
y_{n}=q_{2n+1}=\frac{\mathcal{R}_{2}^{2n+1}(w_{0})+1}{\mathcal{R}_{2}^{2n+1}(w_{0})-1},\eeq
satisfy \eqref {sys:dPn}.

{\begin{remark}{\rm Since the \bt\ $\mathcal{R}_{2}=\T_{-1,-1,-1}\circ\T_{-1,1,-1}$ involves a sequence of two \bts\ then it might be expected that the resulting transformation is quadratic in $\rmd w/\rmd z$. Generally if one applies the \bt\ $\T_{\ep_{1},\ep_{2},\ep_{3}}$ followed by the \bt\ $\T_{\sigma_{1},\sigma_{2},\sigma_{3}}$ then the resulting transformation is quadratic in $\rmd w/\rmd z$, but in some cases there is a simplification, which happens in this case.
}\end{remark}}

\subsection{\label{ssec:tdPI}Third discrete equation: ternary discrete \PI}
Consider the \bt\ $\mathcal{R}_{3}=\T_{-1,1,-1}\circ \T_{1,-1,-1}\circ\T_{1,-1,1}$, which has an inverse $\mathcal{R}_{3}^{-1}=\T_{-1,1,1}$, {then
\begin{align*}
\T_{1,-1,1}(w) &=1-\frac{2zw}{zw'+aw^2-(a+b-z)w+b}, \\
\T_{1,-1,-1}\circ\T_{-1,1,1}(w) &=\frac{1}{w}\left\{1-\frac{(a-b+c+1)(w-1)^2}{zw'+aw^2-(2a+c-z+1)w+b} \right\}, \\
\T_{-1,1,-1}\circ \T_{1,-1,-1}\circ\T_{1,-1,1}(w) &= 1-\frac{2zw}{zw'+aw^2-(a+b-z)w+b},\\
\T_{-1,1,1}(w) &=1+\frac{2zw}{zw'-aw^2+(a-b-z)w+b},
\end{align*}
with $'\equiv \rmd/\rmd z$, where we have used the fact that $w$ satisfies \PV\ \eqref{eq:pv} with $\pms{\hf a^2}{-\hf b^2}{c}$ to eliminate $w''$.
Also for the parameters $(a,b,c)$ we have
\begin{align*}
\T_{1,-1,1}(a,b,c) &=\big(\hf(a-b+c+1),-\hf(a-b-c+1),a+b \big), \\
\T_{1,-1,-1}\circ\T_{-1,1,1}(a,b,c) &= \big(\hf(a+b+c+1),\hf(a+b-c-1),a-b+1 \big), \\
\T_{-1,1,-1}\circ \T_{1,-1,-1}\circ\T_{1,-1,1}(a,b,c) &= \big(\hf(a-b+c+1),\hf(a-b-c+1),a+b \big), \\
\T_{-1,1,1}(a,b,c) &=\big(\hf(a+b+c-1),-\hf(a+b-c-1),a-b \big),
\end{align*}
Consequently, if $w_{n}=w\big(z;\hf a_{n}^2,-\hf b_{n}^2,c_{n}\big)$ and $w_{n\pm1}=w\big(z;\hf a_{n\pm1}^2,-\hf b_{n\pm1}^2,c_{n\pm1}\big)$ then}
\begin{subequations}\label{bt:R3}\begin{align} 
w_{n+1}&=\mathcal{R}_{3}(w_{n})=1-\frac{2z w_{n}}{\ds z w_{n}'+a_{n}w_{n}^2-(a_{n}+b_{n}-z)w_{n}+b_{n}},\label{bt:R3a}\\
w_{n-1}&=\mathcal{R}_{3}^{-1}(w_{n})=1+\frac{2z w_{n}}{\ds z w_{n}'-a_{n}w_{n}^2+(a_{n}-b_{n}-z)w_{n}+b_{n}},\label{bt:R3b}
\end{align}\end{subequations}
with $'\equiv \rmd/\rmd z$,
where $a_{n}$, $b_{n}$ and $c_{n}$ satisfy 
\begin{subequations}\label{bt:R3abc}\begin{align}
(a_{n+1},b_{n+1},c_{n+1})&=\mathcal{R}_{3}(a_{n},b_{n},c_{n})=\left(\hf\big(a_{n}-b_{n}+c_{n}+1\big),\hf\big(a_{n}-b_{n}-c_{n}+1\big), a_{n}+b_{n}\right),\\
(a_{n-1},b_{n-1},c_{n-1})&=\mathcal{R}_{3}^{-1}(a_{n},b_{n},c_{n})=\left(\hf\big(a_{n}+b_{n}+c_{n}-1\big),-\hf\big(a_{n}+b_{n}-c_{n}-1\big), a_{n}-b_{n}\right).
\end{align}\end{subequations}
These have solution
\begin{subequations}\label{eq:abcn3}
\begin{align} 
a_{n}&=\tfrac{1}{3}n+\rho+\la\cos\big(\tfrac{2}{3}\pi n\big)+\tfrac{1}{3}\sqrt{3}\,\ph\sin\big(\tfrac{2}{3}\pi n\big),\label{eq:abcn3a}\\ 
b_{n}&=\tfrac{1}{3}+\sqrt{3}\,\la\sin\big(\tfrac{2}{3}\pi n\big)-\ph\cos\big(\tfrac{2}{3}\pi n\big),\label{eq:abcn3b}\\
c_{n}&=\tfrac{1}{3}n+\rho-2\la\cos\big(\tfrac{2}{3}\pi n\big)-\tfrac{2}{3}\sqrt{3}\,\ph\sin\big(\tfrac{2}{3}\pi n\big),\label{eq:abcn3c}
\end{align}\end{subequations}
with $\la$, $\rho$ and $\ph$ arbitrary constants (the factor of $\sqrt{3}$ is for convenience), i.e. 
\[\begin{split}a_{3n}&=n+\la+\rho,\\ a_{3n+1}&=n-\hf\la+\hf\ph+\rho+\tfrac{1}{3},\\ a_{3n+2}&=n-\hf\la-\hf\ph+\rho+\tfrac{2}{3},\end{split}\qquad
\begin{split}b_{3n}&=-\ph+\tfrac{1}{3},\\ b_{3n+1}&=\tfrac{3}{2}\la+\hf\ph+\tfrac{1}{3},\\ b_{3n+2}&=-\tfrac{3}{2}\la+\hf\ph+\tfrac{1}{3},\end{split}\qquad
\begin{split}c_{3n}&=n-2\la+\rho,\\ c_{3n+1}&=n+\la-\ph+\rho+\tfrac{1}{3},\\ c_{3n+2}&=n+\la+\ph+\rho+\tfrac{2}{3}.\end{split}\]
which have a ternary symmetry in contrast to the binary symmetry for the parameters discussed in \S\ref{ssec:adPII} and \S\ref{ssec:adPn} above.

Setting $w_{n}=1+1/x_{n}$ and $w_{n\pm1}=1+1/x_{n\pm1}$ in \eqref{bt:R3} gives 
\begin{subequations}\label{tdP1:btsys}\begin{align} 
x_{n+1}&-\frac{1}{2x_{n}(x_{n}+1)}\deriv{x_{n}}{z}+\frac{a_{n}}{2z x_{n}}-\frac{b_{n}}{2z (x_{n}+1)}+\frac{1}{2}=0,\\
x_{n-1}&+\frac{1}{2x_{n}(x_{n}+1)}\deriv{x_{n}}{z}+\frac{a_{n}}{2z x_{n}}+\frac{b_{n}}{2z (x_{n}+1)}+\frac{1}{2}=0.
\end{align}\end{subequations}
Hence by adding we obtain 
\beq x_{n}(x_{n+1}+x_{n-1}+1)+\frac{a_{n}}{z} 
=0,\label{eq:tdPI}\eeq
which is known as the \textit{ternary \dPI\ equation} \cite{refGRP97,refTGR02}. We remark that this derivation of \eqref{eq:tdPI} was given by two of us (Clarkson and Mitchell, in collaboration with Dzhamay and Hone), see \cite{refCDHM}.

Since equation \eqref{eq:tdPI} has a ternary symmetry, then it can be written as a system of three equations.
Setting $x_{3n}=X_{n}$, $x_{3n+1}=Y_{n}$ and $x_{3n+2}=Z_{n}$ in \eqref{eq:tdPI} yields the discrete system
\begin{subequations}\label{sys:tertdPI}\begin{align} 
X_{n}(Y_{n}+Z_{n-1}+1)+\frac{a_{3n}}{z}&=0,\\
Y_{n}(Z_{n}+X_{n}+1)+\frac{a_{3n+1}}{z}&=0,\\
Z_{n}(X_{n+1}+Y_{n}+1)+\frac{a_{3n+2}}{z}&=0.
\end{align}\end{subequations}

Further $x_{n}$ satisfies the second-order equation
\beq\deriv[2]{x_{n}}{z}=
\frac{2x_{n}+1}{2x_{n}(x_{n}+1)}\left(\deriv{x_{n}}{z}\right)^{\!2}
-\frac{1}{z}\deriv{x_{n}}{z}-\frac{a_n^2(x_{n}+1)^2-b_n^2x_{n}^2}{2z^2x_{n}(x_{n}+1)} 
-\frac{c_nx_{n}(x_{n}+1)}{z}+\frac{x_{n}(x_{n}+1)(2x_{n}+1)}{2},
\label{eq2:xn}\eeq
with $a_{n}$, $b_{n}$ and $c_{n}$ given by \eqref{eq:abcn3}.

We note that
\[\mathcal{R}_{3}^3(a,b,c)=(a+1,b,c+1),\] and so \[\mathcal{R}_{3}^{3n}(a,b,c)=(a+n,b,c+n).\]
Therefore since
\begin{align*} (a_{0},b_{0},c_{0}) &=\big(\la+\rho,-\ph+\tfrac{1}{3},-2\la+\rho\big),\\
(a_{1},b_{1},c_{1}) &=\big(-\hf\la+\hf\ph+\rho+\tfrac{1}{3},\tfrac{3}{2}\la+\hf\ph+\tfrac{1}{3},\la-\ph+\rho+\tfrac{1}{3}\big), \\
(a_{2},b_{2},c_{2}) &=\big(-\hf\la-\hf\ph+\rho+\tfrac{2}{3},-\tfrac{3}{2}\la+\hf\ph+\tfrac{1}{3},\la+\ph+\rho+\tfrac{2}{3}\big),
\end{align*}
then
\begin{align*} (a_{3n},b_{3n},c_{3n}) &=\big(n+\la+\rho,-\ph+\tfrac{1}{3},n-2\la+\rho\big),\\
(a_{3n+1},b_{3n+1},c_{3n+1}) &=\big(n-\hf\la+\hf\ph+\rho+\tfrac{1}{3},\tfrac{3}{2}\la+\hf\ph+\tfrac{1}{3},n+\la-\ph+\rho+\tfrac{1}{3}\big), \\
(a_{3n+2},b_{3n+2},c_{3n+2}) &=\big(n-\hf\la-\hf\ph+\rho+\tfrac{2}{3},-\tfrac{3}{2}\la+\hf\ph+\tfrac{1}{3},n+\la+\ph+\rho+\tfrac{2}{3}\big).
\end{align*}
If $w_{0}$ satisfies \PV\ \eqref{eq:pv} with $\a=\hf(\la+\rho)^2$, $\b=-\hf(\ph-\tfrac{1}{3})^2$ and $\ga=-2\la+\rho$ then
\[ X_{n}=x_{3n}=\frac{1}{\mathcal{R}_{3}^{3n}(w_{0})-1},\qquad
Y_{n}=x_{3n+1}=\frac{1}{\mathcal{R}_{3}^{3n+1}(w_{0})-1},\qquad
Z_{n}=x_{3n+2}=\frac{1}{\mathcal{R}_{3}^{3n+2}(w_{0})-1},
 \] 
 satisfy \eqref{sys:tertdPI}.
 
\begin{remarks}{\rm
\begin{enumerate}[(i)]\item[]
\item In terms of $w$, the transformation \eqref{bt:R3a} is the same as $\T_{1,-1,1}(w)$. 
However, the parameters $(a,b,c)$ map differently since
\[\T_{1,-1,1}(a,b,c)=\big(\hf(a-b+c-1),\hf(-a+b+c-1), a+b\big). \] 
It is straightforward to show that 
$\T_{1,-1,1}\circ\T_{-1,1,1}(w)=w$, whereas $\T_{-1,1,1}\circ\T_{1,-1,1}(w)\not=w$, 
so $\T_{1,-1,1}$ is \textit{not} the inverse of $\T_{-1,1,1}$.

\item It is interesting to note that given a solution $w_{0}(z)$ of \PV\ \eqref{eq:pv} with $\pmsn{0}{\hf a_{0}^2}{-\hf b_{0}^2}{c_{0}}$, the transformation $\mathcal{R}_{1}$ \eqref{bt:R1a} maps it to the solution $w_{1}(z)$ with $\pmsn{1}{\hf a_{1}^2}{-\hf b_{1}^2}{c_{1}}$ and the transformation $\mathcal{R}_{3}$ \eqref{bt:R3a} maps it to the solution $w_{3}(z)$ with $\pmsn{3}{\hf a_{3}^2}{-\hf b_{3}^2}{c_{3}}$ where
\[ w_{1}(z)=w_{3}(z),\qquad a_{1}=a_{3},\qquad b_{1}=-b_{3},\qquad c_{1}=c_{3},\]
i.e.\ both transformations map $w_{0}(z)$ to the \textit{same} solution of \PV. However applying the transformation a second time gives different solutions since
\begin{align*}\mathcal{R}_{1}^2(a_{0},b_{0},c_{0})&=(a_{0}+1,b_{0}-1, c_{0}),\\ 
\mathcal{R}_{3}^2(a_{0},b_{0},c_{0})&=\big(\hf(a_{0}+b_{0}+c_{0}+1),-\hf(a_{0}+b_{0}-c_{0}-1), a_{0}-b_{0}+1\big),\end{align*}
so the transformations $\mathcal{R}_{1}$ \eqref{bt:R1} and $\mathcal{R}_{3}$ \eqref{bt:R3} will generate \textit{different} hierarchies of solutions.

\item The symmetric case of equation \eqref{eq:tdPI}, i.e.\ when $\la=\ph=0$, 
 \beq x_{n}(x_{n+1}+x_{n-1}+1)+\frac{\rho+\tfrac{1}{3}n}{z}=0,\label{eq:tdPI0}\eeq
is a discrete \PI\ \cite[equation (3.6)]{refRG96}, whilst the full equation, which was derived from a discrete dressing transformation \cite{refGRP97}, is a discrete form of \PIV\ \eqref{eq:piv}.

\item The ternary \dPI\ \eqref{eq:tdPI} is associated with \PV\ \eqref{eq:pv}, though the asymmetric generalisation of \dPI\ \eqref{eq:dPI}, i.e.
\beq\label{eq:adPI} x_{n+1}+x_{n}+x_{n-1}=1+\frac{\la n+\rho+(-1)^n\ph}{x_{n}}, \eeq 
with $\la$, $\rho$ and $\ph$ constants, is associated with \PIV\ \eqref{eq:piv} cf.~\cite{refFIK91,refMagnus95}, and can be derived from \bk\ transformations of \PIV\ 
\cite{refCMW,refFGR}. We remark that the continuum limit of \eqref{eq:adPI} with $\ph\not=0$ is \PII\ \eqref{eq:p2} rather than \PI\ \eqref{eq:p1}, as shown in \cite{refCJ99}.

\item Recently Clarkson \etal\ \cite{refCDHM} studied the special case of the ternary \dPI\ equation given by 
\beq v_{n}(v_{n+1}+v_{n-1}+1)=(n+1)\ep,\qquad \ep>0, \label{tdpI} \eeq
with $v_{-1}=0$, which arises in quantum minimal surfaces \cite{refAHK,refHoppe} and found a unique solution. The expressions for $v_{3n}$, $v_{3n+1}$ and $v_{3n+2}$ were written explicitly as ratios of Wronskian determinants in terms of the modified Bessel functions 
$K_{1/6}(z)$ and $K_{5/6}(z)$ and were all different showing the three-fold structure of the solutions.
\item Felder and Hoppe \cite{refFH} show that ternary \dPI\ equation \eqref{tdpI} also arises in connection with orthogonal polynomials in the complex plane with respect to the weight
\[w(z;t)=\exp\left\{-t|z|^2+\tfrac{1}{3}\rmi\big(z^3+\overline{z}^3\big)\right\},\qquad t>0,\]
with $z=x+\rmi y\in\Com$, and has the \textit{same} unique solution as in \cite{refCDHM}.
Teodorescu \etal\ \cite{refTBAZW}, see also \cite[equation (26)]{refTeo}, show that 
equation \eqref{tdpI} arises in the theory of random normal matrices \cite{refWZ}, which motivated the study in \cite{refFH}.
\item A study of the asymptotics of solutions of \eqref{tdpI} is given by Aptekarev and Novokshenov \cite{refAN}.
\end{enumerate}}\end{remarks}

\subsection{\label{ssec:tdPI2}Fourth discrete equation: new equation with ternary symmetry}
It is possible to derive a discrete equation relating the three solutions $x_{n}$, $x_{n+2}$ and $x_{n-2}$ of the ternary \dPI\ equation \eqref{eq:tdPI} as illustrated in the following Lemma.
\begin{lemma} Suppose that $x_{n}$, $x_{n+2}$ and $x_{n-2}$ are solutions of the ternary \dPI\ equation \eqref{eq:tdPI},
then they satisfy the discrete equation
\beq \frac{a_{n+1}}{x_{n+2}+x_{n}+1}+\frac{a_{n-1}}{x_{n}+x_{n-2}+1}=z+\frac{a_{n}}{x_{n}}, \label{eq:tdPInew}\eeq
where $a_{n}$ is given by \eqref{eq:abcn3a}, i.e.
\beq a_{n}=\tfrac{1}{3}n+\rho+\la\cos\big(\tfrac{2}{3}\pi n\big)+\tfrac{1}{3}\sqrt{3}\,\ph\sin\big(\tfrac{2}{3}\pi n\big),\label{def:an}\eeq
with $\la$, $\rho$ and $\ph$ arbitrary constants.
\end{lemma}
\begin{proof} Letting $n\to n-1$ in equation \eqref{eq:tdPI} gives 
\[x_{n-1}(x_{n}+x_{n-2}+1)+\frac{a_{n-1}}{z}=0.\]
Solving this equation and \eqref{eq:tdPI} for $x_{n+1}$ and $x_{n-1}$ gives
\[ x_{n+1}=\frac{a_{n-1}}{z(x_{n}+x_{n-2}+1)}-1-\frac{a_{n}}{zx_{n}},\qquad x_{n-1}=-\frac{a_{n-1}}{z(x_{n}+x_{n-2}+1)}.\]
Letting $n\to n+2$ in 
the expression for $x_{n-1}$ gives
\[ x_{n+1}=-\frac{a_{n+1}}{z(x_{n+2}+x_{n}+1)}.\]
Equating the two expressions for $x_{n+1}$ gives the result.
\end{proof}

An alternative proof is to consider the \bt\ $\mathcal{R}_{4}=\mathcal{R}_{3}^2$, with $\mathcal{R}_{4}^{-1}=\T_{-1,1,1}^2$, then
\begin{subequations}\label{bt:R32}\begin{align} 
w_{n+2}=\mathcal{R}_{4}(w_{n})&=\frac{1}{w_{n}}\left\{1+\frac{(a_{n}-b_{n}+c_{n}+1)(w_{n}-1)^2}{\ds zw_{n}'+a_{n}w_{n}^2-(2a_{n}+c_{n}+1-z)w_{n}+a_{n}+c_{n}+1}\right\},\\
w_{n-2}=\mathcal{R}_{4}^{-1}(w_{n})&=\frac{1}{w_{n}}\left\{1+\frac{(a_{n}+b_{n}+c_{n}-1)(w_{n}-1)^2}{\ds zw_{n}'-a_{n}w_{n}^2+(2a_{n}+c_{n}-1-z)w_{n}-a_{n}-c_{n}+1}\right\},
\end{align}\end{subequations}
with $'\equiv \rmd/\rmd z$ and $a_{n}$, $b_{n}$ and $c_{n}$ are given by \eqref{eq:abcn3}. Letting $w_{n}=1+1/x_{n}$ and $w_{n\pm2}=1+1/x_{n\pm2}$ gives 
\begin{subequations}\label{bt:R32xn}\begin{align} 
\frac{a_{n}-b_{n}+c_{n}+1}{x_{n+2}+x_{n}+1}&=z-\frac{z}{x_{n}(x_{n}+1)}\deriv{x_{n}}{z}+\frac{(a_n-b_n)x_{n}+a_{n}}{x_{n}(x_{n}+1)},\\
\frac{a_{n}+b_{n}+c_{n}-1}{x_{n}+x_{n-2}+1}&=z+\frac{z}{x_{n}(x_{n}+1)}\deriv{x_{n}}{z}+\frac{(a_n+b_n)x_{n}+a_{n}}{x_{n}(x_{n}+1)},
\end{align}\end{subequations}
and so by adding we obtain
\[\frac{a_{n+1}}{x_{n+2}+x_{n}+1}+\frac{a_{n-1}}{x_{n}+x_{n-2}+1}=z+\frac{a_{n}}{x_{n}},\]
since from \eqref{bt:R3abc}, $a_{n+1}=\hf(a_{n}-b_{n}+c_{n}+1)$ and $a_{n-1}=\hf(a_{n}+b_{n}+c_{n}-1)$.
The parameters $\wt{a}_{n}$, $\wt{b}_{n}$ and $\wt{c}_{n}$ satisfy
\[\begin{split} 
\wt{a}_{n+1}&=\hf(\wt{a}_{n}+\wt{b}_{n}+\wt{c}_{n}+1),\\
\wt{b}_{n+1}&=-\hf(\wt{a}_{n}+\wt{b}_{n}-\wt{c}_{n}-1),\\
\wt{c}_{n+1}&=\wt{a}_{n}-\wt{b}_{n}+1,
\end{split} \qquad\quad
\begin{split} 
\wt{a}_{n-1}&=\hf(\wt{a}_{n}-\wt{b}_{n}+\wt{c}_{n}-1),\\
\wt{b}_{n-1}&=\hf(\wt{a}_{n}-\wt{b}_{n}-\wt{c}_{n}+1),\\
\wt{c}_{n-1}&=\wt{a}_{n}+\wt{b}_{n}-1,
\end{split}\]
which have solution
\begin{align*} 
\wt{a}_{n}&=\tfrac{2}{3}n+\rho+\la\cos\big(\tfrac{2}{3}\pi n\big)+\tfrac{1}{3}\sqrt{3}\,\ph\sin\big(\tfrac{2}{3}\pi n\big),\\
\wt{b}_{n}&=\tfrac{1}{3}-\sqrt{3}\,\la\sin\big(\tfrac{2}{3}\pi n\big)+\ph\cos\big(\tfrac{2}{3}\pi n\big),\\
\wt{c}_{n}&=\tfrac{2}{3}n+\rho-2\la\cos\big(\tfrac{2}{3}\pi n\big)-\tfrac{2}{3}\sqrt{3}\,\ph\sin\big(\tfrac{2}{3}\pi n\big),
\end{align*}
with $\la$, $\rho$ and $\ph$ arbitrary constants.
We remark that the \bt\ \eqref{bt:R32} has similar structure to the \bt\ $\mathcal{R}_{2}$ \eqref{bt:R2}, except the parameters have a ternary symmetry rather than a binary symmetry. We note that since $\mathcal{R}_{3}=\T_{-1,1,-1}\circ \T_{1,-1,-1}\circ\T_{1,-1,1}$, then
\begin{align*}
\mathcal{R}_{4}=\mathcal{R}_{3}^2
&=\T_{-1,1,-1}\circ \T_{1,-1,-1}\circ\T_{1,-1,1}\circ \T_{-1,1,-1}\circ \T_{1,-1,-1}\circ\T_{1,-1,1}\\
&=\T_{-1,1,-1}\circ \T_{1,-1,-1}\circ \T_{1,-1,-1}\circ\T_{1,-1,1},
\end{align*}
as $\T_{1,-1,1}\circ \T_{-1,1,-1}=\mathcal{I}$, the identity.
\begin{corollary}
If $x_{n}$ satisfies ternary \dPI\ equation \eqref{eq:tdPI} and $a_{n}$ is given by \eqref{def:an}, then $q_{n}=x_{2n}$ and $p_{n}=x_{2n+1}$ satisfy
\begin{subequations}\label{eq:tdPInew12}\begin{align} 
\frac{a_{2n+1}}{q_{n+1}+q_{n}+1}+\frac{a_{2n-1}}{q_{n}+q_{n-1}+1}&=z+\frac{a_{2n}}{q_{n}}, \label{eq:tdPInew1}\\
\frac{a_{2n+2}}{p_{n+1}+p_{n}+1}+\frac{a_{2n}}{p_{n}+p_{n-1}+1}&=z+\frac{a_{2n+1}}{p_{n}}, \label{eq:tdPInew2}
\end{align}\end{subequations}
respectively. Further $X_{n}=x_{3n}$, $Y_{n}=x_{3n+1}$ and $Z_{n}=x_{3n+2}$ satisfy
 the system
\begin{subequations}\label{eq:tdPInew3}\begin{align}
\frac{a_{3n+1}}{X_{n}+Z_{n}+1}+\frac{a_{3n-1}}{X_{n}+Y_{n-1}+1}
 &=z+\frac{a_{3n}}{X_{n}},\\
 \frac{a_{3n+2}}{X_{n+1}+Y_{n}+1}+\frac{a_{3n}}{Y_{n}+Z_{n-1}+1}
 &=z+\frac{a_{3n+1}}{Y_{n}},\\
 \frac{a_{3n+3}}{Y_{n+1}+Z_{n}+1}+\frac{a_{3n+1}}{X_{n}+Z_{n}+1}
 &=z+\frac{a_{3n+2}}{Z_{n}}.
\end{align}\end{subequations}
\end{corollary}
\begin{proof} Letting $n\to2n$ and $n\to2n+1$ in \eqref{eq:tdPInew} and then setting $q_{n}=x_{2n}$ and $p_{n}=x_{2n+1}$, respectively, gives \eqref{eq:tdPInew12}. Similarly for the system \eqref{eq:tdPInew3}, which is equivalent to the system \eqref{sys:tertdPI}. 
\end{proof}

\subsection{\label{ssec:diffdiff}Differential-difference equations}
In the previous subsections, we have obtained discrete equations by eliminating the derivative between two equations, i.e.\ by adding the equations in \eqref{dP2:btsys}, \eqref{dPn:btsys} and \eqref{tdP1:btsys}. 
If instead we subtract the equations then we obtain differential-difference equations.

\begin{enumerate}[(i)]
\item Subtracting the equations in \eqref{dP2:btsys2} gives the \textit{modified Volterra lattice} \cite{refWadati}
\[\deriv{q_{n}}{z}=\tfrac{1}{4}(q_{n}^2-1)(q_{n+1}-q_{n-1}).\]
which is a special case of equation \eqref{eq:V1} below. 
\item Subtracting the equations in \eqref{dPn:btsys2} gives
\[ z\deriv{q_{n}}{z}=\hf(1-q_{n}^2)\left(\frac{2 n+2\la+1}{q_{n+1}+q_{n}}-\frac{2 n+2\la-1}{q_{n}+q_{n-1}}\right).\]
which is a non-autonomous special case of equation \eqref{eq:V2} below.

\item Subtracting the equations in \eqref{tdP1:btsys} gives
\[ z\deriv{x_{n}}{z}=z x_{n}(x_{n}+1)(x_{n+1}-x_{n-1})-{b_{n} x_{n}},\]
with $b_{n}$ given by \eqref{eq:abcn3b}.

\item Subtracting the equations in \eqref{bt:R32xn} gives
\[ z\deriv{x_{n}}{z}= -x_{n}(x_{n}+1)\left(\frac{a_{n+1}}{x_{n+2}+x_{n}+1} -\frac{a_{n-1}}{x_{n}+x_{n-2}+1}\right)-b_{n} x_{n},\]
with $a_{n}$ and $b_{n}$ given by \eqref{eq:abcn3}. Consequently {setting $x_{2n}=\hf(q_{n}-1)$ and $x_{2n+1}=\hf(p_{n}-1)$ gives}
\begin{align*} 
z\deriv{q_{n}}{z}&= (1-q_{n}^2)\left(\frac{a_{2n+1}}{q_{n+1}+q_{n}} -\frac{a_{2n-1}}{q_{n}+q_{n-1}}\right)- b_{2n} (q_{n}-1),\\
z\deriv{p_{n}}{z}&= (1-p_{n}^2)\left(\frac{a_{2n+2}}{p_{n+1}+p_{n}} -\frac{a_{2n}}{p_{n}+p_{n-1}}\right)- b_{2n+1} (p_{n}-1),
\end{align*}
which are non-autonomous special cases of equation \eqref{eq:V2} below.
\end{enumerate}

\begin{remark}{\rm 
Yamilov \cite{refYam} gives a classification of (autonomous) Volterra-type equations. The first two are:
\begin{align}
\deriv{u_{n}}{z}&=P(u_{n})(u_{n+1}-u_{n-1}),\label{eq:V1}\\
\deriv{u_{n}}{z}&=P(u_{n}^2)\left(\frac{1}{u_{n+1}+u_{n}}-\frac{1}{u_{n}+u_{n-1}}\right),\label{eq:V2}
\end{align}
where $P(u)=a u^2+b u+c$, with $a$, $b$ and $c$ constants.
}\end{remark}

\section{\label{sec:PVrats}Rational solutions of \PV}
In this section we describe the rational solutions of \PV\ \eqref{eq:pv},
which are classified in the following theorem.
\begin{theorem}{\rm \label{thm:KLM}
\textit{Equation \eqref{eq:pv} has a rational solution if and only if one of the following holds:}
\begin{enumerate}[(i)]
\item\textit{$\a=\hf m^2$, $\b=-\hf(m+2n+1+\k)^2$, $\ga=\k$, for $m\geq1$;}
\item\textit{$\a=\hf(m+\k)^2$, $\b=-\hf(n+\ep\k)^2$, $\ga=m+\ep n$, with $\ep=\pm1$, provided that $m\not=0$ or $n\not=0$;}
\item\textit{$\a=\hf(m+\hf)^2$, $\b=-\hf(n+\hf)^2$, $\ga=\k$, provided that $m\not=0$ or $n\not=0$,}
\end{enumerate}
\textit{where $m,n\in\Integer$ and $\k$ is an arbitrary constant,
together with the solutions obtained through the symmetries 
\begin{align}
\label{PV:S1}\mathcal{S}_{1}:\qquad &{w_{1}}({z};{\a_{1}}, {\b_{1}}, {\ga_{1}},-\hf)
={w(-z;\a,\b,\ga,-\hf)},&&({\a_{1}}, {\b_{1}}, {\ga_{1}},-\hf)=(\a,\b,-\ga,-\hf),\\
\label{PV:S2}\mathcal{S}_{2}:\qquad &{w_{2}}({z};{\a_{2}}, {\b_{2}}, {\ga_{3}},-\hf)
=\frac{1}{w(z;\a,\b,\ga,-\hf)},&&({\a_{2}}, {\b_{2}}, {\ga_{3}},-\hf)=(-\b,-\a,-\ga,-\hf),
\end{align}
where $w(z;\a,\b,\ga,-\hf)$ is a solution of \PV\ \eqref{eq:pv}.}
}\end{theorem}

\begin{proof}See Kitaev, Law and McLeod \cite[Theorem 1.1]{refKLM}; also \cite[Theorem 40.3]{refGLS}.\end{proof}

\begin{remarks}{\rm 
\begin{enumerate}[(i)]\item[]
\item Kitaev, Law and McLeod \cite[Theorem 1.1]{refKLM} give four cases, though their cases (I) and (II) are related by the symmetry \eqref{PV:S2}. 
Further Kitaev, Law and McLeod \cite[Theorem 1.2]{refKLM} prove that if $\ga\not\in\Integer$ in case (i), then the rational solution is unique and if $\ga\in\Integer$ then there are at most two rational solutions, as illustrated in \cite{refCD24}; see also \cite{refAGLZ24a,refAGLZ24b}.
\item Rational solutions in case (i) are expressed in terms of \textit{generalised Laguerre polynomials}, which are discussed in \S\ref{ssec:GenLag}, and in cases (ii) and (iii) in terms of \textit{generalised Umemura polynomials}, discussed in \S\ref{ssec:GenUm}.
\end{enumerate}}\end{remarks}

\subsection{\label{def:GenLag}Generalised Laguerre polynomials}
Rational solutions in case (i) of Theorem\ \ref{thm:KLM} are expressed in terms of \textit{generalised Laguerre polynomials}, which are written in terms of a determinant of Laguerre polynomials, or equivalently as a Wronskian, see \cite[\S3]{refCD24}. These rational solutions are special cases of the solutions of \PV\ \eqref{eq:pv} expressible in terms of the Kummer functions $M(a,b,z)$ and $U(a,b,z)$ discussed in \cite{refMasuda04,refOkamotoPV} since
\beq U(-n,\a+1,z)=(-1)^n(\a+1)_{n}M(-n,\a+1,z)=(-1)^n n!\, \Lag{n}{\a}(z),\eeq
with $L_{n}^{(\a)}(z)$ the Laguerre polynomial 
\beq\label{eq:alp} \Lag{n}{\a}(z)=\frac{z^{-\a}\,\rme^z}{n!}\deriv[n]{}{z}\left(z^{n+\a}\,\rme^{-z}\right),\qquad n\geq0,\eeq see \cite[equation (13.6.19)]{refDLMF}.
First we define the generalised Laguerre polynomial $\Tmn{m,n}{\mu}(z)$; properties of it are discussed in \cite{refBB,refCD24}. 

\begin{definition}
The \textit{generalised Laguerre polynomial} $\Tmn{m,n}{\mu}(z)$, which is a polynomial of degree $(m+1)n$, is defined by 
\beq \label{def:Tmn}
\Tmn{m,n}{\mu}(z)=\det \left[\deriv[j+k]{}{z}L_{m+n}^{(\mu+1)} (z) \right]_{j,k=0}^{n-1}, \qquad m\ge 0, \quad n\geq 1,
\eeq
{with $\Tmn{m,0}{\mu}(z)=1$},
where $L_{n}^{(\mu)}(z)$ is the Laguerre polynomial \eqref{eq:alp}.
\end{definition}

For case (i) in Theorem\ \ref{thm:KLM}, Kitaev, Law and McLeod \cite{refKLM} state that rational solutions arise when
\[\a=\hf\ell^2,\qquad \b=-\hf(k+\ep\ga)^2,\qquad \ep=\pm1,\]
where $\ell>0$ and $k+\ell$ is odd, with $k,\ell\in\Integer$, and $\b\not=0$ when $|k|<\ell$. Since the sign of $\ga$ can be changed by letting $z\to-z$, we set $\ep=1$, without loss of generality. There are three cases (a), $k>\ell$, (b), $|k|\leq\ell$ and (c), $-k>\ell$, which are given in the following theorem.

\begin{theorem}{\label{thm:genLag}\rm
\textit{Suppose that $\Tmn{m,n}{\mu}(z)$ is the generalised Laguerre polynomial given by \eqref{def:Tmn} and $L_{n}^{(\a)}(z)$ is the Laguerre polynomial, then \PV\ \eqref{eq:pv} has the following three sets of rational solutions.}
\begin{enumerate}[(a)]
\item\textit{If $\ell=m$ and $k=m+2n+1$, with $m\geq1$, so $k+\ell=2m+2n+1$, then for $n\geq1$
\begin{subequations}\beq {w}_{m,n}^{(\mu)}(z)=\frac{T_{m-1,n}^{(\mu)}(z)\,T_{m-1,n+1}^{(\mu)}(z)}{\Tmn{m,n}{\mu}(z)\,T_{m-2,n+1}^{(\mu)}(z)}, 
\qquad\pms{\hf m^2}{-\hf(m+2n+1+\mu)^2}{\mu}.\eeq
{In the case when $n=0$, then for $m\geq1$
\beq {w}_{m,0}^{(\mu)}(z)
=\frac{\Lag{m}{\mu+1}(z)}{\Lag{m-1}{\mu+1}(z)}, 
\qquad\pms{\hf m^2}{-\hf(m+1+\mu)^2}{\mu}.\eeq
}\end{subequations}}
\item\textit{If $\ell=m+n+1$ and $k=m-n$, with $m\geq0$, so $k+\ell=2m+1$, then for $m\geq1$
\begin{subequations}
\beq \label{eqn:GenLagvmn}
{v}_{m,n}^{(\mu)}(z)=\frac{m-n+\mu}{m+n+1}\,\frac{\Tmn{m,n}{\mu-2n}(z)\,\Tmn{m-1,n+1}{\mu-2n-2}(z)}{\Tmn{m-1,n}{\mu-2n}(z)\,\Tmn{m,n+1}{\mu-2n-2}(z)},
\quad\pms{\hf(m+n+1)^2}{-\hf(m-n+\mu)^2}{\mu}.\eeq
{In the case when $m=0$ and $n\geq1$ then
\begin{align} \label{GenLag_vmn}
{v}_{0,n}^{(\mu)}(z)&
={\frac{n-\mu}{n+1}\,\frac{L_{n}^{(-\mu-1)}(-z)}{L_{n+1}^{(-\mu-1)}(-z)}}, &&\pms{\hf(n+1)^2}{-\hf(n-\mu)^2}{\mu},\\
\wt{v}_{0,n}^{(\mu)}(z)&
=\frac{L_{n+1}^{(-1-\mu)}(-z)}{L_n^{(-1-\mu)}(-z)}, && \pms{\hf(n+1)^2}{-\hf(n-\mu)^2}{\mu+2}.
\end{align}
In the case when $n=0$ and $m\geq0$ then
\begin{align} 
{v}_{m,0}^{(\mu)}(z)&
=\frac{m+\mu}{m+1}\,\frac{\LaguerreL{\mu-1}{m}(z)}{\LaguerreL{\mu-1}{m+1}(z)},
&&\pms{\hf(m+1)^2}{-\hf(m+\mu)^2}{\mu},\\
\wt{v}_{m,0}^{(\mu)}(z)&
=\frac{\LaguerreL{\mu-1}{m+1}(z)}{\LaguerreL{\mu-1}{m}(z)},
&&\pms{\hf(m+1)^2}{-\hf(m+\mu)^2}{\mu-2}.\end{align}}\end{subequations}}
\item\textit{If $\ell=m$ and $k=-m-2n-1$, with $m\geq1$ and $n\geq0$, so $k+\ell=-2m-1$, then ${u}_{m,n}^{(\mu)}(z)={w}_{m,n}^{(-\mu)}(-z)$, i.e.
\begin{subequations}\beq {u}_{m,n}^{(\mu)}(z)=
\frac{\Tmn{n-1,m}{\mu-2m-2n}(z)\,\Tmn{n,m}{\mu-2m-2n-2}(z)}{\Tmn{n,m-1}{\mu-2m-2n}(z)\,\Tmn{n-1,m+1}{\mu-2m-2n-2}(z)},
\qquad\pms{\hf m^2}{-\hf(m+2n+1-\mu)^2}{\mu}.\eeq
{In the case when $n=0$, then for $m\geq1$
\beq {u}_{m,0}^{(\mu)}(z)
=\frac{\Lag{m}{1-\mu}(-z)}{\Lag{m-1}{1-\mu}(-z)}, 
\qquad\pms{\hf m^2}{-\hf(m+1-\mu)^2}{\mu}.\eeq}\end{subequations}
}\end{enumerate}}\end{theorem}

\begin{proof}{For cases (a) and (c) see Clarkson and Dunning \cite[Theorem 3]{refCD24} and \cite[Corollary 1]{refCD24}, respectively. Case (b) follows in a similar way from the special function solutions of \PV\ \eqref{eq:pv} given by Masuda \cite[Theorem 2.2]{refMasuda04}}.
The solutions when either $m=0$ or $n=0$ follow from recurrence relations for the Laguerre polynomial, cf.~\cite[\S18.9]{refDLMF}.
\end{proof}

{It is known that rational solutions of \PIII\ can be expressed either in terms of four special polynomials or in terms of the logarithmic derivative of the ratio of two special polynomials, cf.~\cite[Theorem 2.4]{refPACpiii}. Hence it might be expected that the rational solutions discussed in Theorem \ref{thm:genLag} can also be written in terms of the logarithmic derivative of the ratio of two generalised Laguerre polynomials. We note that Masuda \cite[Remark 2.3]{refMasuda04} has written the special function solutions of \spv\ \eqref{eq:spv} as logarithmic derivatives.}

\begin{remarks}{\label{rem:45}\rm 
\begin{enumerate}[(i)]\item[]
\item The rational solutions in Theorem \ref{thm:genLag} can also be written as logarithmic derivatives
\begin{subequations}\label{GenLag_log}
\begin{align}
w_{m,n}^{(\mu)}(z)&= \frac{z}{m} \Ln{\Tmn{m-2,n+1}{\mu}(z)}{\Tmn{m,n}{\mu}(z)}-\frac{z-m-2n-1-\mu}{m},\label{GenLag_wmn_log}\\
v_{m,n}^{(\mu)}(z)&=\frac{z}{m+n+1} \Ln{\Tmn{m-1,n}{\mu-2n}(z)}{\Tmn{m,n+1}{\mu-2n-2}(z)}+1,\label{GenLag_vmn_log}\\
u_{m,n}^{(\mu)}(z)&= \frac{z}{m} \Ln{\Tmn{n,m-1}{\mu-2m-2n}}{\Tmn{n-1,m+1}{\mu-2m-2n-2}}+\frac{z+m+2n+1-\mu}{m}.\label{GenLag_umn_log}
\end{align} 
\end{subequations}
{Using computer algebra we have verified these alternative forms of the rational solutions in Theorem \ref{thm:genLag} given by
\eqref{GenLag_log} for several small values of $m$ and $n$. Comparing these expressions with those in Theorem \ref{thm:genLag} gives the relations \eqref{eq:A3}. We envisage these relations that can be proved using the Jacobi identity \cite{refDod}, sometimes known as the \textit{Lewis Carroll formula}, for the determinant $\mathcal{D}$
\beq \mathcal{D}\, \mathcal{D}\bn{i,k}{j,\ell}= \mathcal{D}\bn{i}{j}\mathcal{D}\bn{k}{\ell} - \mathcal{D}\bn{k}{j} \mathcal{D}\bn{i}{\ell}\label{JacobiId}\eeq
where $\mathcal{D}\bn{i}{j} $
is the determinant with the $i^{\rm th}$ row and the $j^{\rm th}$ column removed from $\mathcal{D}$, though we don't pursue this further here.}
\item {From the definition of $\Tmn{m,n}{\mu}(z)$ \eqref{def:Tmn}
\[ \Tmn{m-1,n}{\mu}(z)=\det \left[\deriv[j+k]{}{z}L_{m+n-1}^{(\mu+1)} (z) \right]_{j,k=0}^{n-1},\qquad\Tmn{m-2,n+1}{\mu}(z)=\det \left[\deriv[j+k]{}{z}L_{m+n-1}^{(\mu+1)} (z) \right]_{j,k=0}^{n},\] 
which are well-defined for $m\geq1$ and $n\geq1$.}
\end{enumerate}
}\end{remarks}

Recall that solutions of the discrete equations \eqref{eq:adPII} and \eqref{eq:dPn} also satisfy the second-order equation
\beq \deriv[2]{q}{z}=\frac{q}{q^2-1}\left(\deriv{q}{z}\right)^{\!2}-\frac{1}{z}\deriv{q}{z}-\frac{2\a(q+1)^2+2\b(q-1)^2}{z^2(q^2-1)}
-\frac{\ga(q^2-1)}{2z}+\tfrac{1}{4}q(q^2-1), \label{eq2:q}
\eeq {see equations \eqref{eq2:qn} and \eqref{eq2:qn2}},
and solutions of the ternary \dPI\ \eqref{eq:tdPI} also satisfy the second-order equation
\beq\deriv[2]{x}{z}=
\frac{2x+1}{2x(x+1)}\left(\deriv{x}{z}\right)^{\!2}
-\frac{1}{z}\deriv{x}{z}-\frac{\a(x+1)^2+\b x^2}{z^2x(x+1)} 
-\frac{\ga x(x+1)}{z}+\frac{x(x+1)(2x+1)}{2},\label{eq2:x} \eeq 
{see equation \eqref{eq2:xn}}.
In the following Lemma we derive the solutions of equations \eqref{eq2:q} and \eqref{eq2:x} in terms of generalised Laguerre polynomials.

\begin{lemma} If $w_{m,n}^{(\mu)}(z)$, $v_{m,n}^{(\mu)}(z)$ and $u_{m,n}^{(\mu)}(z)$ are rational solutions of \PV\ defined in Theorem \ref{thm:genLag} and $\Tmn{m,n}{\mu}(z)$ is the generalised Laguerre polynomial given by \eqref{def:Tmn}, then
\begin{subequations}\begin{align}
\mathcal{W}_{m,n}^{(\mu)}(z)=\frac{w_{m,n}^{(\mu)}(z)+1}{w_{m,n}^{(\mu)}(z)-1} &=1+2\deriv{}{z} \ln \frac{\Tmn{m-1,n}{\mu+1}(z)}{\Tmn{m-1,n+1}{\mu-1}(z)},\label{eq:413a}\\
\mathcal{V}_{m,n}^{(\mu)}(z)=\frac{v_{m,n}^{(\mu)}(z)+1}{v_{m,n}^{(\mu)}(z)-1} &=-1+\frac{2\mu}{z}+2\deriv{}{z} \ln \frac{\Tmn{m-1,n+1}{\mu-2n-1}(z)}{\Tmn{m,n}{\mu-2n-1}(z)},\label{eq:413b}\\
\mathcal{U}_{m,n}^{(\mu)}(z)=\frac{u_{m,n}^{(\mu)}(z)+1}{u_{m,n}^{(\mu)}(z)-1} &=1+2\deriv{}{z}\ln\frac{\Tmn{n,m}{\mu-2m-2n-1}}{\Tmn{n-1,m}{\mu-2m-2n-1}},\label{eq:413c}\end{align} \end{subequations}
are solutions of equation \eqref{eq2:q} and
\begin{subequations}\begin{align}
\mathcal{X}_{m,n}^{(\mu)}(z)=\frac{1}{w_{m,n}^{(\mu)}(z)-1}&=\deriv{}{z} \ln \frac{\Tmn{m-1,n}{\mu+1}(z)}{\Tmn{m-1,n+1}{\mu-1}(z)},\label{eq:413d}\\
\mathcal{Y}_{m,n}^{(\mu)}(z)=\frac{1}{v_{m,n}^{(\mu)}(z)-1}&=-1+\frac{\mu}{z}+\deriv{}{z} \ln \frac{\Tmn{m-1,n+1}{\mu-2n-1}(z)}{\Tmn{m,n}{\mu-2n-1}(z)},\label{eq:413e}\\
\mathcal{Z}_{m,n}^{(\mu)}(z)=\frac{1}{u_{m,n}^{(\mu)}(z)-1} &=\deriv{}{z}\ln\frac{\Tmn{n,m}{\mu-2m-2n-1}}{\Tmn{n-1,m}{\mu-2m-2n-1}},\label{eq:413f}
\end{align} \end{subequations}
are solutions of equation \eqref{eq2:x}.
\end{lemma}
\begin{proof} Since
\beq \frac{w+1}{w-1}=1+\frac{2}{w-1},\label{wrel}\eeq
then we prove \eqref{eq:413d} first. 
By definition
\[{w}_{m,n}^{(\mu)}=\frac{\Tmn{m-1,n}{\mu}\,\Tmn{m-1,n+1}{\mu}}{\Tmn{m,n}{\mu}\,\Tmn{m-2,n+1}{\mu}},\]
and so
\[\frac{1}{w_{m,n}^{(\mu)}-1}=\frac{\Tmn{m,n}{\mu}\,\Tmn{m-2,n+1}{\mu}}{\Tmn{m-1,n}{\mu}\,\Tmn{m-1,n+1}{\mu}-\Tmn{m,n}{\mu}\,\Tmn{m-2,n+1}{\mu}}.\]
Also
\[\deriv{}{z} \ln \frac{\Tmn{m-1,n}{\mu+1}}{\Tmn{m-1,n+1}{\mu-1}}=\frac{ \Hir{\Tmn{m-1,n}{\mu+1}}{\Tmn{m-1,n+1}{\mu-1}}}{\Tmn{m-1,n}{\mu+1}\,\Tmn{m-1,n+1}{\mu-1}}
=\frac{\Tmn{m,n}{\mu}\,\Tmn{m-2,n+1}{\mu}}{\Tmn{m-1,n}{\mu}\,\Tmn{m-1,n+1}{\mu}-\Tmn{m,n}{\mu}\,\Tmn{m-2,n+1}{\mu}},\]
since from the identity \eqref{iden:324c}
\[\Hir{\Tmn{m-1,n}{\mu+1}}{\Tmn{m-1,n+1}{\mu-1}}=\Tmn{m,n}{\mu}\,\Tmn{m-2,n+1}{\mu},\]
where $\D_{z}(f\cdot g)$ is the Hirota operator defined by
\beq{\D_{z}(f\cdot g)=\deriv{f}{z}g-f\deriv{g}{z}},\label{Hirop}\eeq
and from the identity \eqref{iden:Tmn3}
\[\Tmn{m-1,n}{\mu+1}\,\Tmn{m-1,n+1}{\mu-1}=\ \Tmn{m-1,n}{\mu}\,\Tmn{m-1,n+1}{\mu}-\Tmn{m,n}{\mu}\,\Tmn{m-2,n+1}{\mu},\]
which completes the proof of \eqref{eq:413d}. 
The identity \eqref{eq:413a} follows from equation \eqref{wrel}. 

By definition
\[{v}_{m,n}^{(\mu+2n+1)}=\frac{m+n+\mu+1}{m+n+1}\,\frac{\Tmn{m,n}{\mu+1}\,\Tmn{m-1,n+1}{\mu-1}}{\Tmn{m-1,n}{\mu+1}\,\Tmn{m,n+1}{\mu-1}},\]
and so
\begin{align*}
\frac{1}{{v}_{m,n}^{(\mu+2n+1)}-1}+1&=\frac{(m+n+\mu+1)\Tmn{m,n}{\mu+1}\,\Tmn{m-1,n+1}{\mu-1}}{(m+n+\mu+1)\Tmn{m,n}{\mu+1}\,\Tmn{m-1,n+1}{\mu-1}-(m+n+1)\Tmn{m-1,n}{\mu+1}\,\Tmn{m,n+1}{\mu-1}}.
\end{align*}
From \eqref{eqn:GenLagvmn} and \eqref{GenLag_vmn_log} it follows
\[ z\Hir{\Tmn{m-1,n}{\mu+1}}{\Tmn{m,n+1}{\mu-1}}=(m+n+\mu+1)\Tmn{m,n}{\mu+1}\,\Tmn{m-1,n+1}{\mu-1}-(m+n+1)\Tmn{m,n+1}{\mu-1}\,\Tmn{m-1,n}{\mu+1},\]
and since from \eqref{iden:324b}
\[\Hir{\Tmn{m-1,n}{\mu+1}}{\Tmn{m,n+1}{\mu-1}}=\Tmn{m,n}{\mu}\,\Tmn{m-1,n+1}{\mu},\]
then we have
\[z\Tmn{m,n}{\mu}\,\Tmn{m-1,n+1}{\mu} =(m+n+\mu+1)\Tmn{m,n}{\mu+1}\,\Tmn{m-1,n+1}{\mu-1}-(m+n+1)\Tmn{m,n+1}{\mu-1}\,\Tmn{m-1,n}{\mu+1}.\]
Therefore
\[\frac{1}{{v}_{m,n}^{(\mu+2n+1)}-1}+1 =\frac{m+n+\mu+1}{z}\,\frac{\Tmn{m,n}{\mu+1}\,\Tmn{m-1,n+1}{\mu-1}}{\Tmn{m-1,n+1}{\mu}\,\Tmn{m,n}{\mu}},\]
and 
\[
\frac{m+n+\mu+1}{z}+\deriv{}{z} \ln \frac{\Tmn{m-1,n+1}{\mu}}{\Tmn{m,n}{\mu}}=\frac{m+n+\mu+1}{z}\,\frac{\Tmn{m,n}{\mu+1}\,\Tmn{m-1,n+1}{\mu-1}}{\Tmn{m-1,n+1}{\mu}\,\Tmn{m,n}{\mu}}.\]
Hence we obtain \eqref{eq:413e}. The identity \eqref{eq:413b} follows from equation \eqref{wrel}. 

By definition 
\[u_{m,n}^{(\mu+2m+2n+1)}(z)=\frac{\Tmn{n-1,m}{\mu+1}\,\Tmn{n,m}{\mu-1}}{\Tmn{n,m-1}{\mu+1}\,\Tmn{n-1,m+1}{\mu-1}},\]
so
\[\frac{1}{u_{m,n}^{(\mu+2m+2n+1)}(z)-1}=\frac{\Tmn{n,m-1}{\mu+1}\,\Tmn{n-1,m+1}{\mu-1}}{\Tmn{n-1,m}{\mu+1}\,\Tmn{n,m}{\mu-1}-\Tmn{n,m-1}{\mu+1}\,\Tmn{n-1,m+1}{\mu-1}}.\]
Also
\[\deriv{}{z}\ln\frac{\Tmn{n,m}{\mu}}{\Tmn{n-1,m}{\mu}}=\frac{\Hir{\Tmn{n,m}{\mu}}{\Tmn{n-1,m}{\mu}}}{\Tmn{n,m}{\mu}\,\Tmn{n-1,m}{\mu}}
=\frac{\Tmn{n,m-1}{\mu+1}\,\Tmn{n-1,m+1}{\mu-1}}{\Tmn{n-1,m}{\mu+1}\,\Tmn{n,m}{\mu-1}-\Tmn{n,m-1}{\mu+1}\,\Tmn{n-1,m+1}{\mu-1}},\]
since from the identity \eqref{iden:324f}
\[\Hir{\Tmn{n,m}{\mu}}{\Tmn{n-1,m}{\mu}}=\Tmn{n,m-1}{\mu+1}\,\Tmn{n-1,m+1}{\mu-1},\]
and from the identity \eqref{iden:Tmn2}
\[\Tmn{n,m}{\mu}\,\Tmn{n-1,m}{\mu}=\Tmn{n-1,m}{\mu+1}\,\Tmn{n,m}{\mu-1}-\Tmn{n,m-1}{\mu+1}\,\Tmn{n-1,m+1}{\mu-1},\]
which completes the proof of \eqref{eq:413f}.
The identity \eqref{eq:413c} follows from equation \eqref{wrel}.
\end{proof}

\subsection{\label{def:GenUm}Generalised Umemura polynomials}
Rational solutions in cases (ii) and (iii) of Theorem \ref{thm:KLM} are expressed in terms of \textit{generalised Umemura polynomials}.
Umemura \cite{refUm20} defined some polynomials through a differential-difference equation to describe rational solutions of \PV\ \eqref{eq:pv}; see also \cite{refPACpv,refNY98ii,refYamada}.
Subsequently these were generalised by Masuda, Ohta and Kajiwara \cite{refMOK}, who defined the {generalised Umemura polynomial} $R_{m,n}^{(\k)}(z)$ through coupled differential-difference equations and gave a representation as the determinant
\beq R_{m,n}^{(\ell)}(z)=\left|\Matrix{
q_{1}^{(\ell)} & \cdots & q_{-m+2}^{(\ell)} & q_{-m+1}^{(\ell)} & \cdots & q_{-m-n+2}^{(\ell)} \\[2pt]
q_{3}^{(\ell)} & \cdots & q_{-m+4}^{(\ell)} & q_{-m+3}^{(\ell)} & \cdots & q_{-m-n+4}^{(\ell)} \\[2pt]
\vdots & \ddots & \vdots & \vdots & \ddots & \vdots \\[2pt]
q_{2m-1}^{(\ell)} & \cdots & q_{m}^{(\ell)} & q_{m-1} ^{(\ell)}& \cdots & q_{m-n}^{(\ell)} \\[2pt]
p_{n-m}^{(\ell)} & \cdots & p_{n+1}^{(\ell)} & p_{n}^{(\ell)} & \cdots & p_{2n-1} ^{(\ell)}\\[2pt]
p_{n-m-2}^{(\ell)} & \cdots & p_{n-1}^{(\ell)} & p_{n-2}^{(\ell)} & \cdots & p_{2n-3}^{(\ell)} \\[2pt]
\vdots & \ddots & \vdots & \vdots & \ddots & \vdots \\[2pt]
p_{-n-m+2}^{(\ell)} & \cdots & p_{-n+1}^{(\ell)} & p_{-n+2}^{(\ell)} & \cdots & p_{1}^{(\ell)}
}\right|,\label{def:Rmn}\eeq
where $p_{k}^{(\ell)}(z)=L_{k}^{(\ell-1)}(\hfz)$ and $q_{k}^{(\ell)}(z)=L_{k}^{(\ell-1)}(-\hfz)$, with $p_{k}^{(\ell)}=q_{k}^{(\ell)}=0$ for $k<0$, and $L_{k}^{(\a)}(z)$ is the {Laguerre polynomial} \eqref{eq:alp}. Recently Masuda \cite{refMasuda25} studied properties of zeros of the polynomial $R_{m,n}^{(\ell)}(z)$. 

Rather than use \eqref{def:Rmn} for the generalised Umemura polynomial, we prefer to use an equivalent representation given in Definition \ref{dfn:GenUm}, which we find simpler to compute and useful to prove algebraic and combinatorial properties of the polynomial, though we do not pursue this further here.

\begin{definition}{\label{dfn:GenUm}\rm The {\textit{generalised Umemura polynomial}} $\Umn{m,n}{\k}(z)$ is given by the Wronskian
\beq \Umn{m,n}{\k}(z)=c_{m,n}\exp(\hf mz)\Wr\!\left({\ph^{(\k)}_{1}},{\ph^{(\k)}_{3}},\ldots,{\ph^{(\k)}_{2m-1}};{\psi^{(\k)}_{1}},{\psi^{(\k)}_{3}},\ldots,{\psi^{(\k)}_{2n-1}}\right),\quad m,n\geq0, 
\label{def:Umn}\eeq
where 
\beq\label{def:cmn} c_{m,n}=(-1)^{m}\, 2^{(m+n)(m+n+1)/2}
\prod_{j=1}^{m-1}\frac{(2j+1)!}{j!}\prod_{k=1}^{n-1}\frac{(2k+1)!}{k!},\eeq
${\ph^{(\k)}_{n}(z)=\rme^{-z/2}\,L_{n}^{(\k)}(\hfz)}$ and ${\psi^{(\k)}_{n}(z)=L_{n}^{(\k)}(-\hfz)}$, 
with $\k$ a parameter, and $\Lag{n}{\k}(z)$ is the {Laguerre polynomial} \eqref{eq:alp} and $\Umn{0,0}{\k}(z)=1$.}
\end{definition}

\begin{remarks}{\rm 
\begin{enumerate}[(i)]\item[] 
\item The generalised Laguerre polynomial $\Tmn{m,n}{\mu}(z)$ is a Wronskian involving one sequence of Laguerre polynomials, whereas the generalised Umemura polynomial $\Umn{m,n}{\k}(z)$ involves two sequences of Laguerre polynomials, unless either $m=0$ or $n=0$. This is analogous to the situation for rational solutions of \PIV\ \eqref{eq:piv}, where one family of rational solutions involves one sequence of Hermite polynomials and a second family involves two sequences of Hermite polynomials, cf.~\cite{refPAC06review,refNY99}.
\item The generalised Umemura polynomial $\Umn{m,n}{\k}(z)$ given by \eqref{def:Umn} is a monic polynomial of degree $\hf m(m+1)+\hf n(n+1)$ and has the symmetry
\beq \Umn{m,n}{\k}(z)=\Umn{n,m}{-\k-2m-2n}(z).\label{Umn:sym}\eeq
\item The polynomials $R_{m,n}^{(\ell)}(z)$ and $\Umn{m,n}{\k}(z)$ given by \eqref{def:Rmn} and \eqref{def:Umn} respectively, are equivalent, which can be shown using matrix row and column operations. Further the two polynomials are related as follows
\[\Umn{m,n}{\k}(z)=c_{m,n}(-1)^{m+n(n+1)/2}\left(\tfrac{1}{2}\right)^{(m+n)(m+n-1)/2}R_{m,n}^{(-\k-m-n)}(z),\]
with $c_{m,n}$ given by \eqref{def:cmn}.
\item
When $m=0$, the polynomial $\Umn{0,n}{\k}(z)$ is equivalent to the polynomial defined by Umemura \cite{refUm20} through a differential-difference equation. Also $\Umn{0,n}{\k}(z)$ is equivalent to the determinant $\tau_{m}^{(n)}(z)$ defined in Theorem \ref{thm:KYO} below; see also \cite{refNY98ii}. Due to the symmetry \eqref{Umn:sym}, the same is true for the polynomial $\Umn{m,0}{\k}(z)$.
\end{enumerate}
}\end{remarks}

\begin{theorem}{\label{thm:genUm}\rm\textit{Suppose that $\Umn{m,n}{\k}(z)$ is the generalised Umemura polynomial given by \eqref{def:Umn}.} 
\begin{enumerate}[(a)]
\item\textit{In case (ii) of Theorem \ref{thm:KLM}, \PV\ \eqref{eq:pv} has the rational solutions for $m\geq1$ and $n\geq1$
\begin{subequations}
\begin{align}
\wh{w}_{m,n}^{(\k)}(z)&=-\frac{\Umn{m,n-1}{2\k}(z)\,\Umn{m-1,n}{2\k+2}(z)}{\Umn{m,n-1}{2\k+2}(z)\,\Umn{m-1,n}{2\k}(z)},&&
\pms{\hf(m+\k)^2}{-\hf(n+\k)^2}{m+n}, \label{def:GenUmw}\\
\wh{v}_{m,n}^{(\k)}(z)&=-\frac{\Umn{m-1,n-1}{2\k-2n+3}(z)\,\Umn{m,n}{2\k-2n-1}(z)}{\Umn{m-1,n-1}{2\k-2n+1}(z)\,\Umn{m,n}{2\k-2n+1}(z)},&&
\pms{\hf(m+\k)^2}{-\hf(n-\k)^2}{m-n},\label{def:GenUmv}
\end{align}
{In the case when $n=0$, then for $m\geq1$
\beq \wh{w}_{m,0}^{(\k)}(z)=\wh{v}_{m,0}^{(\k)}(z)=-\frac{\Umn{m,0}{2\k-1}(z)\,\Umn{m-1,0}{2\k+2}(z)}{\Umn{m,0}{2\k+1}(z)\,\Umn{m-1,0}{2\k}(z)},\qquad
\pms{\hf(m+\k)^2}{-\hf\k^2}{m}, \label{Um:wm0}\eeq
and in the case when $m=0$ then for $n\geq1$
\beq \wh{w}_{0,n}^{(\k)}(z)=\wh{v}_{0,n}^{(-\k)}(z)=-\frac{\Umn{0,n}{2\k+1}(z)\,\Umn{0,n-1}{2\k}(z)}{\Umn{0,n}{2\k-1}(z)\,\Umn{0,n-1}{2\k+2}(z)},\qquad
\pms{\hf\k^2}{-\hf(n+\k)^2}{n}.\label{Um:wn0} \eeq}\end{subequations}}
\item\textit{In case (iii) of Theorem \ref{thm:KLM}, \PV\ \eqref{eq:pv} has the rational solution for $m\geq1$ and $n\geq1$
\begin{subequations}
\beq\wh{u}_{m,n}^{(\k)}(z)=-\,\frac{\Umn{m,n-1}{\k+1}(z)\,\Umn{m,n+1}{\k-1}(z)}{\Umn{m-1,n}{\k+1}(z)\,\Umn{m+1,n}{\k-1}(z)},\qquad
\pms{\hf(m+\hf)^2}{-\hf(n+\hf)^2}{m+n+\k}.\label{def:GenUmu}\eeq
{In the case when
$n=0$, then for $m\geq1$
\beq \wh{u}_{m,0}^{(\k)}(z)=-\,\frac{\Umn{m,0}{\k}(z)\,\Umn{m,1}{\k-1}(z)}{\Umn{m-1,0}{\k+1}(z)\,\Umn{m+1,0}{\k-1}(z)},\qquad
\pms{\hf(m+\hf)^2}{-\tfrac{1}{8}}{m+\k},\eeq
and in the case when $m=0$ then for $n\geq1$
\beq \wh{u}_{0,n}^{(\k)}(z)=-\,\frac{\Umn{0,n-1}{\k+1}(z)\,\Umn{0,n+1}{\k-1}(z)}{\Umn{0,n}{\k}(z)\,\Umn{1,n}{\k-1}(z)},\qquad
\pms{\tfrac{1}{8}}{-\hf(n+\hf)^2}{n+\k}.\eeq}\end{subequations}}
\end{enumerate}}\end{theorem}
\begin{proof} See Masuda, Ohta and Kajiwara \cite[Theorem 1.1]{refMOK}.\end{proof}

\begin{remark}{\label{rem:410}\rm The rational solutions in Theorem \ref{thm:genUm} can also be written as logarithmic derivatives
\begin{subequations}\label{GenUm_log}
\begin{align}
\wh{w}_{m,n}^{(\k)}(z) 
&=\frac{z}{m+\k}\Ln{\Umn{m-1,n}{2\k}(z)}{\Umn{m,n-1}{2\k+2}(z)}-\frac{n+\k}{m+\k},\\
\wh{v}_{m,n}^{(\k+ n)}(z) 
&=\frac{z}{m+n+\k}\Ln{\Umn{m-1,n-1}{2\k+1}(z)}{\Umn{m,n}{2\k+1}(z)}-\frac{\k}{m+n+\k}, \\
\wh{u}_{m,n}^{(\k)}(z) 
&=\frac{2z}{2m+1}\Ln{\Umn{m-1,n}{\k+1}(z)}{\Umn{m+1,n}{\k-1}(z)}+1.
\end{align}\end{subequations}
{Using computer algebra we have verified these alternative forms of the rational solutions in Theorem \ref{thm:genUm} given by
\eqref{GenLag_log} for several small values of $m$ and $n$. Comparing these expressions with those in Theorem \ref{thm:genUm} gives the relations \eqref{eq:A12}, which we expect can be proved using the Jacobi identity \eqref{JacobiId}. We note that Masuda \cite[Appendix A]{refMasuda25} has written the rational solutions expressed in terms of generalised Umemura polynomials of \spv\ \eqref{eq:spv} as logarithmic derivatives}.
}\end{remark}

In the following Lemma we derive the solutions of equations \eqref{eq2:q} and \eqref{eq2:x} in terms of generalised Umemura polynomials.

\begin{lemma}
If $\wh{w}_{m,n}^{(\k)}(z)$, $\wh{v}_{m,n}^{(\k)}(z)$ and $\wh{u}_{m,n}^{(\k)}(z)$ are rational solutions of \PV\ defined in Theorem \ref{thm:genUm} and $\Umn{m,n}{\mu}(z)$ is the generalised Umemura polynomial given by \eqref{def:Umn}, then
\begin{subequations}\label{eq:GenUmPolywvu1}\begin{align}
\wh{\mathcal{W}}_{m,n}^{(\k)}(z)=\frac{\wh{w}_{m,n}^{(\k)}(z)+1}{\wh{w}_{m,n}^{(\k)}(z)-1}&=2\deriv{}{z}\ln\frac{\Umn{m,n}{2\k}(z)}{\Umn{m-1,n-1}{2\k+2}(z)}, \label{eq:GenUmPolyw1}\\
\wh{\mathcal{V}}_{m,n}^{(\k)}(z)=\frac{\wh{v}_{m,n}^{(\k)}(z)+1}{\wh{v}_{m,n}^{(\k)}(z)-1}&=2\deriv{}{z}\ln\frac{\Umn{m,n-1}{2\k-2n+1}(z)}{\Umn{m-1,n}{2\k-2n+1}(z)},\label{eq:GenUmPolyv1}\\
\wh{\mathcal{U}}_{m,n}^{(\k)}(z)=\frac{\wh{u}_{m,n}^{(\k)}(z)+1}{\wh{u}_{m,n}^{(\k)}(z)-1}&=\frac{2(m+n+\k)}{z}+2\deriv{}{z}\ln\frac{\Umn{m,n}{\k+1}(z)}{\Umn{m,n}{\k-1}(z)},\label{eq:GenUmPolyu1}
\end{align}\end{subequations}
are solutions of \eqref{eq2:q}
and
\begin{subequations}\label{eq:GenUmPolywvu2}\begin{align}
\wh{\mathcal{X}}_{m,n}^{(\k)}(z)=\frac{1}{\wh{w}_{m,n}^{(\k)}(z)-1}&=-\frac{1}{2}+\deriv{}{z}\ln\frac{\Umn{m,n}{2\k}(z)}{\Umn{m-1,n-1}{2\k+2}(z)},\label{eq:GenUmPolyw2}\\
\wh{\mathcal{Y}}_{m,n}^{(\k)}(z)=\frac{1}{\wh{v}_{m,n}^{(\k)}(z)-1}&=-\frac{1}{2}+\deriv{}{z}\ln\frac{\Umn{m,n-1}{2\k-2n+1}(z)}{\Umn{m-1,n}{2\k-2n+1}(z)},\label{eq:GenUmPolyv2}\\
\wh{\mathcal{Z}}_{m,n}^{(\k)}(z)=\frac{1}{\wh{u}_{m,n}^{(\k)}(z)-1}&=-\frac{1}{2}+\frac{(m+n+\k)}{z}+\deriv{}{z}\ln\frac{\Umn{m,n}{\k+1}(z)}{\Umn{m,n}{\k-1}(z)}\label{eq:GenUmPolyu2}
\end{align}\end{subequations}
are solutions of \eqref{eq2:x}.
\end{lemma}

\begin{proof} We note that \eqref{eq:GenUmPolywvu2} follow from \eqref{eq:GenUmPolywvu1} using \eqref{wrel}.
By definition
\[\wh{w}_{m,n}^{(\k)}(z)=-\frac{\Umn{m,n-1}{2\k}(z)\,\Umn{m-1,n}{2\k+2}(z)}{\Umn{m,n-1}{2\k+2}(z)\,\Umn{m-1,n}{2\k}(z)},\]
so that we have
\begin{align*} 
&\frac{\wh{w}_{m,n}^{(\k)}(z)+1}{\wh{w}_{m,n}^{(\k)}(z)-1}=\frac{\Umn{m,n-1}{2\k}(z)\,\Umn{m-1,n}{2\k+2}(z)-\Umn{m,n-1}{2\k+2}(z)\,\Umn{m-1,n}{2\k}(z)}{\Umn{m,n-1}{2\k}(z)\,\Umn{m-1,n}{2\k+2}(z)+\Umn{m,n-1}{2\k+2}(z)\,\Umn{m-1,n}{2\k}(z)},\\
&\deriv{}{z}\ln\frac{\Umn{m,n}{2\k}(z)}{\Umn{m-1,n-1}{2\k+2}(z)}=\frac{\Hir{\Umn{m,n}{2\k}(z)}{\Umn{m-1,n-1}{2\k+2}(z)}}{\Umn{m,n}{2\k}(z)\,\Umn{m-1,n-1}{2\k+2}(z)},
\end{align*}
where $\D_{z}(f\cdot g)$ is the Hirota operator \eqref{Hirop}.
From the identities \eqref{eq:Masuda34a} and \eqref{eq:Masuda34b}
\begin{align*} 4\Hir{\Umn{m,n}{2\k}(z)}{\Umn{m-1,n-1}{2\k+2}(z)}&=\Umn{m,n-1}{2\k}(z)\,\Umn{m-1,n}{2\k+2}(z)-\Umn{m,n-1}{2\k+2}(z)\,\Umn{m-1,n}{2\k}(z), \\
2\Umn{m,n}{2\k}(z)\,\Umn{m-1,n-1}{2\k+2}(z)&=\Umn{m,n-1}{2\k}(z)\,\Umn{m-1,n}{2\k+2}(z)+\Umn{m,n-1}{2\k+2}(z)\,\Umn{m-1,n}{2\k}(z), \end{align*}
respectively, therefore
\[\frac{2\Hir{\Umn{m,n}{2\k}(z)}{\Umn{m-1,n-1}{2\k+2}(z)}}{\Umn{m,n}{2\k}(z)\,\Umn{m-1,n-1}{2\k+2}(z)}=\frac{\Umn{m,n-1}{2\k}(z)\,\Umn{m-1,n}{2\k+2}(z)-\Umn{m,n-1}{2\k+2}(z)\,\Umn{m-1,n}{2\k}(z)}{\Umn{m,n-1}{2\k}(z)\,\Umn{m-1,n}{2\k+2}(z)+\Umn{m,n-1}{2\k+2}(z)\,\Umn{m-1,n}{2\k}(z)}, \]
which proves \eqref{eq:GenUmPolyw1}. 

By definition
\[ \wh{v}_{m,n}^{(\k+ n-1/2)}(z)=-\frac{\Umn{m-1,n-1}{2\k+2}(z)\,\Umn{m,n}{2\k-2}(z)}{\Umn{m-1,n-1}{2\k}(z)\,\Umn{m,n}{2\k}(z)},\]
so that we have
\begin{align*} &\frac{\wh{v}_{m,n}^{(\k+ n-1/2)}(z)+1}{\wh{v}_{m,n}^{(\k+ n-1/2)}(z)-1}=\frac{\Umn{m,n}{2\k-2}\,\Umn{m-1,n-1}{2\k+2}-\Umn{m,n}{2\k}\,\Umn{m-1,n-1}{2\k}}{\Umn{m,n}{2\k-2}\,\Umn{m-1,n-1}{2\k+2}+\Umn{m,n}{2\k}\,\Umn{m-1,n-1}{2\k}},\\
&\deriv{}{z}\ln\frac{\Umn{m,n-1}{2\k}}{\Umn{m-1,n}{2\k}}=\frac{\Hir{\Umn{m,n-1}{2\k}}{\Umn{m-1,n}{2\k}}}{\Umn{m,n-1}{2\k}\,\Umn{m-1,n}{2\k}}.
\end{align*}

From the identities \eqref{eq:Masuda12a} and \eqref{eq:Masuda12b} 
\begin{align*} 4\Hir{\Umn{m,n-1}{2\k}}{\Umn{m-1,n}{2\k}} &=\Umn{m,n}{2\k-2}\,\Umn{m-1,n-1}{2\k+2}-\Umn{m,n}{2\k}\,\Umn{m-1,n-1}{2\k},\\
2\Umn{m,n-1}{2\k}\,\Umn{m-1,n}{2\k}&=\Umn{m,n}{2\k-2}\,\Umn{m-1,n-1}{2\k+2}+\Umn{m,n}{2\k}\,\Umn{m-1,n-1}{2\k},\end{align*}
respectively, therefore
\[\frac{2\Hir{\Umn{m,n-1}{2\k}}{\Umn{m-1,n}{2\k}}}{\Umn{m,n-1}{2\k}\,\Umn{m-1,n}{2\k}}=\frac{\Umn{m,n}{2\k-2}\,\Umn{m-1,n-1}{2\k+2}-\Umn{m,n}{2\k}\,\Umn{m-1,n-1}{2\k}}{\Umn{m,n}{2\k-2}\,\Umn{m-1,n-1}{2\k+2}+\Umn{m,n}{2\k}\,\Umn{m-1,n-1}{2\k}}, \]
which proves \eqref{eq:GenUmPolyv1}. 

By definition
\[\wh{u}_{m,n}^{(\k)}(z)=-\,\frac{\Umn{m,n-1}{\k+1}(z)\,\Umn{m,n+1}{\k-1}(z)}{\Umn{m-1,n}{\k+1}(z)\,\Umn{m+1,n}{\k-1}(z)},\]
and so
\[\frac{\wh{u}_{m,n}^{(\k)}+1}{ \wh{u}_{m,n}^{(\k)}-1}=\frac{\Umn{m,n-1}{\k+1}\,\Umn{m,n+1}{\k-1}-\Umn{m-1,n}{\k+1}\,\Umn{m+1,n}{\k-1}}{\Umn{m,n-1}{\k+1}\,\Umn{m,n+1}{\k-1}+\Umn{m-1,n}{\k+1}\,\Umn{m+1,n}{\k-1}}.\]
Also
\begin{align*} 
\deriv{}{z}\ln\frac{\Umn{m,n}{\k+1}}{\Umn{m,n}{\k-1}}+\frac{m+n+\k}{z} 
&=\frac{\Hir{\Umn{m,n}{\k+1}}{\Umn{m,n}{\k-1}}}{\Umn{m,n}{\k+1}\,\Umn{m,n}{\k-1}}+\frac{m+n+\k}{z} \\
&=\frac{z\Hir{\Umn{m,n}{\k+1}}{\Umn{m,n}{\k-1}}+(m+n+\k) \Umn{m,n}{\k+1}\,\Umn{m,n}{\k-1}}{z\Umn{m,n}{\k+1}\,\Umn{m,n}{\k-1}} 
 \end{align*}
It can be shown that
\begin{align} 
\Umn{m,n-1}{\k+1}\,\Umn{m,n+1}{\k-1}+\Umn{m-1,n}{\k+1}\,\Umn{m+1,n}{\k-1}&=2z\Umn{m,n}{\k+1}\,\Umn{m,n}{\k-1},\label{eq:426}\\
\Umn{m,n-1}{\k+1}\,\Umn{m,n+1}{\k-1}-\Umn{m-1,n}{\k+1}\,\Umn{m+1,n}{\k-1}&=4z\Hir{\Umn{m,n}{\k+1}}{\Umn{m,n}{\k-1}}+4(m+n+\k)\Umn{m,n}{\k+1}\,\Umn{m,n}{\k-1},\label{eq:425}
 \end{align}
which proves \eqref{eq:GenUmPolyu1}. 
Equation \eqref{eq:426} is the identity \eqref{eq:B7} and 
using computer algebra, we have verified equation \eqref{eq:425} for several small values of $m$ and $n$. 
\end{proof}

\section{\label{sec:rats}Rational solutions of the discrete equations}
In this section we discuss rational solutions of the discrete equations derived in \S\ref{sec:deqn}. It is known that there are rational solutions of the \dPII\ equation \eqref{eq:dPII}, see \cite{refKajiwara,refKNT,refKYO,refNY98ii}. This is illustrated in the following theorem.

\begin{theorem}{\label{thm:KYO}
Define $\tau_{m}^{(n)}(z)$ to be the $m\times m$ determinant given by
\beq\tau_{m}^{(n)}(z)=\left|\Matrix{
p^{(n)}_{m} & p^{(n)}_{m+1} & \cdots & p^{(n)}_{2m-1} \\
p^{(n)}_{m-2} & p^{(n)}_{m-1} & \cdots & p^{(n)}_{2m-3} \\
\vdots & \vdots & \ddots & \vdots \\
p^{(n)}_{-m+2} & p^{(n)}_{-m+3} & \cdots & p^{(n)}_{1}}\right|,\label{def:taumn}\eeq
where $p^{(n)}_{k}(z)=L^{(n)}_{k}(\hfz)$, with $L^{(n)}_k(z)$ the Laguerre polynomial, for $k\geq0$ and $p^{(n)}_{k}(z)=0$ for $k<0$,
then
\beq \label{sol:KYO} q_{n}(z;m)=\frac{\tau_{m+1}^{(n+1)}(z)\,\tau_{m}^{(n-1)}(z)}{\tau_{m+1}^{(n)}(z)\,\tau_{m}^{(n)}(z)}-1,\eeq
satisfies 
\beq\label{dPII:KYO} q_{n+1}+q_{n-1}=\frac{4}{z}\frac{(n+1)q_{n}-m-1}{1-q_{n}^2},\eeq
which is \dPII\ \eqref{eq:dPII} with $\zeta=4/z$, $\la=1$ and $\rho=-m-1$.
}\end{theorem}
\begin{proof} See Kajiwara, Yamamoto and Ohta \cite{refKYO}.\end{proof}

\comment{\begin{remarks}{\rm
\begin{enumerate}[(i)]\item[]
\item The determinant $\tau_{m}^{(n)}(z)$ can be written as the Wronskian
\beq \tau_{m}^{(n)}(z)=\Wr\big(p^{(n)}_{1},p^{(n)}_{3},\ldots,p^{(n)}_{2m-1} \big),\eeq
with $p^{(n)}_{k}(z)=L^{(n)}_{k}(\hfz)$,
and is equivalent to the generalised Umemura polynomial $\Umn{0,n}{\k}(z)$. 
\item The function $q_{n}(z;m)$ satisfies equation \eqref{eq2:q} 
with 
\[\pms{\tfrac{1}{8}(m-n)^2}{-\tfrac{1}{8}(m+n+2)^2}{-m-1}.\]
\end{enumerate}
}\end{remarks}}

\begin{remarks}{\rm
\begin{enumerate}[(i)]\item[]
\item The determinant $\tau_{m}^{(n)}(z)$ can be written as the Wronskian
\beq \tau_{m}^{(n)}(z)=(-1)^{m(m-1)/2}\Wr\big(p^{(n-m+1)}_{1},p^{(n-m+1)}_{3},\ldots,p^{(n-m+1)}_{2m-1} \big),\eeq
with $p^{(n)}_{k}(z)=L^{(n)}_{k}(\hfz)$,
and is related to the generalised Umemura polynomial by
\[\Umn{m,0}{n-m+1}(z) = \Umn{0,m}{-n-m-1}(z) = c_{0,m}(-1)^{m(m+1)/2}\left(\tfrac{1}{2}\right)^{m(m-1)/2}\tau_{m}^{(n)}(z).\]
\item The function $q_{n}(z;m)$ satisfies equation \eqref{eq2:q} 
with 
\[\pms{\tfrac{1}{8}(m-n)^2}{-\tfrac{1}{8}(m+n+2)^2}{-m-1}.\]
\end{enumerate}
}\end{remarks}

\subsection{\label{ssec:GenLag}Rational solutions in terms of generalised Laguerre polynomials}
Each of the transformations $\mathcal{R}_{1}$ \eqref{bt:R1}, $\mathcal{R}_{2}$ \eqref{bt:R2} and $\mathcal{R}_{3}$ \eqref{bt:R3} is actually four transformations since there are two independent choices of sign for $a=\pm\sqrt{2\a}$ and $b=\pm\sqrt{2\b}$. This can lead to four different solutions from applying the transformation to a given solution. This is illustrated in Table \ref{table:R1params} which shows the effect on the parameters of applying the transformations $\mathcal{R}_{1}$ \eqref{bt:R1} and $\mathcal{R}_{1}^2$ on a solution $w(z)$ of \PV\ \eqref{eq:pv} with $\pms{\hf a^2}{-\hf b^2}{c}$; see also Example \ref{exam:54}.
\begin{table}[h]
\[\begin{array}{|c|ccc|}
\hline
& \a_{1} & \b_{1} & \ga_{1} \\ \hline
\mathcal{R}_{1}(a,b,c) & \tfrac{1}{8}(a-b+c+1)^2 &-\tfrac{1}{8}(a-b-c+1)^2 &a+b\\
\mathcal{R}_{1}(-a,b,c) & \tfrac{1}{8}(a+b-c-1)^2 &-\tfrac{1}{8}(a+b+c-1)^2 &-a+b\\
\mathcal{R}_{1}(a,-b,c) & \tfrac{1}{8}(a+b+c+1)^2 &-\tfrac{1}{8}(a+b-c+1)^2 &a-b\\
\mathcal{R}_{1}(-a,-b,c) & \tfrac{1}{8}(a-b-c-1)^2 &-\tfrac{1}{8}(a-b+c-1)^2 &-a-b \\ \hline\hline
\mathcal{R}_{1}^2(a,b,c) & \hf(a+1)^2 &-\hf(b-1)^2 &c\\
\mathcal{R}_{1}^2(-a,b,c) & \hf(a-1)^2 &-\hf(b-1)^2 &c\\
\mathcal{R}_{1}^2(a,-b,c) & \hf(a+1)^2 &-\hf(b+1)^2 &c\\
\mathcal{R}_{1}^2(-a,-b,c) & \hf(a-1)^2 &-\hf(b+1)^2 &c \\ \hline
\end{array}\]
\caption{\label{table:R1params}The effect on the parameters of applying the transformations $\mathcal{R}_{1}$ \eqref{bt:R1} and $\mathcal{R}_{1}^2$ on a solution $w(z)$ of \PV\ with $\pms{\hf a^2}{-\hf b^2}{c}$.}
\end{table}

\subsubsection{\label{adPII:GenLag}First discrete equation: asymmetric discrete \PII}
For the \bt\ $\mathcal{R}_{1}=\T_{1,-1,1}$, we know that $\mathcal{R}_{1}^2(a,b,c)=(a+1,b-1,c)$, so that
\begin{subequations}\label{bt:R1w}\begin{align} 
\mathcal{R}_{1}^2\big(w_{m,n}^{(\mu)};m,m+2n+1+\mu,\mu\big)&=\big(w_{m+1,n-1}^{(\mu)};m+1,m+2n+\mu,\mu\big),\label{bt:R1wa}\\
\mathcal{R}_{1}^2\big(w_{m,n}^{(\mu)};-m,m+2n+1+\mu,\mu\big)&=\big(w_{m-1,n}^{(\mu)};1-m,m+2n+\mu,\mu\big),\label{bt:R1wb}\\
\mathcal{R}_{1}^2\big(w_{m,n}^{(\mu)};m,-m-2n-1-\mu,\mu\big)&=\big(w_{m+1,n}^{(\mu)};m+1,-m-2n-2-\mu,\mu\big),\label{bt:R1wc}\\
\mathcal{R}_{1}^2\big(w_{m,n}^{(\mu)};-m,-m-2n-1-\mu,\mu\big)&=\big(w_{m-1,n+1}^{(\mu)};1-m,-m-2n-2-\mu,\mu\big).\label{bt:R1wd}
\end{align}\end{subequations}
From \eqref{bt:R1w} we see that except in the case \eqref{bt:R1wc}, the \bt\ $\mathcal{R}_{1}$ \eqref{bt:R1} can only be applied a finite number of times before the sequence terminates as either $m$ or $n$ becomes zero, as illustrated in the following example. 

\begin{example}{\label{exam:54}\rm
Consider the rational solution
{\begin{align*}
w_{1,1}^{(\mu-3)}(z)&=-\frac{(z-\mu+1)\big[z^2-2\mu z+\mu(\mu+1)\big]}{z^2-2\mu z+\mu(\mu-1)},
\end{align*}}
which satisfies \PV\ \eqref{eq:pv} with $\pms{\hf}{-\hf(\mu+1)^2}{\mu-3}$, so that
{\beq q_{0}(z)
=\frac{w_{1,1}^{(\mu-3)}(z)+1}{w_{1,1}^{(\mu-3)}(z)-1}=1+\frac{2}{z-\mu}-\frac{4(z-\mu+1)}{z^2-2(\mu-1) z+\mu(\mu-1)}.\label{sol:q0}\eeq}

{To derive a sequence of solutions we define the \bt\
\beq\mathcal{Q}_{1}(q;a,b,c)=\left(\frac{2}{z (q^2-1)}\left\{{z}\deriv{q}{z}-(a-b)q-a-b\right\};\hf(a-b+c+1),\hf(-a+b+c-1),a+b \right),\label{bt:Q1} \eeq
since
\[ \mathcal{Q}_{1}(q)=\frac{\mathcal{R}_{1}\big(\frac{q+1}{q-1}\big)+1}{\mathcal{R}_{1}\big(\frac{q+1}{q-1}\big)-1}=\frac{2}{z (q^2-1)}\left\{{z}\deriv{q}{z}-(a-b)q-a-b\right\}.\]}

Applying the \bt\ $\mathcal{Q}_{1}$ \eqref{bt:Q1}, with $(a,b,c)=(1,\mu+1,\mu-3)$,
gives the sequence of rational solutions 
\begin{align*}q_1^{[a]}(z) &=1-\frac{2}{z-\mu+1}+\frac{4(z-\mu)}{z^2-2\mu z+\mu(\mu-1)},\\ 
q_2^{[a]}(z)&=1-\frac{4(z-\mu+1)}{z^2-2(\mu-1)z+(\mu-1)(\mu-2)},
\qquad q_{3}^{[a]}(z)=1.
\end{align*}
Applying $\mathcal{Q}_{1}$ \eqref{bt:Q1}, with 
$(a,b,c)=(-1,\mu+1,\mu-3)$, 
gives the sequence of rational solutions
\[ q_1^{[b]}(z)=-1+\frac{2\mu}{z}-\frac{2}{z-\mu+1}, \qquad q_{2}^{[b]}(z)=1,\]
and applying $\mathcal{Q}_{1}$ \eqref{bt:Q1}, with 
$(a,b,c)=(-1,-\mu-1,\mu-3)$, 
gives the sequence of rational solutions
\[ q_1^{[d]}(z)=-1-\frac{4(z-\mu)}{z^2-2\mu z+\mu(\mu+1)},\qquad q_{2}^{[d]}(z)=1.\]
In each case the \bt\ only produces a finite number of solutions.

Applying $\mathcal{Q}_{1}$ \eqref{bt:Q1}, with 
$(a,b,c)=(1,-\mu-1,\mu-3)$, 
gives an infinite sequence of solutions. 
The next few are
\begin{align*}
q_1(z)&=1-\frac{2\mu}{z}+2\Ln{z^2-2\mu z+\mu(\mu-1)}{z^2-2\mu z+\mu(\mu+1)},\\
q_2(z)&=1+2\Ln{z^2-2(\mu+1)z+\mu(\mu+1)}{z^4-4\mu z^3+6\mu^2z^2-4\mu(\mu^2-1)z+\mu^2(\mu^2-1)},\\
q_3(z)&=1-\frac{2\mu}{z}+2\Ln{z^3-3(\mu+1)z^2+3\mu(\mu+1)z-\mu(\mu^2-1)}{z^4-4(\mu+1)z^3+6(\mu+1)^2z^2-4\mu(\mu+1)(\mu+2)z+\mu(\mu+1)^2(\mu+2)}.
\end{align*}
Generally for $n\geq0$
\begin{align*} q_{2n}(z)=1+2\Ln{\Tmn{n,1}{\mu-2}(z)}{\Tmn{n,2}{\mu-4}(z)},\qquad q_{2n+1}(z)=1-\frac{2\mu}{z}+2\Ln{\Tmn{n+1,1}{\mu-3}(z)}{\Tmn{n,2}{\mu-3}(z)},\end{align*}
satisfy
\[ q_{n+1}+q_{n-1}=\frac{4}{z}\,\frac{(n+\mu+2)q_{n}-\tfrac{3}{2}-(\mu-\tfrac{3}{2})(-1)^n}{1-q_{n}^2},\]
which is asymmetric \dPII\ \eqref{eq:adPII} with $\la=\mu+2$, $\rho=-\tfrac{3}{2}$ and $\ph=-\mu+\tfrac{3}{2}$, whilst $x_{n}=q_{2n}$ and $y_{n}=q_{2n+1}$ satisfy
\begin{align*}
x_{n+1}+x_{n}&=\frac{4}{z}\,\frac{(2n+\mu+3)y_{n}+\mu-3}{1-y_{n}^2},\qquad
y_{n}+y_{n-1}=\frac{4}{z}\,\frac{(2n+\mu+2)x_{n}-\mu}{1-x_{n}^2}.
\end{align*}
Also, applying the inverse transformation $\mathcal{Q}_{1}^{-1}$ we obtain
\[ q_{-1}(z)=1-\frac{2\mu}{z}+\frac{2}{z-\mu+1},\qquad q_{-2}(z)=1.\]
}\end{example}

We note that for the solution $q_{0}$ \eqref{sol:q0}
\[\mathcal{Q}_{1}^3\left(q_{0};1,\mu+1,\mu-3\right)=1,\quad \mathcal{Q}_{1}^2\left(q_{0};-1,\mu+1,\mu-3 \right)=1, \quad
\mathcal{Q}_{1}^2\left(q_{0};-1,-\mu-1,\mu-3 \right)=1.\]
where $\mathcal{Q}_{1}$ is the \bt\ 
Further if $\mathcal{R}_{1}$ \eqref{bt:R1} is the \bt\ given by \eqref{bt:R1wc} then
\[ w_{1,1}^{(\mu-3)}\arr{\mathcal{R}_{1}}\frac{1}{{v}_{1,1}^{(\mu)}} \arr{\mathcal{R}_{1}} w_{2,1}^{(\mu-3)}\arr{\mathcal{R}_{1}}\frac{1}{{v}_{2,1}^{(\mu)}} \arr{\mathcal{R}_{1}} w_{3,1}^{(\mu-3)}\]
{or equivalently
\[ \frac{w_{1,1}^{(\mu-3)}+1}{w_{1,1}^{(\mu-3)}-1}\arr{\mathcal{Q}_{1}} 
\frac{1+{v}_{1,1}^{(\mu)}}{1-{v}_{1,1}^{(\mu)}}\arr{\mathcal{Q}_{1}} 
\frac{w_{2,1}^{(\mu-3)}+1}{w_{1,1}^{(\mu-3)}-1}\arr{\mathcal{Q}_{1}} 
\frac{1+{v}_{2,1}^{(\mu)}}{1-{v}_{2,1}^{(\mu)}}\arr{\mathcal{Q}_{1}}
\frac{w_{3,1}^{(\mu-3)}+1}{w_{3,1}^{(\mu-3)}-1}\]}

More generally, using the solution $w_{m,n}^{(\mu)}(z)$, for fixed $m$ and $n$, then we have the following Lemma. 
\begin{lemma}
{Consider the rational functions
\[Q_{2N}(z)=X_{N}(z)=\frac{w_{N+m,n}^{(\mu)}(z)+1}{w_{N+m,n}^{(\mu)}(z)-1},\qquad Q_{2N+1}(z)=Y_{N}(z)=\frac{1+v_{N+m,n}^{(\mu+2n+1)}(z)}{1-v_{N+m,n}^{(\mu+2n+1)}(z)},\]
with $w_{m,n}^{(\mu)}(z)$ and $v_{m,n}^{(\mu)}(z)$ as given in Theorem \ref{thm:genLag},
then $Q_{N}(z)$ satisfies
\[Q_{N+1}+Q_{N-1}= \frac{4}{z}\frac{(N+2m+2n+\mu+1) Q_{N}-n-\hf-(n+\mu+\hf)(-1)^N}{1-Q_{N}^2},\]
which is asymmetric \dPII\ \eqref{eq:adPII} with $\la=2m+2n+\mu+1$, $\rho=-n-\hf$ and $\ph=-n-\mu-\hf$,
whilst $X_{N}(z)$ and $Y_{N}(z)$ satisfy the discrete system
\begin{subequations}\label{sys:lem54}\begin{align} 
X_{N+1}+X_{N}&=\frac{4}{z}\frac{(2N+2m+2n+\mu+2)Y_{N}+\mu}{1-Y_{N}^2},\\
Y_{N}+Y_{N-1}&=\frac{4}{z}\frac{(2N+2m+2n+\mu+1)X_{N}-2n-\mu-1}{1-X_{N}^2}.
\end{align}\end{subequations}
}\end{lemma}
\begin{proof} From the transformations
\begin{align*}
\mathcal{R}_{1}\big(w_{m,n}^{(\mu)};m,-m-2n-1-\mu,\mu\big)&=\big({1/{v}_{m,n}^{(\mu+2n+1)}};m+n+1+\mu,-m-n-1,-\mu-2n-1\big),\\
\mathcal{R}_{1}\big(v_{m,n}^{(\mu)};m+n+1,-m+n-\mu,\mu\big)&=\big({1/{w}_{m+1,n}^{(\mu-2n-1)}};m+\mu+1,-m-1,-\mu+2n+1\big), 
\end{align*}
we obtain the sequence of solutions
\[ w_{m,n}^{(\mu)} \arr{\mathcal{R}_{1}}\frac{1}{{v}_{m,n}^{(\mu+2n+1)}} \arr{\mathcal{R}_{1}} w_{m+1,n}^{(\mu)} \arr{\mathcal{R}_{1}} \frac{1}{{v}_{m+1,n}^{(\mu+2n+1)}} \arr{\mathcal{R}_{1}} w_{m+2,n}^{(\mu)}\]
from which the result follows. 
We note that \[ \mathcal{R}_{1}^2\Big(w_{m,n}^{(\mu)} \Big)=w_{m+1,n}^{(\mu)},\qquad 
\mathcal{R}_{1}^2\Big(1/{v}_{m,n}^{(\mu)} \Big)=1/{v}_{m+1,n}^{(\mu)}.\]
\end{proof}

\begin{remarks}{\rm 
\begin{enumerate}[(i)]\item[]
\item An equivalent result can be obtained using the inverse transformation $\mathcal{R}_{1}^{-1}=\T_{-1,1,-1}$ since
\begin{align*}
\mathcal{R}_{1}^{-1}\big(w_{m,n}^{(\mu)};-m,m+2n+1+\mu,\mu\big)&=\big(v_{m,n}^{(\mu+2n+1)};-m-n-1,m+n+1+\mu,2n+1+\mu\big),\\
\mathcal{R}_{1}^{-1}\big(v_{m,n}^{(\mu)};-m-n-1,m-n+\mu,\mu\big)&=\big(w_{m+1,n}^{(\mu-2n-1)};-m-1,m+1+\mu,-2n-1+\mu\big).
\end{align*}
\item We note that if $\mu=-n-\hf$ then
\[ Q_{2N}(z)=\frac{w_{N+m,n}^{(-n-1/2)}(z)+1}{w_{N+m,n}^{(-n-1/2)}(z)-1},\qquad Q_{2N+1}(z)=\frac{1+v_{N+m,n}^{(n+1/2)}(z)}{1-v_{N+m,n}^{(n+1/2)}(z)},\]
are solutions to 
\[Q_{N+1}+Q_{N-1}=\frac{4}{z}\frac{(N+2m+n+\hf) Q_{N}-n-\hf}{1-Q_{N}^2},\] 
which is \dPII\ \eqref{eq:dPII} with $\zeta=4/z$, $\la=2m+n+\hf$ and $\rho=-n-\hf$.
Hence, although there is no alternating term, the structure of the even and odd terms is different.
\item {As remarked above, Kitaev, Law and McLeod \cite[Theorem 1.2]{refKLM} proved that rational solutions expressed in terms of generalised Laguerre polynomials are unique if $\mu\not\in\Integer$. Consequently the hierarchies can be determined through the parameters using induction. To obtain the rational solution $w_{m,n}^{(\ell)}$, with $\ell\in\Integer$, then we use
$w_{m,n}^{(\ell)}=\lim_{\mu\to\ell}w_{m,n}^{(\mu)}$ and similarly for $v_{m,n}^{(\ell)}$ and $u_{m,n}^{(\ell)}$.
In fact it is necessary to do this in order to evaluate the rational solution when $\mu\in\Integer$ using Maple.}
\end{enumerate}
}\end{remarks}

Alternatively, using the solution $v_{m,n}^{(\mu)}(z)$, for fixed $m$ and $n$, then we have the following Lemma.
\begin{lemma}
Consider the rational functions
\[ Q_{2N}(z)=X_{N}(z)=\frac{v_{m,N+n}^{(\mu)}+1}{v_{m,N+n}^{(\mu)}-1},\qquad Q_{2N+1}(z)=Y_{N}(z)=\frac{u_{N+n+1,m}^{(\mu+2m+1)}+1}{u_{N+n+1,m}^{(\mu+2m+1)}-1},\]
with $v_{m,n}^{(\mu)}(z)$ and $u_{m,n}^{(\mu)}(z)$ as given in Theorem \ref{thm:genLag},
then $Q_{N}(z)$ satisfies
\[Q_{N+1}+Q_{N-1}=\frac{4}{z}\frac{(N+2n+1-\mu) Q_{N}+\mu+(m+\hf)[1+(-1)^N]}{1-Q_{N}^2},\]
which is asymmetric \dPII\ \eqref{eq:adPII} with $\la=2n+1-\mu$, $\rho=m+\hf+\mu$ and $\ph=m+\hf$,
whilst $X_{N}(z)$ and $Y_{N}(z)$ satisfy the discrete system
\begin{subequations}\label{sys:lem56}\begin{align} 
X_{N+1}+X_{N}&=\frac{4}{z}\frac{(2N+2n+2-\mu)Y_{N}+\mu}{1-Y_{N}^2},\\
Y_{N}+Y_{N-1}&=\frac{4}{z}\frac{(2N+2n+1-\mu)X_{N}+\mu+2m+1}{1-X_{N}^2}.
\end{align}\end{subequations}
\end{lemma}

\begin{proof} From the transformations
{\begin{align*}
\mathcal{R}_{1}\big(v_{m,n}^{(\mu)};m+n+1,m-n+\mu,\mu\big)&=\big(u_{n+1,m}^{(\mu+2m+1)}; n+1,-n-1+\mu,2m+1+\mu\big),\\
\mathcal{R}_{1}\big(u_{m,n}^{(\mu)};m,-m-2n-1+\mu,\mu\big)&=\big(v_{n,m}^{(\mu-2n-1)};m+n+1,-m-n-1+\mu,-2n-1+\mu\big),
\end{align*}}
we obtain the sequence of solutions
\[v_{m,n}^{(\mu)} \arr{\mathcal{R}_{1}} u_{n+1,m}^{(\mu+2m+1)} \arr{\mathcal{R}_{1}} v_{m,n+1}^{(\mu)} \arr{\mathcal{R}_{1}} u_{n+2,m}^{(\mu+2m+1)} \arr{\mathcal{R}_{1}} v_{m,n+2}^{(\mu)}\]
from which the result follows. We note that \[ \mathcal{R}_{1}^2\Big(v_{m,n}^{(\mu)} \Big)=v_{m,n+1}^{(\mu)},\qquad 
\mathcal{R}_{1}^2\Big(u_{n,m}^{(\mu)} \Big)=u_{n+1,m}^{(\mu)}.\]
\end{proof}

\subsubsection{\label{dPn:GenLag}Second discrete equation}
The effect on the parameters of applying the transformations $\mathcal{R}_{2}$ \eqref{bt:R2} and $\mathcal{R}_{2}^2$ on a solution $w(z)$ of \PV\ with $\pms{\hf a^2}{-\hf b^2}{c}$ is given in Table \ref{table:R2params}.
\begin{table}[h]
{\[\begin{array}{|c|ccc|}
\hline
& \a_{1} & \b_{1} & \ga_{1} \\ \hline
\mathcal{R}_{2}(a,b,c) & \tfrac{1}{8}(a+b-c-1)^2 &-\tfrac{1}{8}(a+b+c+1)^2 &-a+b+1\\
\mathcal{R}_{2}(-a,b,c) & \tfrac{1}{8}(a-b+c+1)^2 &-\tfrac{1}{8}(a-b-c-1)^2 &a+b+1\\
\mathcal{R}_{2}(a,-b,c) & \tfrac{1}{8}(a-b-c-1)^2 &-\tfrac{1}{8}(a-b+c+1)^2 &-a-b+1\\
\mathcal{R}_{2}(-a,-b,c) & \tfrac{1}{8}(a+b+c+1)^2 &-\tfrac{1}{8}(a+b-c-1)^2 &a-b +1\\ \hline\hline
\mathcal{R}_{2}^2(a,b,c) & \hf(a-1)^2 &-\hf(b+1)^2 &c+2\\
\mathcal{R}_{2}^2(-a,b,c) & \hf(a+1)^2 &-\hf(b+1)^2 &c+2\\
\mathcal{R}_{2}^2(a,-b,c) & \hf(a-1)^2 &-\hf(b-1)^2 &c+2\\
\mathcal{R}_{2}^2(-a,-b,c) & \hf(a+1)^2 &-\hf(b-1)^2 &c+2 \\ \hline
\end{array}\]
\caption{\label{table:R2params}The effect on the parameters of applying the transformations $\mathcal{R}_{2}$ \eqref{bt:R2} and $\mathcal{R}_{2}^2$ on a solution $w(z)$ of \PV\ with $\pms{\hf a^2}{-\hf b^2}{c}$.}
}\end{table}

In the following example, we illustrate how using the \bt\ $\mathcal{R}_{2}$ \eqref{bt:R2} gives rise to a hierarchy of solutions of \eqref{eq:dPn} in terms of generalised Laguerre polynomials.
\begin{example}
Applying the \bt\ $\mathcal{R}_{2}$ \eqref{bt:R2}, with $(a,b,c)=(-3,\mu,\mu)$, to the solution
\[ v_{1,1}^{(\mu)}(z)=1-\frac{z\big[z^2-2\mu z+\mu(\mu+1)\big]\big[z^2-2\mu z+\mu(\mu-1)\big]}{(z-\mu)\big[z^4-4\mu z^3+6\mu^2 z^2-4\mu(\mu^2-1)z+\mu^2(\mu^2-1)\big]},\]
we obtain the sequence of rational solutions of \PV\ \eqref{eq:pv}
\[ v_{1,1}^{(\mu)}\arr{\mathcal{R}_{2}} u_{2,1}^{(\mu+4)}\arr{\mathcal{R}_{2}} v_{1,2}^{(\mu+2)}\arr{\mathcal{R}_{2}}u_{3,1}^{(\mu+6)}\arr{\mathcal{R}_{2}} v_{1,3}^{(\mu+4)}\]
Hence we obtain the solutions of equation \eqref{eq:dPn}
\begin{align*}
q_{0}(z)&=\frac{v_{1,1}^{(\mu)}+1}{v_{1,1}^{(\mu)}-1}
= -1+\frac{2\mu}{z}+\Ln{z^2-2\mu+\mu(\mu+1)}{z^2-2\mu+\mu(\mu-1)},\\
q_{1}(z)&=\frac{u_{2,1}^{(\mu+4)}+1}{u_{2,1}^{(\mu+4)}-1}
= 1 +2\Ln{z^4-4(\mu+1)z^3+6(\mu+1)^2z^2-4\mu(\mu+1)(\mu+2)z+\mu(\mu+1)^2(\mu+2)}{z^2-2\mu+\mu(\mu+1)}.
 \end{align*}
Consequently, if we define the rational functions 
\begin{align*} 
q_{2n}(z)&=x_{n}(z)=\frac{v_{1,n+1}^{(2n+\mu)}(z)+1}{v_{1,n+1}^{(2n+\mu)}(z)-1}
=-1+\frac{2(\mu+2n)}{z}+2\Ln{\Tmn{0,n+2}{\mu-3}}{\Tmn{1,n+1}{\mu-3}(z)},\\
 q_{2n+1}(z)&=y_{n}(z)=\frac{u_{n+2,1}^{(2n+\mu+4)}(z)+1}{u_{n+2,1}^{(2n+\mu+4)} (z)-1}
=1+2\Ln{\Tmn{1,n+2}{\mu-3}}{\Tmn{0,n+2}{\mu-3}(z)},
 \end{align*}
then $q_{n}(z)$ satisfies
\[\frac{n+\mu+2}{q_{n+1}+q_{n}}+\frac{n+\mu+1}{q_{n}+q_{n-1}}=\hf z-\frac{\big\{n+\mu+\tfrac{3}{2}[1+(-1)^n]\big\}q_{n}+3-\mu}{1-q_{n}^2},\]
which is equation \eqref{eq:dPn} with $\la=\mu+\tfrac{3}{2}$, $\rho=3-\mu$ and $\ph=\tfrac{3}{2}$,
whilst $x_{n}(z)$ and $y_{n}(z)$ satisfy the discrete system
\begin{align*} 
\frac{2n+\mu+3}{x_{n+1}+y_{n}}+\frac{2n+\mu+2}{x_{n}+y_{n}}&=\hf z-\frac{(2n+\mu+1)y_{n}+3-\mu}{1-y_{n}^2},\\
\frac{2n+\mu+2}{x_{n}+y_{n}}+\frac{2n+\mu+1}{x_{n}+y_{n-1}}&=\hf z-\frac{(2n+\mu+3)x_{n}+3-\mu}{1-x_{n}^2}.
\end{align*}
\end{example}
\comment{Alternatively, to derive a sequence of solutions we can use the \bt\
\beq\mathcal{Q}_{2}(q)=\frac{\mathcal{R}_{2}\big(\frac{q+1}{q-1}\big)+1}{\mathcal{R}_{2}\big(\frac{q+1}{q-1}\big)-1}=-q+\frac{2(a-b-c-1)(q^2-1)}{2zq'-zq^2+2(a-b)q+2a+2b+z}.\label{bt:Q2} \eeq}

More generally, using the solution $v_{m,n}^{(\mu)}(z)$, for fixed $m$ and $n$, then we have the following Lemma. 
\begin{lemma}
Consider the rational functions 
\[ Q_{2N}(z)=X_{N}(z)=\frac{v_{m,N+n}^{(2N+\mu)}(z)+1}{v_{m,N+n}^{(2N+\mu)}(z)-1},\qquad Q_{2N+1}(z)=Y_{N}(z)=\frac{u_{N+n+1,m}^{(2N+\mu+2m+2)}(z)+1}{u_{N+n+1,m}^{(2N+\mu+2m+2)} (z)-1}, \]
with $v_{m,n}^{(\mu)}(z)$ and $u_{m,n}^{(\mu)}(z)$ as given in Theorem \ref{thm:genLag},
then $Q_{N}(z)$ satisfies
\[\frac{N+m+\mu+1}{Q_{N+1}+Q_{N}}+\frac{N+m+\mu}{Q_{N}+Q_{N-1}}=\hfz-\frac{\big\{N+[1+(-1)^N](m+\hf)+\mu\big\}Q_{N}+2n+1-\mu}{1-Q_{N}^2},\]
which is equation \eqref{eq:dPn} with $\la=m+\mu-\hf$, $\rho=2n-\mu+1$ and $\ph=m+\hf$,
whilst $X_{N}(z)$ and $Y_{N}(z)$ satisfy the discrete system
\begin{subequations}\label{sys:lem57}\begin{align} 
\frac{2N+m+\mu+2}{X_{N+1}+Y_{N}}+\frac{2N+m+\mu+1}{X_{N}+Y_{N}}&=\hfz-\frac{(2N+\mu+1)Y_{N}+2n+1-\mu}{1-Y_{N}^2},\\
\frac{2N+m+\mu+1}{X_{N}+Y_{N}}+\frac{2N+m+\mu}{X_{N}+Y_{N-1}}&=\hfz-\frac{(2N+2m+\mu+1)X_{N}+2n+1-\mu}{1-X_{N}^2}.
\end{align}\end{subequations}
\end{lemma}
\begin{proof} From the transformations
\begin{align*}
\mathcal{R}_{2}\big(v_{m,n}^{(\mu)};-m-n-1,m-n+\mu,\mu\big)&=\big(u_{n+1,m}^{(\mu+2m+2)};-n-1,-n+\mu,2m+2+\mu\big),\\
\mathcal{R}_{2}\big(u_{n,m}^{(\mu+2m+2)};-n,-n+1+\mu,2m+2+\mu\big)&=\big(v_{m,n}^{(\mu+2)};m+n+1,-m-n-1+\mu,2m-1+\mu\big),
\end{align*}
we obtain the sequence of solutions
\[ v_{m,n}^{(\mu)} \arr{\mathcal{R}_{2}} u_{n+1,m}^{(\mu+2m+2)} \arr{\mathcal{R}_{2}} v_{m,n+1}^{(\mu+2)} \arr{\mathcal{R}_{2}} u_{n+2,m}^{(\mu+2m+4)} \arr{\mathcal{R}_{2}} v_{m,n+2}^{(\mu+4)}\] 
from which the result follows.
We note that \[ \mathcal{R}_{2}^2\Big(v_{m,n}^{(\mu)} \Big)=v_{m,n+1}^{(\mu+2)},\qquad 
\mathcal{R}_{2}^2\Big(u_{n,m}^{(\mu)} \Big)=u_{n+1,m}^{(\mu+2)}.\]
\end{proof}

Alternatively, using the solution $w_{m,n}^{(\mu)}(z)$, for fixed $m$ and $n$, then we have the following Lemma.
\begin{lemma}
Consider the rational functions 
\[ Q_{2N}(z)=X_{N}(z)=\frac{w_{N+m,n}^{(\mu-2N)}(z)+1}{w_{N+m,n}^{(\mu-2N)}(z)-1},\qquad Q_{2N+1}(z)=Y_{N}(z)=\frac{v_{N+m,n}^{(\mu+2n-2N)}(z)+1}{v_{N+m,n}^{(\mu+2n-2N)} (z)-1},\] 
with $w_{m,n}^{(\mu)}(z)$ and $v_{m,n}^{(\mu)}(z)$ as given in Theorem \ref{thm:genLag},
then $Q_{N}(z)$ satisfies
\[\frac{N-n-\mu}{Q_{N+1}+Q_{N}}+\frac{N-n-\mu-1}{Q_{N}+Q_{N-1}}+\hfz+\frac{\big\{N-\mu-(n+\hf)[1+(-1)^N] \big\}Q_{N}+2 m+2 n+1+\mu}{1-Q_{N}^{2}}=0,\]
which is equation \eqref{eq:dPn} with $\la=-n-\mu-\hf$, $\rho=2m+2n+\mu+1$ and $\ph=-n-\hf$,
whilst $X_{N}(z)$ and $Y_{N}(z)$ satisfy the discrete system
\begin{subequations}\label{sys:lem58}\begin{align} 
&\frac{2 N-n-\mu+1}{X_{N+1}+Y_{N}}+\frac{2 N-n-\mu}{X_{N}+Y_{N}}+\hfz+\frac{(2 N-\mu+1) Y_{N}+2 m+2 n+1+\mu}{1-Y_{N}^{2}}=0,\\
&\frac{2 N-n-\mu}{X_{N}+Y_{N}}+\frac{2 N-n-\mu-1}{X_{N}+Y_{N-1}}+\hfz+\frac{(2 N-2 n-\mu-1 ) X_{N}+2 m+2 n+1+\mu}{1-X_{N}^{2}}=0.
\end{align}\end{subequations}
\end{lemma}
\begin{proof} From the transformations
\begin{align*}
\mathcal{R}_{2}^{-1}\big(w_{m,n}^{(\mu)};m,m+2n+1+\mu,\mu\big)&=\big(v_{m,n}^{(\mu+2n)}; m+n+1,m+n+\mu,2n+\mu\big),\\
\mathcal{R}_{2}^{-1}\big(v_{m,n}^{(\mu+2n)}; m+n+1,m+n+\mu,2n+\mu\big)&=\big(w_{m+1,n}^{(\mu-2)};m,m+2n+1+\mu,\mu\big)
\end{align*}
we obtain the sequence of solutions
\[ w_{m,n}^{(\mu)} \arr{\mathcal{R}_{2}^{-1}} v_{m,n}^{(\mu+2n)} \arr{\mathcal{R}_{2}^{-1}} w_{m+1,n}^{(\mu-2)} \arr{\mathcal{R}_{2}^{-1}} v_{m,n}^{(\mu+2n-2)} \arr{\mathcal{R}_{2}^{-1}} w_{m+2,n}^{(\mu-4)}\]
from which the result follows.
\end{proof}
\begin{remark}{{\rm From the transformations
\begin{align*}
\mathcal{R}_{2}^{-1}\big(w_{m,n}^{(\mu)};m,-m-2n-1-\mu,\mu\big)&=\big(1/{u}_{n+1,m}^{(\mu+2m+2n+2)};-n-\mu,-n-1,-2m-2n-2-\mu\big),\\
\mathcal{R}_{2}^{-1}\big(1/{u}_{n,m}^{(\mu)}; 2m+n+1-\mu,-n,-\mu\big)&=\big(w_{m+1,n}^{(\mu-2m-2n-2)};m+1,m-\mu,-2m-2n-2+\mu\big),
\end{align*}
we obtain the sequence of solutions
\[ w_{m,n}^{(\mu)} \arr{\mathcal{R}_{2}^{-1}} \frac{1}{{u}_{n+1,m}^{(\mu+2m+2n+2)}} \arr{\mathcal{R}_{2}^{-1}} w_{m+1,n+1}^{(\mu-2)}
\arr{\mathcal{R}_{2}^{-1}} \frac{1}{{u}_{n+2,m+1}^{(\mu+2m+2n+4)}} \arr{\mathcal{R}_{2}^{-1}} w_{m+2,n+2}^{(\mu-4)}\]
}}\end{remark}


\subsubsection{\label{atdPI:GenLag}Third discrete equation: ternary discrete \PI}
The effect on the parameters of applying the transformations $\mathcal{R}_{3}$ \eqref{bt:R3}, $\mathcal{R}_{3}^2$ and $\mathcal{R}_{3}^3$ on a solution $w(z)$ of \PV\ with $\pms{\hf a^2}{-\hf b^2}{c}$ is given in Table \ref{table:R3params}.
\begin{table}[h]
{\[\begin{array}{|c|ccc|}
\hline
& \a_{1} & \b_{1} & \ga_{1} \\ \hline
\mathcal{R}_{3}(a,b,c) & \tfrac{1}{8}(a-b+c+1)^2 &-\tfrac{1}{8}(a-b-c+1)^2 &a+b\\
\mathcal{R}_{3}(-a,b,c) & \tfrac{1}{8}(a+b-c-1)^2 &-\tfrac{1}{8}(a+b+c-1)^2 &-a+b\\
\mathcal{R}_{3}(a,-b,c) & \tfrac{1}{8}(a+b+c+1)^2 &-\tfrac{1}{8}(a+b-c+1)^2 &a-b\\
\mathcal{R}_{3}(-a,-b,c) & \tfrac{1}{8}(a-b-c-1)^2 &-\tfrac{1}{8}(a-b+c-1)^2 &-a-b \\ 
\hline\hline
\mathcal{R}_{3}^2(a,b,c) & \tfrac{1}{8}(a+b+c+1)^2 &-\tfrac{1}{8}(a+b-c-1)^2 &a-b+1\\
\mathcal{R}_{3}^2(-a,b,c) & \tfrac{1}{8}(a-b-c-1)^2 &-\tfrac{1}{8}(a-b+c+1)^2 &-a-b+1\\
\mathcal{R}_{3}^2(a,-b,c) & \tfrac{1}{8}(a-b+c+1)^2 &-\tfrac{1}{8}(a-b-c-1)^2&a+b+1\\
\mathcal{R}_{3}^2(-a,-b,c) & \tfrac{1}{8}(a+b-c-1)^2 &-\tfrac{1}{8}(a+b+c+1)^2 &-a+b+1 \\ 
\hline\hline
\mathcal{R}_{3}^3(a,b,c) & \hf(a+1)^2 &-\hf b^2 &c+1\\
\mathcal{R}_{3}^3(-a,b,c) & \hf(a-1)^2 &-\hf b^2 &c+1\\
\mathcal{R}_{3}^3(a,-b,c) & \hf(a+1)^2 &-\hf b^2 &c+1\\
\mathcal{R}_{3}^3(-a,-b,c) & \hf(a-1)^2 &-\hf b^2 &c+1 \\ \hline
\end{array}\]
\caption{\label{table:R3params}The effect on the parameters of applying the transformation $\mathcal{R}_{3}$ \eqref{bt:R3}, $\mathcal{R}_{3}^2$ and $\mathcal{R}_{3}^3$ on a solution $w(z)$ of \PV\ with $\pms{\hf a^2}{-\hf b^2}{c}$.}
}\end{table}

In the following example, we illustrate how using the \bt\ $\mathcal{R}_{3}$ \eqref{bt:R3} gives rise to a hierarchy of solutions of ternary \dPI\ \eqref{eq:tdPI} in terms of generalised Laguerre polynomials.
\begin{example}
Applying the \bt\ $\mathcal{R}_{3}$ \eqref{bt:R3}, with $(a,b,c)=(\mu+1,-1,-\mu+3)$, to the solution
\[\frac{1}{{w}_{1,1}^{(\mu-3)}(z)}=-\frac{z^2-2\mu z+\mu(\mu-1)}{(z-\mu+1)\big[z^2-2\mu z+\mu(\mu+1)\big]}, \]
we obtain the sequence of rational solutions of \PV\ \eqref{eq:pv}
\[ \frac{1}{{w}_{1,1}^{(\mu-3)}} \arr{\mathcal{R}_{3}}{v}_{1,1}^{(\mu)} \arr{\mathcal{R}_{3}} u_{2,1}^{(\mu+3)} \arr{\mathcal{R}_{3}}\frac{1}{{w}_{1,2}^{(\mu-4)}}
\arr{\mathcal{R}_{3}}{v}_{1,2}^{(\mu+1)} \arr{\mathcal{R}_{3}} u_{3,1}^{(\mu+4)}\]
so that the effect of applying $\mathcal{R}_{3}$ three times is
\[ \mathcal{R}_{3}^3\Big(1/{{w}_{1,1}^{(\mu-3)}}\Big)= 1/{{w}_{1,1}^{(\mu-4)}},\qquad
\mathcal{R}_{3}^3\Big( {v}_{1,1}^{(\mu)}\Big)= {v}_{1,2}^{(\mu+1)},\qquad
\mathcal{R}_{3}^3\Big( {u}_{2,1}^{(\mu+3)}\Big)= {u}_{3,1}^{(\mu+4)}.\]
Hence we obtain the rational solutions of the ternary \dPI\ \eqref{eq:tdPI}
\begin{align*} 
x_{0}(z)&=\frac{w_{1,1}^{(\mu-3)}(z)}{1-w_{1,1}^{(\mu-3)}(z)} =-1+
\Ln{z^2-2(\mu-1)z+\mu(\mu-1)}{z-\mu},\\
x_{1}(z)&= \frac{1}{v_{1,1}^{(\mu)}(z)-1}=-1+\frac{\mu}{z}
+ \Ln{z^2-2\mu z +\mu(\mu-1)}{z^2-2\mu z +\mu(\mu+1)}, \\ 
x_{2}(z)&= \frac{1}{u_{2,1}^{(\mu+3)}(z)-1}
=\Ln{z^4-4\mu z^3+6\mu^2z^2-4\mu(\mu^2-1)z+\mu^2(\mu^2-1)z}{z^2-2(\mu+1)z +\mu(\mu-1)}.
\end{align*}
Consequently, if we define the rational functions
\begin{align*} 
x_{3n}(z)=X_{n}(z)&=\frac{w_{1,n+1}^{(\mu-n-3)}(z)}{1-w_{1,n+1}^{(\mu-n-3)}(z)} =-1+\Ln{\Tmn{0,n+2}{\mu-n-4}(z)}{\Tmn{0,n+1}{\mu-n-2}(z)},\\
x_{3n+1}(z)=Y_{n}(z)&= \frac{1}{v_{1,n+1}^{(\mu+n)}(z)-1}=-1+\frac{\mu+n}{z}+ \Ln{\Tmn{0,n+2}{\mu-n-3}(z)}{\Tmn{1,n+1}{\mu-n-3}(z)},\\
x_{3n+2}(z)=Z_{n}(z)&= \frac{1}{u_{n+2,1}^{(\mu+n+3)}(z)-1}=\Ln{\Tmn{1,n+2}{\mu-n-4}(z)}{\Tmn{0,n+2}{\mu-n-4}(z)},
\end{align*}
then $x_{n}(z)$ satisfies ternary \dPI\ 
\[ x_{n}(x_{n+1}+x_{n-1}+1) +\frac{a_{n}}{z}=0,\]
with \[a_{n}=\tfrac{1}{3}(n+\mu+5)+\tfrac{2}{3}(\mu-1)\cos \!\big(\tfrac{2}{3}\pi n\big)+\tfrac{4}{9}\sqrt{3}\,\sin \!\big(\tfrac{2}{3}\pi n\big),\]
whilst $X_{n}(z)$, $Y_{n}(z)$ and $Z_{n}(z)$ satisfy the discrete system
\begin{align*} 
&X_{n}(Y_{n}+Z_{n-1}+1)+\frac{n+\mu+1}{z}=0,\\
&Y_{n}(Z_{n}+X_{n}+1)+\frac{n+3}{z}=0,\\
&Z_{n}(X_{n+1}+Y_{n}+1)+\frac{n+2}{z}=0.
\end{align*}
\end{example}
More generally, using the solution $w_{m,n}^{(\mu)}(z)$, for fixed $m$ and $n$, as the initial solution then we have the following Lemma.
\begin{lemma}\label{lem:512}
Consider the rational functions
\beq X_{N}(z)=\frac{w_{m,N+n}^{(\mu-N)}(z)}{1-w_{m,N+n}^{(\mu-N)}(z)},\quad
Y_{N}(z)= \frac{1}{v_{m,N+n}^{(\mu+N+2n+1)}(z)-1},\quad
Z_{N}(z)= \frac{1}{u_{N+n+1,m}^{(\mu+N+2m+2n+2)}(z)-1}, 
\eeq
with $w_{m,n}^{(\mu)}(z)$, $v_{m,n}^{(\mu)}(z)$ and $u_{m,n}^{(\mu)}(z)$ as given in Theorem \ref{thm:genLag},
then $X_{N}(z)$, $Y_{N}(z)$ and $Z_{N}(z)$ satisfy the discrete system
\begin{subequations}\label{sys:lem511}\begin{align} 
&X_{N}(Y_{N}+Z_{N-1}+1)+\frac{N+m+2n+\mu+1}{z}=0,\\
&Y_{N}(Z_{N}+X_{N}+1)+\frac{N+m+n+1}{z}=0,\\
&Z_{N}(X_{N+1}+Y_{N}+1)+\frac{N+n+1}{z}=0,
\end{align}\end{subequations}
which is the system \eqref{sys:tertdPI} with
$a_{3N}=N+m+2n+\mu+1$, $a_{3N+1}=N+m+n+1$ and $a_{3N+2}=N+n+1$, i.e.\ \eqref{eq:abcn3a} with
$\la=\tfrac{1}{3}(m+2n+2\mu+1)$, $\rho=\tfrac{1}{3}(2m+4n+\mu+2)$ and $\ph=m+\tfrac{1}{3}$.
\end{lemma}
\begin{proof}From the transformations
\begin{align*}
&\mathcal{R}_{3}\big(1/{w}_{m,n}^{(\mu)};m+2n+1+\mu,-m,-\mu\big)=\big(v_{m,n}^{(\mu+2n+1)}; m+n+1,m+n+\mu+1,2n+1+\mu\big),\\ 
&\mathcal{R}_{3}\big(v_{m,n}^{(\mu+2n+1)};m+n+1,m+n+\mu+1,2n+1+\mu\big)=\big(u_{n+1,m}^{(\mu+2m+2n+2)};n+1,-n-\mu,2m+2n+2+\mu\big),\\
&\mathcal{R}_{3}\big(u_{n+1,m}^{(\mu+2m+2n+2)};n+1,-n-\mu,2m+2n+2+\mu\big)=\big(1/{w}_{m,n+1}^{(\mu-1)};m+2n+\mu+2,-m,\mu-1\big),
\end{align*}
we obtain the sequence of solutions
\[ \frac{1}{{w}_{m,n}^{(\mu)}} \arr{\mathcal{R}_{3}} v_{m,n}^{(\mu+2n+1)} \arr{\mathcal{R}_{3}} u_{n+1,m}^{(\mu+2m+2n+2)} \arr{\mathcal{R}_{3}} \frac{1}{{w}_{m,n+1}^{(\mu-1)}} \arr{\mathcal{R}_{3}} v_{m,n+1}^{(\mu+2n+2)} \]
from which the result follows. We note that \[ \mathcal{R}_{3}^3\Big(1/w_{m,n}^{(\mu)}\Big)=1/w_{m,n+1}^{(\mu-1)},\qquad 
\mathcal{R}_{3}^3\Big({v}_{m,n}^{(\mu)}\Big)={v}_{m,n+1}^{(\mu+1)},\qquad
\mathcal{R}_{3}^3\Big({u}_{n,m}^{(\mu)}\Big)={u}_{n+1,m}^{(\mu+1)}\]
\end{proof}

\begin{lemma}\label{lem:513}
Consider the rational functions 
\beq X_{N}(z)=\frac{1}{v_{m,N+n}^{(N+\mu)}(z)-1},\quad Y_{N}(z)=\frac{1}{u_{N+n+1,m}^{(N+\mu+2m+1)}(z)-1},\quad 
Z_{N}(z)=\frac{w_{m,N+n+1}^{(-N+\mu-2n-2)}(z)}{1-w_{m,N+n+1}^{(-N+\mu-2n-2)}(z)},\eeq
with $w_{m,n}^{(\mu)}(z)$, $v_{m,n}^{(\mu)}(z)$ and $u_{m,n}^{(\mu)}(z)$ as given in Theorem \ref{thm:genLag},
then $X_{N}(z)$, $Y_{N}(z)$ and $Z_{N}(z)$ satisfy the discrete system
\begin{subequations}\label{sys:lem510}\begin{align} 
&X_{N}(Y_{N}+Z_{N-1}+1)+\frac{N+m+n+1}{z}=0,\\
&Y_{N}(Z_{N}+X_{N}+1)+\frac{N+n+1}{z}=0,\\
&Z_{N}(X_{N+1}+Y_{N}+1)+\frac{N+m+\mu+1}{z}=0,
\end{align}\end{subequations}
which is the system \eqref{sys:tertdPI} with
$a_{3N}=N+m+n+1$, $a_{3N+1}=N+n+1$ and $a_{3N+2}=N+m+\mu+1$, i.e.\ \eqref{eq:abcn3a} with
$\la=\tfrac{1}{3}(m+n-\mu+1)$, $\rho=\tfrac{1}{3}(2m+2n+\mu+2)$ and $\ph=-m+n-\mu+\tfrac{1}{3}$.
\end{lemma}
\begin{proof} From the transformations
\begin{align*}
\mathcal{R}_{3}\big(v_{m,n}^{(\mu)};m+n+1,m-n+\mu,\mu\big)&=\big(u_{n+1,m}^{(\mu+2m+1)};n+1,n+1-\mu,2m+1+\mu\big),\\
\mathcal{R}_{3}\big(u_{m,n}^{(\mu)};m,m+2n+1-\mu,\mu\big)&=\big(1/{w}_{n,m}^{(\mu-2m-2n-1)};-n+\mu,-n,2m+2n+1-\mu\big),\\
\mathcal{R}_{3}\big(1/{w}_{m,n}^{(\mu)};m+2n+1+\mu,-m,-\mu\big)&=\big(v_{m,n}^{(\mu+2n+1)}; m+n+1,m+n+\mu+1,2n+1+\mu\big),
\end{align*}
we obtain the sequence of solutions
\[ v_{m,n}^{(\mu)} \arr{\mathcal{R}_{3}} u_{n+1,m}^{(\mu+2m+1)} \arr{\mathcal{R}_{3}} \frac{1}{{w}_{m,n+1}^{(\mu-2n-2)}} \arr{\mathcal{R}_{3}} v_{m,n+1}^{(\mu+1)}
\arr{\mathcal{R}_{3}} u_{n+2,m}^{(\mu+2m+2)}\]
from which the result follows.
We note that \[ \mathcal{R}_{3}^3\Big(1/w_{m,n}^{(\mu)} \Big)=1/w_{m,n+1}^{(\mu-1)},\qquad 
\mathcal{R}_{3}^3\Big({v}_{m,n}^{(\mu)} \Big)={v}_{m,n+1}^{(\mu+1)},\qquad
\mathcal{R}_{3}^3\Big({u}_{n,m}^{(\mu)} \Big)={u}_{n+1,m}^{(\mu+1)}.\]
\end{proof}

\comment{\begin{remark}{\rm Equivalent results can be obtained from the sequence of solutions
\[ {\wt{w}_{m,n}^{(\mu)}} \arr{\mathcal{R}_{3}} v_{m,n}^{(\mu+2n+1)} \arr{\mathcal{R}_{3}} u_{n+1,m}^{(\mu+2m+2n+2)} \arr{\mathcal{R}_{3}} {\wt{w}_{m,n+1}^{(\mu-1)}} \]
where
\[ a_{3N}=N+m+2n+1+\mu,\qquad a_{3N+1}=N+m+n+1,\qquad a_{3N+2}=N+n+1,\]
and the sequence of solutions
\[ {\wt{w}_{m,n}^{(\mu)}} \arr{\mathcal{R}_{3}} u_{n+1,m-1}^{(\mu+2m+2n+1)} \arr{\mathcal{R}_{3}} v_{m-1,n+1}^{(\mu+2n+2)} \arr{\mathcal{R}_{3}} {\wt{w}_{m,n+1}^{(\mu-1)}} \]
where
\[ a_{3N}=N+m+2n+1+\mu,\qquad a_{3N+1}=N+n+1,\qquad a_{3N+2}=n+m+n+1.\]
}\end{remark}}
\begin{remarks}
\begin{enumerate}[(i)]\item[]
\item There are no solutions in terms of generalised Laguerre polynomials in the symmetric case when $\la=\ph=0$.
\item Lemmas \ref{lem:512} and \ref{lem:513} imply that there are hierarchies rational solutions of the fourth discrete equation \eqref{eq:tdPInew} in terms of generalised Laguerre polynomials.
\end{enumerate}
\end{remarks}

\subsubsection{Partial difference systems}
In this subsection we derive some partial difference systems which have solutions in terms of generalised Laguerre polynomials.

\begin{lemma} If we define
\beq \label{def:WVULag}\mathcal{W}_{m,n}^{(\mu)}=\frac{{w}_{m,n}^{(\mu)}+1}{{w}_{m,n}^{(\mu)}-1}, \qquad 
\mathcal{V}_{m,n}^{(\mu)}=\frac{{v}_{m,n}^{(\mu)}+1}{{v}_{m,n}^{(\mu)}-1}, \qquad 
\mathcal{U}_{m,n}^{(\mu)}=\frac{{u}_{m,n}^{(\mu)}+1}{{u}_{m,n}^{(\mu)}-1},\eeq
with $w_{m,n}^{(\mu)}(z)$, $v_{m,n}^{(\mu)}(z)$ and $u_{m,n}^{(\mu)}(z)$ as given in Theorem \ref{thm:genLag},
then they satisfy the partial difference systems
\begin{subequations}\begin{align}
\mathcal{W}_{m+1,n}^{(\mu)}+\mathcal{W}_{m,n}^{(\mu)} &+\frac{4}{z}\frac{(2m+2n+\mu+2)\mathcal{V}_{m,n}^{(\mu+2n+1)}-\mu}{1-\big(\mathcal{V}_{m,n}^{(\mu+2n+1)}\big)^{\!2}}=0,\label{eqn:sysref}\\
\mathcal{V}_{m,n}^{(\mu+2n+1)}+\mathcal{V}_{m-1,n}^{(\mu+2n+1)}&+\frac{4}{z}\frac{(2m+2n+\mu+1)\mathcal{W}_{m,n}^{(\mu)}-2n-\mu-1}{1-\big(\mathcal{W}_{m,n}^{(\mu)}\big)^{\!2}}=0,
\end{align}\end{subequations}
and
\begin{subequations}\begin{align}
\mathcal{V}_{m,n+1}^{(\mu)}+\mathcal{V}_{m,n}^{(\mu)}&=\frac{4}{z}\frac{(2n+2-\mu)\,\mathcal{U}_{n+1,m}^{(\mu+2m+1)}-\mu}{1-\big(\mathcal{U}_{n+1,m}^{(\mu+2m+1)}\big)^{\!2}},\\
\mathcal{U}_{n+1,m}^{(\mu+2m+1)}+\mathcal{U}_{n,m}^{(\mu+2m+1)}&=\frac{4}{z}\frac{(2n+1-\mu)\mathcal{V}_{m,n}^{(\mu)}+2m+\mu+1}{1-\big(\mathcal{V}_{m,n}^{(\mu)} \big)^{\!2}}.
\end{align}\end{subequations}
\end{lemma}
\begin{proof}These follow from setting $N=0$ in \eqref{sys:lem54} and \eqref{sys:lem56}, respectively.
\end{proof}

\begin{lemma}
If $\mathcal{W}_{m,n}^{(\mu)}$, $\mathcal{V}_{m,n}^{(\mu)}$ and $\mathcal{U}_{m,n}^{(\mu)}$ are given by \eqref{def:WVULag}, then they satisfy the partial difference systems
\begin{subequations}\label{lem514:sys1}\begin{align}
\mathcal{W}_{m+1,n}^{(\mu)}+\mathcal{W}_{m,n}^{(\mu)} &+\frac{4}{z}\frac{(2n+\mu+1)\,\mathcal{V}_{m,n}^{(\mu+2n+1)}-\mu}{1-\big(\mathcal{V}_{m,n}^{(\mu+2n+1)}\big)^{\!2}}=0,\label{lem514:sys1a}\\
\mathcal{V}_{m,n}^{(\mu+2n+1)}+\mathcal{V}_{m-1,n}^{(\mu+2n+1)}&+\frac{4}{z}\frac{(2m+2n+\mu+1)\mathcal{W}_{m,n}^{(\mu)}-(2n+\mu+1)}{1-\big(\mathcal{W}_{m,n}^{(\mu)}\big)^{\!2}}=0,
\end{align}\end{subequations}
and
\begin{subequations}\label{lem514:sys2}\begin{align}
\mathcal{W}_{m+1,n-1}^{(\mu)}+\mathcal{W}_{m,n}^{(\mu)} &+\frac{4}{z}\frac{(2n+\mu)\,\mathcal{U}_{n,m}^{(\mu+2m+2n+1)}-\mu}{1-\big(\mathcal{U}_{n,m}^{(\mu+2m+2n+1)}\big)^{\!2}}=0,\\
\mathcal{U}_{n+1,m-1}^{(\mu+2m+2n+1)}+\mathcal{U}_{n,m}^{(\mu+2m+2n+1)}&+\frac{4}{z}\frac{(2n+\mu+1)\mathcal{W}_{m,n}^{(\mu)}-(2m+2n+\mu+1)}{1-\big(\mathcal{W}_{m,n}^{(\mu)}\big)^{\!2}}=0.
\end{align}\end{subequations}
\end{lemma}
\begin{proof}
From the transformations 
\begin{align*}
\mathcal{R}_{1}\big(w_{m,n}^{(\mu)};-m,m+2n+\mu+1,\mu\big) &= \big(v_{m-1,n}^{(\mu+2n+1)};-m-n,m+n+\mu,\mu+2n+1\big),\\
\mathcal{R}_{1}\big(v_{m,n}^{(\mu)};-m-n-1,m-n+\mu,\mu\big)&=\big(w_{m,n}^{(\mu-2n-1)};-m,m+\mu,\mu-2n-1\big),
\end{align*}
we obtain the sequence
\[ w_{m,n}^{(\mu)} \arr{\mathcal{R}_{1}} v_{m-1,n}^{(\mu+2n+1)} \arr{\mathcal{R}_{1}} w_{m-1,n}^{(\mu)} \arr{\mathcal{R}_{1}} v_{m-2,n}^{(\mu+2n+1)} \arr{\mathcal{R}_{1}} w_{m-2,n}^{(\mu)} \]
from which the system \eqref{lem514:sys1} follows.

From the transformations 
\begin{align*}
\mathcal{R}_{1}\big(w_{m,n}^{(\mu)};m,m+2n+\mu+1,\mu\big) &= \big(u_{n,m}^{(\mu+2m+2n+1)};-n,n+\mu,\mu+2m+2n+1\big),\\
\mathcal{R}_{1}\big(u_{m,n}^{(\mu)};-m,-m-2n-1+\mu,\mu\big)&=\big(w_{n+1,m-1}^{(\mu-2m-2n-1)};n+1,-n-1+\mu,\mu-2m-2n-1\big),
\end{align*}
we obtain the sequence
\[ w_{m,n}^{(\mu)} \arr{\mathcal{R}_{1}} u_{n,m}^{(\mu+2m+2n+1)} \arr{\mathcal{R}_{1}} w_{m+1,n-1}^{(\mu)} \arr{\mathcal{R}_{1}} u_{n-1,m+1}^{(\mu+2m+2n+1)} \arr{\mathcal{R}_{1}} w_{m+2,n-2}^{(\mu)}\]
from which the system \eqref{lem514:sys2} follows.
\end{proof}

\begin{lemma}
If $\mathcal{W}_{m,n}^{(\mu)}$, $\mathcal{V}_{m,n}^{(\mu)}$ and $\mathcal{U}_{m,n}^{(\mu)}$ are given by \eqref{def:WVULag}, then they satisfy the partial difference systems
\begin{subequations}\begin{align}
\frac{m+\mu+1}{\mathcal{V}_{m,n}^{(\mu+1)}+\mathcal{U}_{n,m}^{(\mu+2m+1)}}+\frac{m+\mu}{\mathcal{V}_{m,n-1}^{(\mu-1)}+\mathcal{U}_{n,m}^{(\mu+2m+1)}}&=\hfz-\frac{\mu\,\mathcal{U}_{n,m}^{(\mu+2m+1)}+2n-\mu}{1-\big(\mathcal{U}_{n,m}^{(\mu+2m+1)}\big)^{\!2}},\\
\frac{m+\mu}{\mathcal{V}_{m,n-1}^{(\mu-1)}+\mathcal{U}_{n,m}^{(\mu+2m+1)}}+\frac{m+\mu-1}{\mathcal{V}_{m,n-1}^{(\mu-1)}+\mathcal{U}_{n-1,m}^{(\mu+2m-1)}}&=\hfz-\frac{(2m+\mu)\mathcal{V}_{m,n-1}^{(\mu-1)}+2n-\mu}{1-\big(\mathcal{V}_{m,n-1}^{(\mu-1)} \big)^{\!2}},
\end{align}\end{subequations}
and
\begin{subequations}\begin{align}
\frac{n+\mu-1}{\mathcal{W}_{m+1,n}^{(\mu-2)}+\mathcal{V}_{m,n}^{(\mu+2n)}}+\frac{n+\mu}{\mathcal{W}_{m,n}^{(\mu)}+\mathcal{V}_{m,n}^{(\mu+2n)}}&=\hfz-\frac{(\mu-1) \mathcal{V}_{m,n}^{(\mu+2n)}-(2 m+2 n+1+\mu)}{1-\big(\mathcal{V}_{m,n}^{(\mu+2n)}\big)^{\!2}},\\
\frac{n+\mu}{\mathcal{W}_{m,n}^{(\mu)}+\mathcal{V}_{m,n}^{(\mu+2n)}}+\frac{n+\mu+1}{\mathcal{W}_{m,n}^{(\mu)}+\mathcal{V}_{m-1,n}^{(\mu+2n+2)}}&=\hfz-\frac{(\mu+2 n+1 ) \mathcal{W}_{m,n}^{(\mu)}-(2 m+2 n+1+\mu)}{1-\big(\mathcal{W}_{m,n}^{(\mu)}\big)^{\!2}}.
\end{align}\end{subequations}
\end{lemma}
\begin{proof}These follow from setting $N=0$ in \eqref{sys:lem57} and \eqref{sys:lem58}, respectively.
\end{proof}

\begin{lemma}
Consider the rational functions
\[\mathcal{X}_{m,n}^{(\mu)}=\frac{w_{m,n}^{(\mu)}}{1-w_{m,n}^{(\mu)}},\qquad \mathcal{Y}_{m,n}^{(\mu)}=\frac{1}{v_{m,n}^{(\mu)}-1},\qquad 
\mathcal{Z}_{m,n}^{(\mu)}=\frac{1}{u_{m,n}^{(\mu)}-1},\]
with $w_{m,n}^{(\mu)}(z)$, $v_{m,n}^{(\mu)}(z)$ and $u_{m,n}^{(\mu)}(z)$ as given in Theorem \ref{thm:genLag},
then they satisfy the partial difference system
\begin{subequations}\begin{align}
&\mathcal{X}_{m,n}^{(\mu)}\big(\mathcal{Y}_{m,n}^{(\mu+2n+1)}+\mathcal{Z}_{n,m}^{(\mu+2m+2n+1)}+1\big)+\frac{m+2n+\mu+1}{z}=0,\\
&\mathcal{Y}_{m,n}^{(\mu+2n+1)}\big(\mathcal{Z}_{n+1,m}^{(\mu+2m+2n+2)}+\mathcal{X}_{m,n}^{(\mu)}+1\big)+\frac{m+n+1}{z}=0,\\
&\mathcal{Z}_{n+1,m}^{(\mu+2m+2n+2)}\big(\mathcal{X}_{m,n+1}^{(\mu-1)}+\mathcal{Y}_{m,n}^{(\mu+2n+1)}+1\big)+\frac{n+1}{z}=0.
\end{align}\end{subequations}
\end{lemma}
\begin{proof}This follows from setting $N=0$ in \eqref{sys:lem511}.
\end{proof}

\subsection{\label{ssec:GenUm}Rational solutions in terms of generalised Umemura polynomials}
\def\RR#1#2#3#4#5{\mathcal{R}_{#1}\big(#5;#2,#3,#4\big)}
\subsubsection{\label{adPII:GenUm}First discrete equation: asymmetric discrete \PII}
In the following example, we illustrate how using the \bt\ $\mathcal{R}_{1}$ \eqref{bt:R1} gives rise to a hierarchy of solutions of asymmetric \dPII\ \eqref{eq:adPII} in terms of generalised Umemura polynomials. {We remark that Kitaev, Law and McLeod \cite[Theorem 1.2]{refKLM} proved that rational solutions expressed in terms of generalised Umemura polynomials are unique.}
\begin{example}
Applying the \bt\ $\mathcal{R}_{1}$ with $(a,b,c)=(\k+\hf,-\k+\hf,2)$ to the solution
\[\wh{w}_{2,1}^{(\k-3/2)}(z) = -\frac{\big[z^2-16\k(\k+1)\big]\big[z^3-6(2\k-1)z^2+48\k(\k-1)z-32\k(\k-1)(2\k-1)\big]}{\big[z^2-16\k(\k-1)\big]\big[z^3-6(2\k+1)z^2+48\k(\k+1)z-32\k(\k+1)(2\k+1)\big]},\] 
gives the sequence of rational solutions
\[ \wh{w}_{2,1}^{(\k-3/2)} \arr{\mathcal{R}_{1}} \wh{v}_{2,1}^{(\k)} \arr{\mathcal{R}_{1}} \wh{w}_{2,1}^{(\k-1/2)} \arr{\mathcal{R}_{1}} \wh{v}_{2,1}^{(\k+1)} \arr{\mathcal{R}_{1}} \wh{w}_{2,1}^{(\k+1/2)}\]
Hence we obtain solutions of asymmetric \dPII\ \eqref{eq:adPII}
\begin{align*} 
q_{0}(z)&=\frac{\wh{w}_{2,1}^{(\k-3/2)}(z)+1}{\wh{w}_{2,1}^{(\k-3/2)}(z)-1} 
=2\Ln{z^{4}-8\k z^{3}+128\k(\k^2-1)z -256\k^2(\k^2-1)}{z-4\k},\\[5pt]
q_{1}(z)&=\frac{\wh{v}_{2,1}^{(\k)}(z)+1}{\wh{v}_{2,1}^{(\k)}(z)-1} 
=2\Ln{z^{2}-16\k(\k+1)}{z^{3}-6 (2\k+1) z^{2}+48\k(\k+1)z -32\k(\k+1)(2\k+1)},\\[5pt]
q_{2}(z)&=\frac{\wh{w}_{2,1}^{(\k-1/2)}(z)+1}{\wh{w}_{2,1}^{(\k-1/2)}(z)-1}
=2\Ln{z^4-8(\k+1)+128\k(\k+1)(\k+2)-256\k(\k+1)^2(\k+2)}{z-4(\k+1)},\\[5pt]
q_{3}(z)&=\frac{\wh{v}_{2,1}^{(\k+1)}(z)+1}{\wh{v}_{2,1}^{(\k+1)}(z)-1} =2\Ln{z^{2}-16(\k+1)(\k+2)}{z^{3}-6 (2\k+3) z^{2}+48(\k+1)(\k+2)z -32(\k+1)(\k+2)(2\k+3)}.
\end{align*}
If we define the rational functions
\begin{align*} 
q_{2n}(z)&=x_{n}(z)=\frac{\wh{w}_{2,1}^{(\k+n-3/2)}(z)+1}{\wh{w}_{2,1}^{(\k+n-3/2)}(z)-1}=2\Ln{\Umn{2,1}{2\k+2n-3}(z)}{\Umn{1,0}{2\k+2n+1}(z)}\\ 
&=2\Ln{z^{4}-8(n+\k) z^{3}+128(n+\k)\big[(n+\k)^2-1\big]z -256(n+\k)^2\big[(n+\k)^2-1\big]}{z-4(n+\k)},\\
q_{2n+1}(z)&=y_{n}(z)=\frac{\wh{v}_{2,1}^{(\k+n)}(z)+1}{\wh{v}_{2,1}^{(\k+n)}(z)-1}=2\Ln{\Umn{2,0}{2\k+2n+1/2}(z)}{\Umn{1,1}{2\k+2n+1/2}(z)}\\
&=2\Ln{z^3-6(2n+2\k+1)z^2+48(n+\k)(n+\k+1)-32(n+\k)(n+\k+1)(2n+2\k+1)}{z^2-16(n+\k)(n+\k+1)},
\end{align*}
then $q_{n}(z)$ satisfies
\[q_{n+1}+q_{n-1}=\frac{4}{z}\frac{(n+2\k)q_{n}+2-(-1)^n}{1-q_{n}^2},\]
which is asymmetric \dPII\ \eqref{eq:adPII} with $\la=2\k$, $\rho=2$ and $\ph=-1$,
whilst $x_{n}(z)$ and $y_{n}(z)$ satisfy the discrete system
\begin{align*} 
x_{n+1}+x_{n}&=\frac{4}{z}\frac{(2n+2\k+1)y_{n}+3}{1-y_{n}^2},\qquad
y_{n}+y_{n-1}=\frac{4}{z}\frac{2(n+\k)x_{n}+1}{1-x_{n}^2}.
\end{align*}
\end{example}

More generally, using the solution $\wh{w}_{m,n}^{(\k)}$, for fixed $m$ and $n$, then we have the following Lemma.
\begin{lemma}\label{dPII_Um_case1} Consider the rational functions
\[Q_{2N}(z)=X_{N}(z)=\frac{\wh{w}_{m,n}^{(N+\k)}(z)+1}{\wh{w}_{m,n}^{(N+\k)}(z)-1},\qquad 
Q_{2N+1}(z)=Y_{N}(z)=\frac{\wh{v}_{m,n}^{(N+\k+n+1/2)}(z)+1}{\wh{v}_{m,n}^{(N+\k+n+1/2)}(z)-1},\]
with $\wh{w}_{m,n}^{(\k)}(z)$ and $\wh{v}_{m,n}^{(\k)}(z)$ as given in Theorem \ref{thm:genUm},
then $Q_{N}(z)$ satisfies
\[Q_{N+1}+Q_{N-1}=\frac{4}{z}\frac{(N+m+n+2\k)Q_{N}+m-n(-1)^N}{1-Q_{N}^2},\]
which is asymmetric \dPII\ \eqref{eq:adPII} with $\la=m+n+2\k$, $\rho=m$ and $\ph=-n$,
whilst $X_{N}(z)$ and $Y_{N}(z)$ satisfy the discrete system
\begin{subequations}\label{sys:513}\begin{align} 
X_{N+1}+X_{N}&=\frac{4}{z}\frac{(2N+m+n+2\k+1)Y_{N}+m+n}{1-Y_{N}^2},\\
Y_{N}+Y_{N-1}&=\frac{4}{z}\frac{(2N+m+n+2\k)X_{N}+m-n}{1-X_{N}^2}.
\end{align}\end{subequations}
\end{lemma}
\begin{proof}From the transformations
\begin{align*}
\RR{1}{m+\k}{-n-\k}{m+n}{\wh{w}_{m,n}^{(\k)}} &=\big(\wh{v}_{m,n}^{(\k+n+1/2)};m+n+\k+\hf,-\k-\hf,m-n\big), \\
\RR{1}{m+\k}{n-\k}{m-n}{\wh{v}_{m,n}^{(\k)}} &=\big(\wh{w}_{m,n}^{(\k-n+1/2)};m-n+\k+\hf,-\k-\hf,m+n\big),
\end{align*}
we obtain the sequence of solutions
\[ \wh{w}_{m,n}^{(\k)} \arr{\mathcal{R}_{1}} \wh{v}_{m,n}^{(\k+n+1/2)} \arr{\mathcal{R}_{1}} \wh{w}_{m,n}^{(\k+1)} \arr{\mathcal{R}_{1}} \wh{v}_{m,n}^{(\k+n+3/2)} \arr{\mathcal{R}_{1}} \wh{w}_{m,n}^{(\k+2)} 
\] 
from which the result follows. We note that
\[ \mathcal{R}_{1}^2\Big(\wh{w}_{m,n}^{(\k)} \Big)=\wh{w}_{m,n}^{(\k+1)},\qquad \mathcal{R}_{1}^2\Big(\wh{v}_{m,n}^{(\k)}\Big)=\wh{v}_{m,n}^{(\k+1)},\]
so that only the parameter changes.
\end{proof}

Alternatively, using the solution $\wh{v}_{m,n}^{(\k)}$, for fixed $m$ and $n$, then we have the following Lemma.
\begin{lemma}\label{dPII_Um_case2} Consider the rational functions
\[Q_{2N}(z)=X_{N}(z)=\frac{\wh{v}_{N+m,N+n}^{(\k)}(z)+1}{\wh{v}_{N+m,N+n}^{(\k)}(z)-1},\qquad 
Q_{2N+1}(z)=Y_{N}(z)=\frac{\wh{u}_{N+m,N+n}^{(2\k-2N-2n)}(z)+1}{\wh{u}_{N+m,N+n}^{(2\k-2N-2n)}(z)-1},\]
with $\wh{v}_{m,n}^{(\k)}(z)$ and $\wh{u}_{m,n}^{(\k)}(z)$ as given in Theorem \ref{thm:genUm},
then $Q_{N}(z)$ satisfies
\[Q_{N+1}+Q_{N-1}=\frac{4}{z}\frac{(N+m+n)Q_{N}+m-n+\k[1+(-1)^N]}{1-Q_{N}^2},\]
which is asymmetric \dPII\ \eqref{eq:adPII} with $\la=m+n$, $\rho=m-n+\k$ and $\ph=\k$,
whilst $X_{N}(z)$ and $Y_{N}(z)$ satisfy the discrete system
\begin{subequations}\label{sys:514}\begin{align} 
X_{N+1}+X_{N}&=\frac{4}{z}\frac{(2N+m+n+1)Y_{N}+m-n}{1-Y_{N}^2},\\
Y_{N}+Y_{N-1}&=\frac{4}{z}\frac{(2N+m+n)X_{N}+2\k+m-n}{1-X_{N}^2}.
\end{align}\end{subequations}
\end{lemma}
\begin{proof}From the transformations 
\begin{align*}
\RR{1}{m+\k}{-n+\k}{m-n}{\wh{v}_{m,n}^{(\k)}} &=\big(\wh{u}_{m,n}^{(2\k-2n)};m+\hf,-n-\hf,m-n+2\k\big), \\
\RR{1}{m+\hf}{-n-\hf}{m+n+2\k}{\wh{u}_{m,n}^{(2\k)}} &=\big(\wh{v}_{m+1,n+1}^{(n+\k)};m+n+\k+1,\k-1,m-n\big),
\end{align*}
we obtain the sequence of solutions
\[ \wh{v}_{m,n}^{(\k)} \arr{\mathcal{R}_{1}} \wh{u}_{m,n}^{(2\k-2n)} \arr{\mathcal{R}_{1}} \wh{v}_{m+1,n+1}^{(\k)} \arr{\mathcal{R}_{1}} \wh{u}_{m+1,n+1}^{(2\k-2n-2)} \arr{\mathcal{R}_{1}} \wh{v}_{m+2,n+2}^{(\k)}
\] 
from which the result follows. We note that
\[ \mathcal{R}_{1}^2\Big(\wh{w}_{m,n}^{(\k)} \Big)=\wh{w}_{m+1,n+1}^{(\k)},\qquad \mathcal{R}_{1}^2\Big(\wh{u}_{m,n}^{(2\k-2n)}\Big)=\wh{u}_{m+1,n+1}^{(2\k-2n-2)}.\]
\end{proof}

\begin{remark}{\rm From the transformations
\begin{subequations}\label{sys:519}\begin{align}
\RR{1}{m+\k}{n+\k}{m+n}{\wh{w}_{m,n}^{(\k)}} &=\big(\wh{u}_{m,n-1}^{(2\k+1)};m+\hf,n-\hf,m+n+2\k\big), \label{sys:519a}\\
\RR{1}{m+\hf}{n-\hf}{m+n+2\k}{u_{m,n-1}^{(2\k+1)}} &=\big(\wh{w}_{m+1,n-1}^{(\k)};m+\k+1,n+\k-1,m+n\big),\label{sys:519b}
\end{align}\end{subequations}
we obtain the sequence
\[ \wh{w}_{m,n}^{(\k)} \arr{\mathcal{R}_{1}} \wh{u}_{m,n-1}^{(2\k+1)} \arr{\mathcal{R}_{1}} \wh{w}_{m+1,n-1}^{(\k)} \arr{\mathcal{R}_{1}} \wh{u}_{m,n-2}^{(2\k+1)} \arr{\mathcal{R}_{1}} \wh{w}_{m+2,n-2}^{(\k)}. \]
It might be expected that the sequence would terminate when the second index becomes zero. However it can be shown that
\[\wh{w}_{m,1}^{(\k)} \arr{\mathcal{R}_{1}} \wh{u}_{m,0}^{(2\k+1)} \arr{\mathcal{R}_{1}} \wh{w}_{m+1,0}^{(\k)} \arr{\mathcal{R}_{1}} \wh{u}_{m+1,0}^{(2\k)} \]
and
\[ \wh{u}_{m+1,0}^{(2\k)} \arr{\mathcal{R}_{1}} \wh{v}_{m+2,1}^{(\k)} \arr{\mathcal{R}_{1}} \wh{u}_{m+2,1}^{(2\k-2)} \arr{\mathcal{R}_{1}} \wh{v}_{m+3,2}^{(\k)}.\]
}\end{remark}

This is illustrated in the following example.
\begin{example}{\rm Consider the rational solution
\[ \wh{w}_{1,1}^{(\k)}(z)=-\frac{(z-4\k-2)(z+4\k+6)}{(z+4\k+2)(z-4\k-6)},\]
which satisfies \PV\ with $\pms{\hf(\k+1)^2}{\hf(\k+1)^2}{2}$, and so using \eqref{sys:519} we obtain the sequence
\[ \wh{w}_{1,1}^{(\k)} \arr{\mathcal{R}_{1}} \wh{u}_{1,0}^{(2\k+1)} \arr{\mathcal{R}_{1}} \wh{w}_{2,0}^{(\k)} \arr{\mathcal{R}_{1}} \wh{u}_{2,0}^{(2\k)} \arr{\mathcal{R}_{1}}
\wh{v}_{3,1}^{(\k)} \arr{\mathcal{R}_{1}} u_{3,1}^{(2\k-2)}.
\]
The first few solutions are
\begin{align*} q_{0}(z)&=\frac{\wh{w}_{1,1}^{(\k)}(z)+1}{\wh{w}_{1,1}^{(\k)}(z)-1}=\frac{4z}{z^2-4(2\k+1)(2\k+3)},\\
q_{1}(z) &=\frac{4(\k+1)}{z}+2\deriv{}{z}\ln\frac{z-4\k-6}{z-4\k-2},\\
q_{2}(z)&=
\frac{6\big[z^2-8(\k+1)z+4(2\k+1)(2\k+3) \big]}{z^3-12(\k+1)z^2+12(2\k+1)(2\k+3) z-16(\k+1)(2\k+1)(2\k+3)}-\frac{2}{z-4\k-4},\\
q_{3}(z)&=\frac{4(\k+1)}{z}+2\deriv{}{z}\ln\frac{z^3-6(2\k+3)z^2+48(\k+1)(\k+2) z-32(\k+1)(\k+2)(2\k+3)}{z^3-6(2\k+1)z^2+48\k(\k+1)z-32\k(\k+1)(2\k+1)},
\end{align*}
and satisfy
\beq q_{n+1}+q_{n-1}=
\frac{4}{z} \frac{nq_{n}+2+[1+(-1)^n]\k}{1-q_{n}^2},\label{eq:525}\eeq
which is asymmetric \dPII\ \eqref{eq:adPII} with $\la=0$, $\rho=2+\k$ and $\ph=\k$. In general
\[q_{2n}=x_{n}=2\Ln{\Umn{n+1,n-2}{2\k-2n+3}}{\Umn{n,n-1}{2\k-2n+3}},\qquad
q_{2n+1}=y_{n}=\frac{4(\k+1)}{z}+2\Ln{\Umn{n+1,n-1}{2\k-2n+3}}{\Umn{n+1,n-1}{2\k-2n+1}},
\]
satisfy \eqref{eq:525} and
\begin{align*}
x_{n+1}+x_{n}&=\frac{4}{z}\,\frac{(2n+1)y_{n}+2}{1-y_{n}^2},\qquad
y_{n}+y_{n-1}=\frac{4}{z}\,\frac{2nx_{n}+2\k+2}{1-x_{n}^2}.
\end{align*}
}\end{example}
\comment{\begin{example}{\rm Consider the rational solution
\[ w_{1,1}^{(1/2)}=-\frac{(z-4)(z+8)}{(z+4)(z-8)},\]
which satisfies \PV\ with $\pms{\tfrac{9}{2}}{\tfrac{9}{2}}{2}$ and so
\[ w_{1,1}^{(1/2)} \arr{\mathcal{R}_{1}} u_{1,0}^{(2)} \arr{\mathcal{R}_{1}} w_{2,0}^{(1/2)} \arr{\mathcal{R}_{1}} u_{2,0}^{(2\k)} \arr{\mathcal{R}_{1}}
v_{3,1}^{(1/2)} \arr{\mathcal{R}_{1}} u_{3,1}^{(-1)}
\]
The first few associated solutions are
\[ q_{0}=\frac{w_{1,1}^{(1/2)}+1}{w_{1,1}^{(1/2)}-1}=\frac{4z}{z^2-32},\] 
and
\begin{align*} q_{1} &=\frac{2(3z^2-32z+96)}{z(z-4)(z-8)},\qquad q_{2}=\frac{4(z^{3}-72 z^{2}+432z-768)}{(z-6) (z^{3}-18 z^{2}+96z-192)},\\ 
q_{3}&=\frac{6 (z^{6}-192 z^{5}+2448 z^{4}-16128 z^{3}+61920 z^{2}-138240z+138240)}{z (z^{3}-12 z^{2}+36z-48)(z^{3}-24 z^{2}+180z-480)}
\end{align*}
which satisfy
\[q_{n+1}+q_{n-1}=
\frac{2\big[2nq_{n}+5+(-1)^n\big]}{z(1-q_{n}^2)},\]
}\end{example}}

\smallskip\noindent\textbf{Symmetric case}.
In the symmetric case when $\ph=0$ in \eqref{eq:adPII}, then
\[ q_{n}(z)=\frac{\wh{w}_{m,0}^{(\k+n/2)}(z)+1}{\wh{w}_{m,0}^{(\k+n/2)}(z)-1},\qquad p_{n}(z)=\frac{\wh{w}_{0,m}^{(\k-n/2)}(z)+1}{\wh{w}_{0,m}^{(\k-n/2)}(z)-1},\]
where $\wh{w}_{m,0}^{(\k)}(z)$ and $\wh{w}_{0,m}^{(\k)}(z)$ are given by \eqref{Um:wm0} and \eqref{Um:wn0} respectively,
which satisfy
\begin{align*}
q_{n+1}+q_{n-1}&=\frac{4}{z}\frac{(n+m+2\k)q_{n}+m}{1-q_{n}^2},\qquad
p_{n+1}+p_{n-1}=\frac{4}{z}\frac{(n-m-2\k)p_{n}+m}{1-p_{n}^2}.
\end{align*}

\subsubsection{\label{dPn:GenUm}Second discrete equation}
In the following example, we illustrate how using the \bt\ $\mathcal{R}_{2}$ \eqref{bt:R2} gives rise to a hierarchy of solutions of \eqref{eq:dPn} in terms of generalised Umemura polynomials.
\begin{example}
Applying the \bt\ $\mathcal{R}_{2}$ \eqref{bt:R2}, with $(a,b,c)=(-\k-1,-\k-m,m+1)$ to the solution $\wh{w}_{1,m}^{(\k)}(z)$ gives the sequence
\[ \wh{w}_{1,m}^{(\k)}\arr{\mathcal{R}_{2}} \wh{v}_{2,m}^{(\k+m-1/2)}\arr{\mathcal{R}_{2}} \wh{w}_{3,m}^{(\k-1)} \arr{\mathcal{R}_{2}} \wh{v}_{4,m}^{(\k+m-3/2)}\]
Hence if we define the rational functions
\begin{align*} 
q_{2n}(z)&=x_{n}(z)=\frac{\wh{w}_{2n+1,m}^{(\k-n)}(z)+1}{\wh{w}_{2n+1,m}^{(\k-n)}(z)-1}=2\Ln{\Umn{2n+1,m}{2\k-2n}(z)}{\Umn{2n,m-1}{2\k-2n+2}(z)},\\ 
q_{2n+1}(z)&=y_{n}(z)=\frac{\wh{v}_{2n+2,m}^{(\k+m-n-1/2)}(z)+1}{\wh{v}_{2n+2,m}^{(\k+m-n-1/2)}(z)-1}=2\Ln{\Umn{2n+2,m-1}{2\k-2n-2m+2}(z)}{\Umn{2n+1,m}{2\k-2n-2m+2}(z)},
\end{align*}
then $q_{n}(z)$ satisfies
\[\frac{2n+3}{q_{n+1}+q_{n}}+\frac{2n+1}{q_{n}+q_{n-1}}=z-\frac{2\big[n+1-m(-1)^n\big]q_{n}+4\k+2m+2}{1-q_{n}^2},\]
which is equation \eqref{eq:dPn} with $\la=m$, $\ph=-n$ and $\rho=m+n+2 \k$,
whilst $x_{n}(z)$ and $y_{n}(z)$ satisfy the discrete system
\begin{align*} 
\frac{4n+5}{x_{n+1}+y_{n}}+\frac{4n+3}{x_{n}+y_{n}}&=z-\frac{2(2n+m+1)y_{n}+4\k+2m+2}{1-y_{n}^2},\\
\frac{4n+3}{x_{n}+y_{n}}+\frac{4n+1}{x_{n}+y_{n-1}}&=z-\frac{2(2n-m+1)x_{n}+4\k+2m+2}{1-x_{n}^2}.
\end{align*}
\end{example}

More generally, using the solution $\wh{w}_{m,n}^{(\k)}$, for fixed $m$ and $n$, then we have the following Lemma.
\begin{lemma}
Consider the rational functions 
\[ Q_{2N}(z)=X_{N}(z)=\frac{\wh{w}_{2N+m,n}^{(\k-N)}(z)+1}{\wh{w}_{2N+m,n}^{(\k-N)}(z)-1},\qquad Q_{2N+1}(z)=Y_{N}(z)=\frac{\wh{v}_{2N+m+1,n}^{(\k+n-N-1/2)}(z)+1}{\wh{v}_{2N+m+1,n}^{(\k+n-N-1/2)}(z)-1},
\]
with $\wh{w}_{m,n}^{(\k)}(z)$ and $\wh{v}_{m,n}^{(\k)}(z)$ as given in Theorem \ref{thm:genUm},
then $Q_{N}(z)$ satisfies
\[\frac{2N+2m+1}{Q_{N+1}+Q_{N}}+\frac{2N+2m-1}{Q_{N}+Q_{N-1}}=z-\frac{2\big[N+m-n(-1)^N\big]Q_{N}+2(m+n+2\k)}{1-Q_{N}^2},\]
which is equation \eqref{eq:dPn} with $\la=m$, $\ph=-n$ and $\rho=m+n+2 \k$,
whilst $X_{N}(z)$ and $Y_{N}(z)$ satisfy the discrete system
\begin{subequations}\label{sys:515}\begin{align} 
\frac{4N+2m+3}{X_{N+1}+Y_{N}}+\frac{4N+2m+1}{X_{N}+Y_{N}}&=z-\frac{2[(2N+m+n+1)Y_{N}+m+n+2\k]}{1-Y_{N}^2},\\
\frac{4N+2m+1}{X_{N}+Y_{N}}+\frac{4N+2m-1}{X_{N}+Y_{N-1}}&=z-\frac{2[(2N+m-n)X_{N}+m+n+2\k]}{1-X_{N}^2}.
\end{align}\end{subequations}
\end{lemma}
\begin{proof}From the transformations
\begin{align*}
\RR{2}{-m-\k}{-n-\k}{m+n}{\wh{w}_{m,n}^{(\k)}} &=\big(\wh{v}_{m+1,n}^{(\k+n-1/2)};-m-n-\k-\hf,-\k+\hf,m-n+1\big),\\
\RR{2}{-m-\k}{n-\k}{m-n}{\wh{v}_{m,n}^{(\k)}} &=\big(\wh{w}_{m+1,n}^{(\k-n-1/2)};-m-\k-\hf,n-\k+\hf,m+n+1\big)
\end{align*}
we obtain the sequence of solutions
\[\wh{w}_{m,n}^{(\k)}\arr{\mathcal{R}_{2}} \wh{v}_{m+1,n}^{(\k+n-1/2)}\arr{\mathcal{R}_{2}} \wh{w}_{m+2,n}^{(\k-1)}\arr{\mathcal{R}_{2}} \wh{v}_{m+3,n}^{(\k+n-3/2)}\arr{\mathcal{R}_{2}} \wh{w}_{m+4,n}^{(\k-2)}
\] 
from which the result follows. We note that
\[ \mathcal{R}_{2}^2\Big(\wh{w}_{m,n}^{(\k)} \Big)=\wh{w}_{m+2,n}^{(\k-1)},\qquad \mathcal{R}_{2}^2\Big(\wh{v}_{m,n}^{(\k)}\Big)=\wh{v}_{m+2,n}^{(\k-1)}.\]
\end{proof}
When $n=0$ there is no alternating term and so
\[Q_{2N}(z)=\frac{\wh{w}_{2N+m,0}^{(\k-N)}(z)+1}{\wh{w}_{2N+m,0}^{(\k-N)}(z)-1},\qquad Q_{2N+1}(z)=\frac{\wh{v}_{2N+m+1,0}^{(\k-N-1/2)}(z)+1}{\wh{v}_{2N+m+1,0}^{(\k-N-1/2)}(z)-1},\]
satisfy
\[\frac{2N+2m+1}{Q_{N+1}+Q_{N}}+\frac{2N+2m-1}{Q_{N}+Q_{N-1}}=z-\frac{2(N+m)Q_{N}+2m+4\k}{1-Q_{N}^2}.\]

It is also possible to obtain a sequence of solutions involving the solutions 
$\wh{w}_{m,n}^{(\k)}$ and $\wh{u}_{m,n}^{(2\k+1)}$ as illustrated in the following Lemma.
\begin{lemma}
Consider the rational functions 
\[ Q_{2N}(z)=X_{N}(z)=\frac{\wh{w}_{N+m,N+n}^{(\k)}(z)+1}{\wh{w}_{N+m,N+n}^{(\k)}(z)-1},\qquad Q_{2N+1}(z)=Y_{N}(z)=\frac{\wh{u}_{N+m,N+n}^{(2\k+1)}(z)+1}{\wh{u}_{N+m,N+n}^{(2\k+1)}(z)-1},\]
with $\wh{w}_{m,n}^{(\k)}(z)$ and $\wh{u}_{m,n}^{(\k)}(z)$ as given in Theorem \ref{thm:genUm},
then $Q_{N}(z)$ satisfies
\[\frac{2N+2m+2n+2\k+1}{Q_{N+1}+Q_{N}}+\frac{2N+2m+2n+2\k-1}{Q_{N}+Q_{N-1}}=z-\frac{2\big[N+m+n+\k\big\{1+(-1)^N\big\}\big] Q_{N}+2(m-n)}{1-Q_{N}^2},\]
which is equation \eqref{eq:dPn} with $\la=m+n+\k$, $\ph=\k$ and $\rho=m-n$,
whilst $X_{N}(z)$ and $Y_{N}(z)$ satisfy the discrete system
\begin{subequations}\label{sys:516}\begin{align} 
\frac{4 N+2m+2n+2\k+3}{X_{N+1}+Y_{N}}+\frac{4 N+2m+2n+2\k+1}{X_{N}+Y_{N}} &=z-\frac{2[(2 N+m+n+1) Y_{N}+m-n]}{1-Y_{N}^{2}},\\
\frac{4 N+2m+2n+2\k+1}{X_{N}+Y_{N}}+\frac{4 N+2m+2n+2\k-1}{X_{N}+Y_{N-1}} &=z-\frac{2[(2 N+m+n+2 \k) X_{N}+m-n]}{1-X_{N}^{2}}.
\end{align}\end{subequations}
\end{lemma}
\begin{proof}From the transformations
\begin{align*}
\RR{2}{-m-\k}{n+\k}{m+n}{\wh{w}_{m,n}^{(\k)}} &=\big(\wh{u}_{m,n}^{(2\k+1)};-m-\hf,n+\hf,m+n+2\k+1\big), \\
\RR{2}{-m-\hf}{n+\hf}{m+n+2\k+1}{\wh{u}_{m,n}^{(2\k+1)}} &=\big(\wh{w}_{m+1,n+1}^{(\k)};-m-\k-1,n+\k+1,m+n+2\big).
\end{align*}
we obtain the sequence of solutions
\[ \wh{w}_{m,n}^{(\k)} \arr{\mathcal{R}_{2}} \wh{u}_{m,n}^{(2\k+1)} \arr{\mathcal{R}_{2}} \wh{w}_{m+1,n+1}^{(\k)} \arr{\mathcal{R}_{2}} \wh{u}_{m+1,n+1}^{(2\k+1)} \arr{\mathcal{R}_{2}} \wh{w}_{m+2,n+2}^{(\k)} \]
from which the result follows. We note that
\[ \mathcal{R}_{2}^2\Big(\wh{w}_{m,n}^{(\k)} \Big)=\wh{w}_{m+1,n+1}^{(\k)},\qquad \mathcal{R}_{2}^2\Big(\wh{u}_{m,n}^{(2\k+1)}\Big)=\wh{u}_{m+1,n+1}^{(2\k+1)}.\]
\end{proof}
When $\k=0$ then there is no alternating term and so
\[ Q_{2N}(z)=\frac{\wh{w}_{N+m,N+n}^{(0)}(z)+1}{\wh{w}_{N+m,N+n}^{(0)}(z)-1},\qquad Q_{2N+1}(z)=\frac{\wh{u}_{N+m,N+n}^{(1)}(z)+1}{\wh{u}_{N+m,N+n}^{(1)}(z)-1},\]
satisfy
\[\frac{2N+2m+2n+1}{Q_{N+1}+Q_{N}}+\frac{2N+2m+2n-1}{Q_{N}+Q_{N-1}}=z-\frac{2\big[(N+m+n)Q_{N}+m-n\big]}{1-Q_{N}^2},\]
which is equation \eqref{eq:dPn} with $\la=m+n$, $\rho=m-n$ and $\ph=0$.
As previously, although there is no alternating term, the structure of the even and odd terms is different.

\smallskip\noindent\textbf{Symmetric case}.
In the symmetric case when $\ph=0$ in \eqref{eq:dPn1}, from the transformations
\begin{align*}
\RR{2}{-m-\k}{-\k}{m}{\wh{w}_{m,0}^{(\k)}} &=\big(\wh{w}_{m+1,0}^{(\k-1/2)};-m-\k-\hf,-\k+\hf,m+1\big), \\
\RR{2}{n}{n+\hf}{n+\k}{\wh{w}_{0,n}^{(\k)}} &=\big(\wh{w}_{0,n+1}^{(\k-1/2)};\k-\hf,n+\k+\hf,n+1\big),
\end{align*}
recall \eqref{Um:wm0} and \eqref{Um:wn0}, we obtain the respective sequences
\[\begin{split}&\wh{w}_{m,0}^{(\k)} \arr{\mathcal{R}_{2}} \wh{w}_{m+1,0}^{(\k-1/2)} \arr{\mathcal{R}_{2}} \wh{w}_{m+2,0}^{(\k-1)}\\ 
&\wh{w}_{0,n}^{(\k)} \arr{\mathcal{R}_{2}} \wh{w}_{0,n+1}^{(\k-1/2)} \arr{\mathcal{R}_{2}} \wh{w}_{0,n+2}^{(\k-1)}\end{split}.\]
If we define
\[ q_{n}(z)=\frac{\wh{w}_{m+n,0}^{(\k-n/2)}(z)+1}{\wh{w}_{m+n,0}^{(\k-n/2)}(z)-1},\qquad p_{n}(z)=\frac{\wh{w}_{0,n+m}^{(\k-n/2)}(z)+1}{\wh{w}_{0,n+m}^{(\k-n/2)}(z)-1}, \]
then ${q}_{n}(z)$ and $p_{n}(z)$ satisfy
\begin{align*}
\frac{2n+2m+1}{q_{n+1}+q_{n}}+\frac{2n+2m-1}{q_{n}+q_{n-1}}&=z-\frac{2(n+m)q_{n}+2m+4\k}{1-q_{n}^2},\\
\frac{2n+2m+1}{p_{n+1}+p_{n}}+\frac{2n+2m-1}{p_{n}+p_{n-1}}&=z-\frac{2(n+m)p_{n}-2m-4\k}{1-p_{n}^2},
\end{align*}
which are the discrete equation \eqref{eq:dPn1} with $\la=m$, $\rho=2\k+m$, $\ph=0$ and
$\la=m$, $\rho=-2\k-m$, $\ph=0$, respectively.

\subsubsection{\label{tdPI:GenUm}Third discrete equation: ternary discrete \PI}
In the following example, we illustrate how using the \bt\ $\mathcal{R}_{3}$ \eqref{bt:R3} gives rise to a hierarchy of solutions of ternary \dPI\ \eqref{eq:tdPI} in terms of generalised Umemura polynomials.

\begin{example}
Applying the \bt\ $\mathcal{R}_{3}$ \eqref{bt:R3}, with $(a,b,c)=(\k+1,-\k-m,m+1)$ to the solution ${w}_{1,m}^{(\k)}$ gives the sequence of rational solutions
\[ \wh{w}_{1,m}^{(\k)} \arr{\mathcal{R}_{3}} \wh{v}_{1,m}^{(\k+m+1/2)} \arr{\mathcal{R}_{3}} \wh{u}_{1,m}^{(2\k+1)} \arr{\mathcal{R}_{3}} \wh{w}_{2,m}^{(\k)} \arr{\mathcal{R}_{3}} 
\wh{v}_{2,m}^{(\k+m+1/2)} \arr{\mathcal{R}_{3}} \wh{u}_{2,m}^{(2\k+1)}\]
so that
\[ \mathcal{R}_{3}^3\Big(\wh{w}_{1,m}^{(\k)}\Big)=\wh{w}_{2,m}^{(\k)},\qquad \mathcal{R}_{3}^3\Big(\wh{v}_{1,m}^{(\k+m+1/2)} \Big)= \wh{v}_{2,m}^{(\k+m+1/2)},\qquad
\mathcal{R}_{3}^3\Big(\wh{u}_{1,m}^{(2\k+1)} \Big)=\wh{u}_{2,m}^{(2\k+1)}.\]
Hence if we define the rational functions
\begin{align*} 
x_{3n}(z)={X}_{n}(z)&=\frac{1}{\wh{w}_{n+1,m}^{(\k)}(z)-1} = -\frac{1}{2}+\Ln{\Umn{n+1,m}{2\k}(z)}{\Umn{n,m-1}{2\k+2}(z)},\\
x_{3n+1}(z)={Y}_{n}(z)&=\frac{1}{\wh{v}_{n+1,m}^{(\k+m+1/2)}(z)-1}= -\frac{1}{2}+\Ln{\Umn{n+1,m-1}{2\k+2}(z)}{\Umn{n,m}{2\k+2}(z)},\\
x_{3n+2}(z)={Z}_{n}(z)&=\frac{1}{\wh{u}_{n+1,m}^{(2\k+1)}(z)-1}= -\frac{1}{2}+\frac{n+m+2\k+2}{z}+\Ln{\Umn{n+1,m}{2\k+2}(z)}{\Umn{n+1,m}{2\k}(z)},
\end{align*}
then $x_{n}(z)$ satisfies ternary \dPI\ 
\[ x_{n}(x_{n+1}+x_{n-1}+1) +\frac{a_{n}}{z}=0,\]
with
\[ a_{n}=\tfrac{1}{3}(n+m+2\k)+1+\tfrac{1}{3}(\k-m)\cos\!\big(\tfrac{2}{3}\pi n\big)+\tfrac{1}{9}(3m+3\k+1)\sqrt{3}\sin\!\big(\tfrac{2}{3}\pi n\big),\]
so
\[ a_{3n}=n+\k+1,\qquad a_{3n+1}=n+m+\k+\tfrac{3}{2}, \qquad a_{3n+2}=n+\tfrac{3}{2} ,\]
whilst ${X}_{n}(z)$, ${Y}_{n}(z)$ and $ {Z}_{n}(z)$ satisfy
\begin{align*}
&{X}_{n}\big({Y}_{n}+{Z}_{n-1}+1\big)+\frac{n+\k+1}{z}=0,\\
&{Y}_{n}\big({Z}_{n}+{X}_{n}+1\big)+\frac{n+m+\k+\tfrac{3}{2}}{z}=0,\\
&Z_{n}\big({X}_{n+1}+{Y}_{n}+1\big)+\frac{n+\tfrac{3}{2}}{z}=0.
\end{align*}
\end{example}

More generally, using the solution $\wh{w}_{m,n}^{(\k)}$, for fixed $m$ and $n$, then we have the following Lemma.
\begin{lemma}\label{lem:528}
Consider the rational functions
\beq\label{xyz:case1} {X}_{N}(z)=\frac{1}{\wh{w}_{N+m,n}^{(\k)}(z)-1},\qquad {Y}_{N}(z)=\frac{1}{\wh{v}_{N+m,n}^{(\k+n+1/2)}(z)-1}, \qquad {Z}_{N}(z)=\frac{1}{\wh{u}_{N+m,n}^{(2\k+1)}(z)-1},\eeq
with $\wh{w}_{m,n}^{(\k)}(z)$, $\wh{v}_{m,n}^{(\k)}(z)$ and $\wh{u}_{m,n}^{(\k)}(z)$ as given in Theorem \ref{thm:genUm},
then ${X}_{N}(z)$, ${Y}_{N}(z)$ and $ {Z}_{N}(z)$ satisfy
\begin{subequations}\label{xyz:sys1}\begin{align}
&{X}_{N}\big({Y}_{N}+{Z}_{N-1}+1\big)+\frac{N+m+\k}{z}=0,\\
&{Y}_{N}\big({Z}_{N}+{X}_{N}+1\big)+\frac{N+m+n+\k+\hf}{z}=0,\\
&Z_{N}\big({X}_{N+1}+{Y}_{N}+1\big)+\frac{N+m+\hf}{z}=0,
\end{align} \end{subequations}
which is the system \eqref{sys:tertdPI} with
$\la=\tfrac{1}{3}(\k-n)$, $\ph=n+\k+\tfrac{1}{3}$ and $\rho=m+\tfrac{1}{3}n+\tfrac{2}{3}\k$.
\end{lemma}

\begin{proof}From the transformations
\begin{align*}
\RR{3}{m+\k}{-n-\k}{m+n}{\wh{w}_{m,n}^{(\k)}} &=\big(\wh{v}_{m,n}^{(\k+n+1/2)};m+n+\k+\hf,\k+\hf,m-n\big), \\
\RR{3}{m+\k}{-n+\k}{m-n}{\wh{v}_{m,n}^{(\k)}} &=\big(\wh{u}_{m,n}^{(2\k-2n)};m+n+\hf,\hf,m-n+2\k\big), \\
\RR{3}{m+\hf}{n+\hf}{m+n+2\k}{\wh{u}_{m,n}^{(2\k)}} &=\big(\wh{w}_{m+1,n}^{(\k-1/2)};m+\k+\hf,-n-\k+\hf,m+n+1\big),
\end{align*}
we obtain the sequence of solutions
\[ \wh{w}_{m,n}^{(\k)} \arr{\mathcal{R}_{3}} \wh{v}_{m,n}^{(\k+n+1/2)} \arr{\mathcal{R}_{3}} \wh{u}_{m,n}^{(2\k+1)} \arr{\mathcal{R}_{3}} \wh{w}_{m+1,n}^{(\k)} \arr{\mathcal{R}_{3}} \wh{v}_{m+1,n}^{(\k+n+1/2)} \arr{\mathcal{R}_{3}} \wh{u}_{m+1,n}^{(2\k+1)}\] 
from which the result follows. We note that
\[ \mathcal{R}_{3}^3\Big(\wh{w}_{m,n}^{(\k)} \Big)=\wh{w}_{m+1,n}^{(\k)},\qquad 
\mathcal{R}_{3}^3\Big(\wh{v}_{m,n}^{(\k)} \Big)=\wh{v}_{m+1,n}^{(\k)},\qquad \mathcal{R}_{3}^3\Big(\wh{u}_{m,n}^{(2\k+1)}\Big)=\wh{u}_{m+1,n}^{(2\k+1)}.\]
\end{proof}
\begin{remark}
It is interesting to note that ${X}_{N}(z)$, ${Y}_{N}(z)$ and $Z_{N}(z)$ arise from the three different sets of rational solutions of \PV\ that are expressed in terms of the generalised Umemura polynomials.
\end{remark}

As was the case for $\mathcal{R}_{2}$, it is also possible to obtain another sequence of solutions as illustrated in the following Lemma.
\begin{lemma}\label{lem:530}
Consider the rational functions
\beq\label{xyz:case2} {X}_{N}(z)=\frac{1}{\wh{w}_{N+m,n}^{(\k)}(z)-1},\qquad {Y}_{N}(z)=\frac{1}{\wh{u}_{N+m,n-1}^{(2\k+1)}(z)-1}, \qquad {Z}_{N}(z)=\frac{1}{\wh{v}_{N+m+1,n}^{(\k+n-1/2)}(z)-1},\eeq
with $\wh{w}_{m,n}^{(\k)}(z)$, $\wh{v}_{m,n}^{(\k)}(z)$ and $\wh{u}_{m,n}^{(\k)}(z)$ as given in Theorem \ref{thm:genUm},
then ${X}_{N}(z)$, ${Y}_{N}(z)$ and $ {Z}_{N}(z)$ satisfy
\begin{subequations}\label{xyz:sys2}\begin{align}
&{X}_{N}\big({Y}_{N}+{Z}_{N-1}+1\big)+\frac{N+m+\k}{z}=0,\\
&{Y}_{N}\big({Z}_{N}+{X}_{N}+1\big)+\frac{N+m+\hf}{z}=0,\\
&{Z}_{N}\big({X}_{N+1}+{Y}_{N}+1\big)+\frac{N+m+n+\k+\hf}{z}=0,
\end{align}\end{subequations}
which is the system \eqref{sys:tertdPI} with
$\la=\tfrac{1}{3}(\k-n)$, $\ph=-n-\k+\tfrac{1}{3}$ and $\rho=m+\tfrac{1}{3}n+\tfrac{2}{3}\k$.
\end{lemma}
\begin{proof}From the transformations
\begin{align*}
\RR{3}{m+\k}{n+\k}{m+n}{\wh{w}_{m,n}^{(\k)}} &=\big(\wh{u}_{m,n-1}^{(2\k+1)};m+\hf,-n+\hf,m+n+2\k\big),\\ 
\RR{3}{m+\hf}{-n-\hf}{m+n+2\k}{\wh{u}_{m,n}^{(2\k)}} &=\big(\wh{v}_{m+1,n+1}^{(\k+n)};m+n+\k+1,-\k+1,m-n\big),\\
\RR{3}{m+\k}{n-\k}{m+n}{\wh{v}_{m,n}^{(\k)}} &=\big(\wh{w}_{m,n}^{(\k-n+1/2)};m+\k+\hf,-n+\k+\hf,m+n\big),
\end{align*}
we obtain the sequence of solutions
\[ \wh{w}_{m,n}^{(\k)} \arr{\mathcal{R}_{3}} \wh{u}_{m,n-1}^{(2\k+1)} \arr{\mathcal{R}_{3}} \wh{v}_{m+1,n}^{(\k+n-1/2)} \arr{\mathcal{R}_{3}} \wh{w}_{m+1,n}^{(\k)} \arr{\mathcal{R}_{3}} \wh{u}_{m+1,n-1}^{(2\k+1)} \] 
from which the result follows. We note that
\[\mathcal{R}_{3}^3\Big(\wh{w}_{m,n}^{(\k)} \Big)=\wh{w}_{m+1,n}^{(\k)},\qquad 
\mathcal{R}_{3}^3\Big(\wh{v}_{m,n}^{(\k)} \Big)=\wh{v}_{m+1,n}^{(\k)},\qquad \mathcal{R}_{3}^3\Big(\wh{u}_{m,n}^{(2\k+1)}\Big)=\wh{u}_{m+1,n}^{(2\k+1)}.\]
\end{proof}

\comment{\item From the transformations
\begin{align*}
&\mathcal{R}_{3}\big(\wh{w}_{m,n}^{(\k)};-m-\k,n+\k,m+n\big)=\big(1/v_{m,n}^{(\k+n-1/2)};-\k+\hf,-m-n-\k+\hf,n-m\big),\\
&\mathcal{R}_{3}\big(1/v_{m,n}^{(\k+n-1/2)};-\k+\hf,-m-n-\k+\hf,n-m\big)=\big(1/u_{m,n}^{(2\k-1)};-n+\hf,m+\hf-m-n+1-2\k\big),\\
&\mathcal{R}_{3}\big(1/u_{m,n}^{(2\k-1)};-n+\hf,m+\hf-m-n+1-2\k\big)=\big(w_{m+1,n}^{(\k)};m+1+\k,n+\k,m+n+1\big),
\end{align*}
we obtain the sequence of solutions
\[ \wh{w}_{m,n}^{(\k)} \arr{\mathcal{R}_{3}} \frac{1}{v_{m,n}^{(\k+n-1/2)}} \arr{\mathcal{R}_{3}} \frac{1}{u_{m,n}^{(2\k-1)}} \arr{\mathcal{R}_{3}} w_{m+1,n}^{(\k)}\] 
\item From the transformations
\begin{align*}
&\mathcal{R}_{3}\big(\wh{w}_{m,n}^{(\k)};-m-\k,-n-\k,m+n\big)=\big(1/u_{m-1,n}^{(2\k+1)};n+\hf,-m+\hf,-m-n-2\k\big),\\
&\mathcal{R}_{3}\big(1/u_{m-1,n}^{(2\k+1)};n+\hf,-m+\hf,-m-n-2\k\big)=\big(1/v_{m,n+1}^{(\k+n+1/2)};-\k+\hf,m+n+\k+\hf,-m+n+1\big),\\
&\mathcal{R}_{3}\big(1/v_{m,n+1}^{(\k+n+1/2)};-\k+\hf,m+n+\k+\hf,-m+n+1\big)=\big(w_{m+1,n}^{(\k)};m+1+\k,n+\k,m+n+1\big),
\end{align*}
we obtain the sequence of solutions
\[ \wh{w}_{m,n}^{(\k)} \arr{\mathcal{R}_{3}} \frac{1}{u_{m-1,n}^{(2\k+1)}}\arr{\mathcal{R}_{3}} \frac{1}{v_{m,n+1}^{(\k+n+1/2)}} \arr{\mathcal{R}_{3}} w_{m+1,n}^{(\k)}\]}

\begin{remarks}
\begin{enumerate}[(i)]\item[]
\item There are no solutions in terms of generalised Umemura polynomials in the symmetric case when $\la=\ph=0$.
\item Lemmas \ref{lem:528} and \ref{lem:530} imply that there are hierarchies rational solutions of the fourth discrete equation \eqref{eq:tdPInew} in terms of generalised Umemura polynomials.
\end{enumerate}
\end{remarks}

\subsubsection{Partial difference systems}
In this subsection we derive some partial difference systems which have solutions in terms of generalised Umemura polynomials.
\begin{lemma}
If we define 
\beq\label{def:WVU} \wh{\mathcal{W}}_{m,n}^{(\k)}(z)=\frac{\wh{w}_{m,n}^{(\k)}(z)+1}{\wh{w}_{m,n}^{(\k)}(z)-1},\qquad 
\wh{\mathcal{V}}_{m,n}^{(\k)}(z)=\frac{\wh{v}_{m,n}^{(\k)}(z)+1}{\wh{v}_{m,n}^{(\k)}(z)-1},\qquad\wh{\mathcal{U}}_{m,n}^{(\k)}(z)=\frac{\wh{u}_{m,n}^{(\k)}(z)+1}{\wh{u}_{m,n}^{(\k)}(z)-1},\eeq
with $\wh{w}_{m,n}^{(\k)}(z)$, $\wh{v}_{m,n}^{(\k)}(z)$ and $\wh{u}_{m,n}^{(\k)}(z)$ as given in Theorem \ref{thm:genUm},
then they satisfy the partial difference systems
\begin{subequations}\begin{align}
\wh{\mathcal{W}}_{m,n}^{(\k+1)}+\wh{\mathcal{W}}_{m,n}^{(\k)}&=\frac{4}{z} \frac{(m+n+2\k+1) \wh{\mathcal{V}}_{m,n}^{(\k+n+1/2)}+m+n}{1-\Big(\wh{\mathcal{V}}_{m,n}^{(\k+n+1/2)}\Big)^{\!2}},\\
\wh{\mathcal{V}}_{m,n}^{(\k+n+1/2)}+\wh{\mathcal{V}}_{m,n}^{(\k+n-1/2)}&=\frac{4}{z} \frac{(m+n+2\k )\wh{\mathcal{W}}_{m,n}^{(\k)}+m-n}{1-\Big(\wh{\mathcal{W}}_{m,n}^{(\k)}\Big)^{\!2}},
\end{align}\end{subequations}
and
\begin{subequations}\begin{align}
\wh{\mathcal{V}}_{m,n}^{(\k+n)}+\wh{\mathcal{V}}_{m+1,n+1}^{(\k+n)}&=\frac{4}{z} \frac{(m+n+1)\,\wh{\mathcal{U}}_{m,n}^{(2\k)}+m-n}{1-\Big(\wh{\mathcal{U}}_{m,n}^{(2\k)}\Big)^{\!2}},\\
\wh{\mathcal{U}}_{m,n}^{(2\k)}+\wh{\mathcal{U}}_{m+1,n+1}^{(2\k-2)}&=\frac{4}{z} \frac{(m+n+2) \wh{\mathcal{V}}_{m+1,n+1}^{(\k+n)}+m+n+2\k}{1-\Big(\wh{\mathcal{V}}_{m+1,n+1}^{(\k+n)}\Big)^{\!2}}.
\end{align}\end{subequations}
\end{lemma}
\begin{proof}These follow from setting $N=0$ in \eqref{sys:513} and \eqref{sys:514}, respectively.
\end{proof}
\begin{lemma}
If $\wh{\mathcal{W}}_{m,n}^{(\k)}(z)$ and $\wh{\mathcal{U}}_{m,n}^{(\k)}(z)$ are given by \eqref{def:WVU},
\begin{subequations}\begin{align}
\wh{\mathcal{W}}_{m+1,n}^{(\k)}+\wh{\mathcal{W}}_{m,n+1}^{(\k)}&=\frac{4}{z} \frac{(m-n)\,\wh{\mathcal{U}}_{m,n}^{(2\k+1)}+m+n+1}{1-\Big(\wh{\mathcal{U}}_{m,n}^{(2\k+1)}\Big)^{\!2}},\\
\wh{\mathcal{U}}_{m-1,n}^{(2\k+1)}+\wh{\mathcal{U}}_{m,n-1}^{(2\k+1)}&=\frac{4}{z} \frac{(m-n )\wh{\mathcal{W}}_{m,n}^{(\k)}+m+n+2\k}{1-\Big(\wh{\mathcal{W}}_{m,n}^{(\k)}\Big)^{\!2}}.
\end{align}\end{subequations}
\end{lemma}
\begin{proof}From the transformations \eqref{sys:519} we obtain the sequences
\[\begin{array}{ccccc} \wh{w}_{m,n+1}^{(\k)} &\arr{\mathcal{R}_{1}} &\wh{u}_{m,n}^{(2\k+1)} &\arr{\mathcal{R}_{1}} &\wh{w}_{m+1,n}^{(\k)}\\
\wh{u}_{m-1,n}^{(2\k+1)} &\arr{\mathcal{R}_{1}} &\wh{w}_{m,n}^{(\k)} &\arr{\mathcal{R}_{1}} &\wh{u}_{m,n-1}^{(2\k+1)} 
\end{array}\]
from which the result follows. We note that
\[ \mathcal{R}_{1}^3\Big(\wh{w}_{m,n}^{(\k)} \Big)=\wh{w}_{m+1,n}^{(\k)},\qquad 
\mathcal{R}_{1}^3\Big(\wh{v}_{m,n}^{(\k)} \Big)=\wh{v}_{m+1,n}^{(\k)},\qquad \mathcal{R}_{1}^3\Big(\wh{u}_{m,n}^{(2\k+1)}\Big)=\wh{u}_{m+1,n}^{(2\k+1)}.\]
\end{proof}

\begin{lemma}
If $\wh{\mathcal{W}}_{m,n}^{(\k)}(z)$, $\wh{\mathcal{V}}_{m,n}^{(\k)}(z)$ and $\wh{\mathcal{U}}_{m,n}^{(\k)}(z)$ are given by \eqref{def:WVU}, then they satisfy the partial difference systems
\begin{align*} 
\frac{2m+1}{\wh{\mathcal{W}}_{m+1,n}^{(\k-1)}+\wh{\mathcal{V}}_{m,n}^{(\k+n-1/2)}}+\frac{2m-1}{\wh{\mathcal{W}}_{m-1,n}^{(\k)}+\wh{\mathcal{V}}_{m,n}^{(\k+n-1/2)}}&=z-\frac{2\Big[(m+n)\wh{\mathcal{V}}_{m,n}^{(\k+n-1/2)}+m+n+2\k-1\Big]}{1-\Big(\wh{\mathcal{V}}_{m,n}^{(\k+n-1/2)}\Big)^{\!2}},\\
\frac{2m+1}{\wh{\mathcal{W}}_{m,n}^{(\k)}+\wh{\mathcal{V}}_{m+1,n}^{(\k+n-1/2)}}+\frac{2m-1}{\wh{\mathcal{W}}_{m,n}^{(\k)}+\wh{\mathcal{V}}_{m-1,n}^{(\k+n+1/2)}}&=z-\frac{2\Big[(m-n)\wh{\mathcal{W}}_{m,n}^{(\k)}+m+n+2\k\Big]}{1-\Big(\wh{\mathcal{W}}_{m,n}^{(\k)}\Big)^{\!2}},
\end{align*}
and
\begin{align*} 
\frac{2m+2n+2\k+3}{\wh{\mathcal{W}}_{m+1,n+1}^{(\k)}+\wh{\mathcal{U}}_{m,n}^{(2\k+1)}}+\frac{2m+2n+2\k+1}{\wh{\mathcal{W}}_{m,n}^{(\k)}+\wh{\mathcal{U}}_{m,n}^{(2\k+1)}} &=z-\frac{2\Big[(m+n+1) \wh{\mathcal{U}}_{m,n}^{(2\k+1)}+m-n\Big]}{1-\Big(\wh{\mathcal{U}}_{m,n}^{(2\k+1)}\Big)^{\!2}},\\
\frac{2m+2n+2\k+1}{\wh{\mathcal{W}}_{m,n}^{(\k)}+\wh{\mathcal{U}}_{m,n}^{(2\k+1)}}+\frac{2m+2n+2\k-1}{\wh{\mathcal{W}}_{m,n}^{(\k)}+\wh{\mathcal{U}}_{m-1,n-1}^{(2\k+1)}} 
&=z-\frac{2\Big[(m+n+2 \k) \wh{\mathcal{W}}_{m,n}^{(\k)}+m-n\Big]}{1-\Big(\wh{\mathcal{W}}_{m,n}^{(\k)}\Big)^{\!2}}.
\end{align*}
\end{lemma}
\begin{proof}These follow from setting $N=0$ in \eqref{sys:515} and \eqref{sys:516}, respectively.
\end{proof}

\begin{lemma}
If we define
\beq \wh{\mathcal{X}}_{m,n}^{(\k)}(z)=\frac{1}{\wh{w}_{m,n}^{(\k)}(z)-1},\qquad \wh{\mathcal{Y}}_{m,n}^{(\k)}(z)=\frac{1}{\wh{v}_{m,n}^{(\k)}(z)-1}, \qquad
 \wh{\mathcal{Z}}_{m,n}^{(\k)}(z)=\frac{1}{\wh{u}_{m,n}^{(\k)}(z)-1},\eeq
with $\wh{w}_{m,n}^{(\k)}(z)$, $\wh{v}_{m,n}^{(\k)}(z)$ and $\wh{u}_{m,n}^{(\k)}(z)$ as given in Theorem \ref{thm:genUm},
then they satisfy the partial difference system
\begin{subequations}\begin{align}
&\wh{\mathcal{X}}_{m,n}^{(\k)}\Big(\wh{\mathcal{Y}}_{m,n}^{(\k+n+1/2)}+\wh{\mathcal{Z}}_{m-1,n}^{(2\k+1)}+1\Big)+\frac{m+\k}{z}=0,\\
&\wh{\mathcal{Y}}_{m,n}^{(\k+n+1/2)}\Big(\wh{\mathcal{Z}}_{m,n}^{(2\k+1)}+\wh{\mathcal{X}}_{m,n}^{(\k)}+1\Big)+\frac{m+n+\k+\hf}{z}=0,\\
& \wh{\mathcal{Z}}_{m,n}^{(2\k+1)}\Big(\wh{\mathcal{X}}_{m+1,n}^{(\k)}+\wh{\mathcal{Y}}_{m,n}^{(\k+n+1/2)}+1\Big)+\frac{m+\hf}{z}=0.
\end{align}\end{subequations}
\end{lemma}
\begin{proof}This follows from setting $N=0$ in \eqref{xyz:case1} and \eqref{xyz:sys1}.
\end{proof}

\subsection{Non-unique solutions of the discrete equations}
Kitaev, Law and McLeod \cite[Theorem 1.2]{refKLM} prove that if $\ga\not\in\Integer$ in case (i), then the rational solution of \PV\ is unique and if $\ga\in\Integer$ then there are at most two rational solutions. Clarkson and Dunning \cite{refCD24} gave examples of non-unique rational solutions which arise for solutions expressed in terms of generalised Laguerre polynomials discussed in \S\ref{ssec:GenLag}. 
{This is illustrated in the following example}.
\begin{example}
The rational solutions
\begin{align} w_{1,1}^{(1)}(z)&=-\frac{(z-3)(z^2-8z+20)}{(z-2)(z-6)},\qquad 
u_{1,2}^{(1)}(z)=\frac{(z^2+4z+6)(z^3+9z^2+36z+60)}{z^4+12z^3+54z^2+96z+72}, \label{sols:w11u12} 
\end{align}
both satisfy \PV\ \eqref{eq:pv} for the parameters $\pms{\hf}{-\tfrac{25}{2}}{1}$. Hence the solutions 
\begin{subequations}\label{sol:qp}\begin{align}
q_{1}(z)&=\frac{w_{1,1}^{(1)}(z)+1}{w_{1,1}^{(1)}(z)-1}= \frac{z^3-12z^2+52z-72}{(z-4)(z^2-6z+12)},\\
p_{1}(z)&=\frac{u_{1,2}^{(1)}(z)+1}{u_{1,2}^{(1)}(z)-1}=\frac{z^5 + 14z^4 + 90z^3 + 312z^2+552z+432}{(z^2+6z+12)(z^3+6z^2+18z+24)},
\end{align}\end{subequations}
both satisfy 
equation \eqref{eq2:q} for the parameters $\pms{\hf}{-\tfrac{25}{2}}{1}$. 
\end{example}

\begin{remarks}{\rm
\begin{enumerate}[(i)]\item[]
\item{The rational solutions \eqref{sol:qp} are special cases of the solution of
equation \eqref{eq2:q} for the parameters $\pms{\hf}{-\tfrac{25}{2}}{1}$ given by
\begin{subequations}\label{sol:ganrat1}\beq q(z)=\frac{f(z;C_{1},C_{2})}{g(z;C_{1},C_{2})},\eeq
where
\begin{align} 
f(z;C_{1},C_{2})&=2 C_{1}^{2} (z^{3}-12 z^{2}+52 z -72)\rme^{2 z}- C_{1} C_{2} (z^{5}-2 z^{4}+36 z^{3}+48 z^{2}-240 z -864)\rme^{z}\nonumber\\
&\qquad-3 C_{2}^{2} (z^{5}+14 z^{4}+90 z^{3}+312 z^{2}+552 z +432),\\
g(z;C_{1},C_{2})&=2 C_{1}^{2} (z -4) (z^{2}-6 z +12)\rme^{2 z}+C_{1} C_{2} (z^{5}-12 z^{3}-24 z^{2}+144 z +576)\rme^{z}\nonumber\\
&\qquad-3 C_{2}^{2} (z^{2}+6 z +12) (z^{3}+6 z^{2}+18 z +24),
\end{align} \end{subequations}
with $C_{1}$ and $C_{2}$ arbitrary constants.} 

\item{Special function solutions of \PV\ \eqref{eq:pv} are known to be expressible in terms of \textit{Kummer functions} $M(a,b,z)$ and $U(a,b,z)$, or equivalently \textit{Whittaker functions} $M_{\k,\mu}(z)$ and $W_{\k,\mu}(z)$, see, e.g.\ \cite{refMasuda04,refOkamotoPV,refWat}, also \cite[\S40]{refGLS} and \cite[\S32.10(v)]{refDLMF}.
Special cases of the Kummer functions $M(a,b,z)$ and $U(a,b,z)$, are expressed in terms of \textit{modified Bessel functions} $I_{\nu}(z)$ and $K_{\nu}(z)$, see \cite[\S13.6(iii)]{refDLMF}. 
Further, when $\nu=n+\hf$, with $n\in\Integer$, then modified Bessel functions become \textit{spherical Bessel functions} $I_{\pm(n+1/2)}(z)$, with $n\in\mathbb{N}$, which have the form
\beq I_{\pm(n+1/2)}(z)=\frac{1}{\sqrt{2\pi z}} \sum_{k=0}^{n} \frac{(n+k)!}{(2z)^{k}k!(n-k)!}\Big\{(-1)^k\rme^z\pm(-1)^{n+1}\rme^{-z}\Big\},\eeq
see \cite[\S10.49(ii)]{refDLMF}.
The solution
\beq\label{sol:Bessel1} q(z)=\frac{(2\nu+3)\Zn[2]{\nu+1}+(4\nu+4+z)\Zn{\nu+1}\Zn{\nu}+z\Zn[2]{\nu}}{\Zn{\nu+1}\big\{(2\nu+1+z)\Zn{\nu+1}-z\Zn{\nu}\big\}},\eeq
where
\beq \mathcal{Z}_{\nu}(z)=\begin{cases} c_{1}I_{\nu}(z)+c_{2}I_{-\nu}(z),&\text{if}\quad \nu\not\in\Integer,\\
c_{1}I_{n}(z)+c_{2}(-1)^nK_{n}(z),\quad&\text{if}\quad\nu=n\in\Integer,
\end{cases}\label{def:Zn}\eeq
with $c_{1}$ and $c_{2}$ arbitrary constants,
satisfies \eqref{eq2:q} for the parameters $\pms{\hf}{-2(\nu+1)^2}{1}$. Setting $\nu=\tfrac{3}{2}$ in \eqref{sol:Bessel1} gives \eqref{sol:ganrat1}, with $C_{1}$ and $C_{2}$ expressed in terms of $c_{1}$ and $c_{2}$.}
\end{enumerate}
}\end{remarks}

In this subsection, we use pairs of solutions such as \eqref{sol:qp} to derive examples of distinct hierarchies of rational solutions which satisfy the same discrete equation.

\subsubsection{First discrete equation: asymmetric discrete \PII}
Applying the \bt\ $\mathcal{R}_{1}$ \eqref{bt:R1}, with $(a,b,c)=(1,-5,1)$, to the solutions \eqref{sols:w11u12} gives the sequences of rational solutions
\begin{align*} 
&w_{1,1}^{(1)}\arr{\mathcal{R}_{1}}\frac{1}{{v}_{1,1}^{(4)}}\arr{\mathcal{R}_{1}} w_{2,1}^{(1)}\arr{\mathcal{R}_{1}}\frac{1}{{v}_{2,1}^{(4)}} \arr{\mathcal{R}_{1}} w_{3,1}^{(1)}\\
&u_{1,2}^{(1)} \arr{\mathcal{R}_{1}} {v}_{2,1}^{(4)}\arr{\mathcal{R}_{1}} u_{2,2}^{(1)}\arr{\mathcal{R}_{1}} v_{2,2}^{(4)} \arr{\mathcal{R}_{1}} u_{3,2}^{(1)}
\end{align*} 
Hence we obtain the hierarchies
\begin{subequations} \label{eq1:qhier}
\begin{align}
q_{2n}(z)&=\frac{1+v_{n,1}^{(4)}(z)}{1-v_{n,1}^{(4)}(z)}=1-\frac{8}{z}+2\Ln{\Tmn{n,1}{1}(z)}{\Tmn{n-1,2}{1}(z)},\\
q_{2n+1}(z)&=\frac{w_{n+1,1}^{(1)}(z)+1}{w_{n+1,1}^{(1)}(z)-1}=1+2\Ln{\Tmn{n,1}{2}(z)}{\Tmn{n,2}{0}(z)},
\end{align}\end{subequations} and 
\begin{subequations}\label{eq1:phier}\begin{align}
p_{2n}(z)&=\frac{v_{2,n}^{(-4)}(z)+1}{v_{2,n}^{(-4)}(z)-1}=-1-\frac{8}{z}+2\Ln{\Tmn{1,n+1}{-2n-5}(z)}{\Tmn{2,n}{-2n-5}(z)},\\
p_{2n+1}(z)&=\frac{u_{n+1,2}^{(1)}(z)+1}{u_{n+1,2}^{(1)}(z)-1}=1+2\Ln{\Tmn{2,n+1}{-2n-6}(z)}{\Tmn{1,n+1}{-2n-6}(z)},
\end{align}\end{subequations}
where $q_{2n}$ and $p_{2n}$ both satisfy equation \eqref{eq2:q} for the parameters $\pms{\hf(n+3)^2}{-\hf(n+2)^2}{-4}$ whilst
$q_{2n+1}$ and $p_{2n+1}$ both satisfy equation \eqref{eq2:q} for the parameters $\pms{\hf(n+1)^2}{-\hf(n+5)^2}{1}$.
We note that there is a difference in the structure of the solutions in the hierarchies \eqref{eq1:qhier} and \eqref{eq1:phier}. {The solutions $q_{n}$ \eqref{eq1:qhier} are expressed in terms of the generalised Laguerre polynomials $\Tmn{n,1}{\mu}(z)$ and $\Tmn{n,2}{\mu}(z)$ which are determinants of fixed size, whereas the solutions $p_{n}$ \eqref{eq1:phier} are expressed in terms of the generalised Laguerre polynomials $\Tmn{1,n}{\mu}(z)$ and $\Tmn{2,n}{\mu}(z)$ which are determinants that increase in size as $n$ increases.}

Both $q_{n}$ and $p_{n}$ satisfy the discrete equation
\beq q_{n+1}+q_{n-1}=\frac{4}{z}\frac{(n+5)q_{n}-\tfrac{3}{2}+\frac{5}{2}(-1)^n}{1-q_{n}^2},\eeq
which is asymmetric \dPII\ \eqref{eq:adPII} with $\la=5$, $\rho=-\tfrac{3}{2}$ and $\ph=\tfrac{5}{2}$.
Further, if either $(x_{n},y_{n})=(q_{2n},q_{2n+1})$ or 
$(x_{n},y_{n})=(p_{2n},p_{2n+1})$ then we obtain solutions to the discrete system
\begin{align*}
x_{n+1}+x_{n}&=\frac{4}{z}\frac{(2n+6)y_{n}-4}{1-y_{n}^2},\qquad y_{n}+y_{n-1}=\frac{4}{z}\frac{(2n+5)x_{n}+1}{1-x_{n}^2}.
\end{align*}
The first few solutions in the hierarchies \eqref{eq1:qhier} and \eqref{eq1:phier} are given by
\begin{align*}
q_{0}&=1-\frac{8}{z}+2\deriv{}{z}\ln(z-3),\\
q_{2}&=1-\frac{8}{z}+2\Ln{z^2-8z+12}{z^2-8z+20},\\
q_{4}&=1-\frac{8}{z}+2\Ln{z^{3}-15 z^{2}+60 z -60}{z^{4}-20 z^{3}+150 z^{2}-480 z +600},\\
q_{6}&=1-\frac{8}{z}+2\Ln{z^{4}-24 z^{3}+180 z^{2}-480 z +360}{z^{6}-36 z^{5}+522 z^{4}-3840 z^{3}+15120 z^{2}-30240 z +25200},\\[7.5pt]
q_{1}&=1+2\Ln{z-4}{z^2-6z+12},\\
q_{3}&=1+2\Ln{z^2-10z+20}{z^{4}-16 z^{3}+96 z^{2}-240 z +240},\\
q_{5}&=1+2\Ln{z^{3}-18 z^{2}+90 z -120}{z^{6}-30 z^{5}+360 z^{4}-2160 z^{3}+6840 z^{2}-10800 z +7200},\\ 
q_{7}&=1+2\Ln{z^{4}-28 z^{3}+252 z^{2}-840 z +840}{z^{8}-48 z^{7}+960 z^{6}-10320 z^{5}+64800 z^{4}-241920 z^{3}+524160 z^{2}-604800 z +302400},
\end{align*}
and
\begin{align*}
p_{0}&=-1-\frac{8}{z}+2\deriv{}{z}\ln(z^2+4z+6),\\
p_{2}&=-1-\frac{8}{z}+2\Ln{z^{4}+12 z^{3}+54 z^{2}+96 z +72}{z^{3}+9 z^{2}+36 z +60},\\
p_{4}&=-1-\frac{8}{z}+2\Ln{z^{6}+24 z^{5}+228 z^{4}+1056 z^{3}+2520 z^{2}+2880 z +1440}{z^{6}+24 z^{5}+252 z^{4}+1440 z^{3}+4680 z^{2}+8640 z +7200},\\[7.5pt]
p_{1}&=1+2\Ln{z^{3}+6 z^{2}+18 z +24}{z^{2}+6 z +12},\\
p_{3}&=1+2\Ln{z^{6}+18 z^{5}+144 z^{4}+624 z^{3}+1512 z^{2}+2160 z +1440}{z^{4}+16 z^{3}+96 z^{2}+240 z +240},\\
p_{5}&=1+2\Ln{z^{9}+36 z^{8}+576 z^{7}+5280 z^{6}+30240 z^{5}+112320 z^{4}+270720 z^{3}+414720 z^{2}+388800 z +172800}{z^{6}+30 z^{5}+360 z^{4}+2160 z^{3}+6840 z^{2}+10800 z +7200},
\end{align*}
with $q_{-1}(z)=p_{-1}(z)=1$.

\begin{remark}{\rm{The solutions $q_{0}$ and $p_{0}$ are special cases of the solution
\[ q(z)=\frac{c_{1} (z^{2}-9 z +24) \rme^{z}-c_{2}(z^{3}+8 z^{2}+30 z +48)}{c_{1} z(z -3)\rme^{z}+c_{2} z(z^{2}+4 z +6)},\]
which satisfies equation \eqref{eq2:q} for the parameters $\pms{\tfrac{9}{2}}{-2}{-4}$ and is obtained by setting $\nu=\tfrac{3}{2}$ is the solution
\[q(z)= \frac{(4\nu+2+z)\Zn{\nu+1}-z\Zn{\nu}}{z\big\{\Zn{\nu+1}-\Zn{\nu}\big\}},\]
which satisfies equation \eqref{eq2:q} for the parameters $\pms{\hf(\nu+\tfrac{3}{2})^2}{-\hf(\nu+\hf)^2}{-2\nu-1}$.
}}\end{remark}

\subsubsection{Second discrete equation}
The rational solutions 
\[v_{0,n}^{(m+n+1)}(z)=(-1)^n\frac{m+1}{n+1}\,\frac{\Tmn{0,n}{m-n+1}(z)}{\Tmn{0,n+1}{m-n-1}(z)},\qquad 
\frac{1}{v_{m,0}^{(-m-n-1)}}(z) =-\frac{m+1}{n+1}\,\frac{\Tmn{m,1}{-m-n-3}(z)}{\Tmn{m-1,1}{-m-n-3}(z)}, \]
both satisfy \PV\ \eqref{eq:pv} for the parameters $\pms{\hf(n+1)^2}{-\hf(m+1)^2}{m+n+1}$. In the case when $m=n-1$, applying the \bt\ $\mathcal{R}_{2}$ \eqref{bt:R2} gives the hierarchy of solutions
\begin{subequations} \label{eq2:qhier}
\begin{align}
q_{2n}(z)&=\frac{1+v_{n,0}^{(-2n)}(z)}{1-v_{n,0}^{(-2n)}(z)}
= 1+\frac{4n}{z}-2\dln \Tmn{n-1,1}{-2n-1}(z),\\ q_{2n+1}(z)&=\frac{1+\wt{v}_{n,0}^{(-2n)}(z)}{1-\wt{v}_{n,0}^{(-2n)}(z)} = -1+2\dln \Tmn{n,1}{-2n-3}(z),
\end{align}\end{subequations}
and
\begin{subequations} \label{eq2:phier}\begin{align}
p_{2n}(z)&=\frac{v_{0,n-1}^{(2n)}(z)+1}{v_{0,n-1}^{(2n)}(z)-1}
=-1+\frac{4n}{z}-2\dln \Tmn{0,n-1}{1}(z)\\ 
p_{2n+1}(z)&=\frac{\wt{v}_{0,n-1}^{(2n)}(z)-1}{\wt{v}_{0,n-1}^{(2n)}(z)+1 }=1+2\dln \Tmn{0,n}{1}(z),
\end{align}\end{subequations}
where
\begin{align*}
{v}_{n,0}^{(\mu)}(z)&=\frac{n+\mu}{n+1}\frac{\Tmn{n-1,1}{\mu-2}(z)}{\Tmn{n,1}{\mu-2}(z)}=\frac{n+\mu}{n+1}\,\frac{\LaguerreL{\mu-1}{n}(z)}{\LaguerreL{\mu-1}{n+1}(z)},\\
\wt{v}_{n,0}^{(\mu)}(z)&=\frac{\Tmn{n,1}{\mu-2}(z)}{\Tmn{n-1,1}{\mu-2}(z)}=\frac{\LaguerreL{\mu-1}{n+1}(z)}{\LaguerreL{\mu-1}{n}(z)},\\
{v}_{0,n}^{(\mu)}(z)&=(-1)^n\frac{\mu-n}{n+1}\,\frac{\Tmn{0,n}{\mu-2n}(z)}{\Tmn{0,n+1}{\mu-2n-2}(z)}={\frac{n-\mu}{n+1}\,\frac{L_{n}^{(-\mu-1)}(-z)}{L_{n+1}^{(-\mu-1)}(-z)}}, \\
\wt{v}_{0,n}^{(\mu)}(z)&=(-1)^{n+1}\frac{\Tmn{0,n+1}{\mu-2n-2}(z)}{\Tmn{0,n}{\mu-2n}(z)}=\frac{L_{n+1}^{(-1-\mu)}(-z)}{L_n^{(-1-\mu)}(-z)}.
\end{align*}
The solutions $q_{2n}$ and $p_{2n}$ both satisfy equation \eqref{eq2:q} for the parameters $\pms{\hf n^2}{-\hf(n+1)^2}{2n}$ whilst
$q_{2n+1}$ and $p_{2n+1}$ both satisfy equation \eqref{eq2:q} for the parameters $\pms{\hf n^2}{-\hf(n+1)^2}{2n+2}$. Again there is a difference in the structure of the solutions in the hierarchies \eqref{eq2:qhier} and \eqref{eq2:phier}.

Both $q_{n}$ and $p_{n}$ satisfy the discrete equation
\[\frac{2n+2}{q_{n+1}+q_{n}}+\frac{2n}{q_{n}+q_{n-1}}=z-\frac{\big[2n+1+(-1)^n\big]q_{n}-2}{1-q_{n}^2},\]
which is equation \eqref{eq:dPn} with $\la=\hf$, $\rho=-2$ and $\ph=\hf$.
Further, if either $(x_{n},y_{n})=(q_{2n},q_{2n+1})$ or 
$(x_{n},y_{n})=(p_{2n},p_{2n+1})$ then we obtain solutions of the discrete system
\begin{align*}
\frac{2n+2}{x_{n+1}+y_{n}}+\frac{2n+1}{x_{n}+y_{n}}&=\hf z-\frac{(2n+1)y_{n}-1}{1-y_{n}^2},\\
\frac{2n+1}{x_{n}+y_{n}}+\frac{2n}{x_{n}+y_{n-1}}&=\hf z-\frac{(2n+1)x_{n}-1}{1-x_{n}^2}.
\end{align*}

The first few solutions in the hierarchies \eqref{eq2:qhier} and \eqref{eq2:phier} are given by
\[\begin{split}
q_{2} &= 1+\frac{4}{z}-\frac{2}{z+1},\\ 
q_{4} &= 1+\frac{8}{z}-\frac{4 (2+z )}{z^{2}+4 z +6}, \\q_{6} &= 
1+\frac{12}{z}-\frac{6 (z^{2}+6 z +12)}{z^{3}+9 z^{2}+36 z +60},\\ 
q_{8} &= 1+\frac{16}{z}-\frac{8(z^{3}+12 z^{2}+60 z +120)}{z^{4}+16 z^{3}+120 z^{2}+480 z +840},
\end{split} \qquad \begin{split}
q_{1} &= -1+\frac{2}{z+1},\\ q_{3} &= -1+\frac{4(z+2)}{z^{2}+4 z +6},\\
 q_{5} &= -1+\frac{6(z^{2}+6 z +12)}{z^{3}+9 z^{2}+36 z +60},\\
q_{7} &= 
-1+\frac{8(z^{3}+12 z^{2}+60 z +120)}{z^{4}+16 z^{3}+120 z^{2}+480 z +840},
\end{split} 
\] and \[
\begin{split}
p_{2} &= -1+\frac{4}{z}, \\
p_{4} &= -1+\frac{8}{z}-\frac{2}{z-3}, \\
p_{6} &= -1+\frac{12}{z}-\frac{4(z-4)}{z^{2}-8 z +20}, \\
p_{8} &= -1+\frac{16}{z}-\frac{6(z^{2}-10 z +30)}{z^{3}-15 z^{2}+90 z -210},
\end{split} \qquad \begin{split}
p_{1} &= 1, \\
p_{3} &= 1+\frac{2}{z-3},\\
p_{5} &= 1+\frac{4(z-4)}{z^{2}-8 z +20}, \\
p_{7} &= 1+\frac{6(z^{2}-10 z +30)}{z^{3}-15 z^{2}+90 z -210}. 
\end{split}\]

\subsubsection{Third discrete equation: ternary discrete \PI}
The rational solutions 
\beq v_{3,1}^{(-1)}(z)=-\frac{(z-2)(z-6)}{(z-3)(z^2-8z+20)},\qquad v_{2,2}^{(-1)}(z)=\frac{z^4+12z^3+54z^2+96z+72}{(z^2+4z+6)(z^3+9z^2+36z+60)},\label{sol:v31v22} \eeq
both satisfy \PV\ \eqref{eq:pv} for the parameters $\pms{\tfrac{25}{2}}{-\hf}{-1}$ and so
\begin{subequations}\label{sol:xy}\begin{align}x(z)&=\frac{1}{v_{3,1}^{(-1)}(z)-1}=-\frac{(z-3)(z^{2}-8 z +20)}{(z-4)(z^{2}-6 z +12)},\\
y(z)&=\frac{1}{v_{2,2}^{(-1)}(z)-1}=-\ds\frac{(z^{2}+4 z +6)(z^{3}+9 z^{2}+36 z +60)}{(z^{2}+6 z +12)(z^{3}+6 z^{2}+18 z +24)},
\end{align}\end{subequations}
both satisfy equation \eqref{eq2:x}
for the parameters $\pms{\tfrac{25}{2}}{-\hf}{-1}$. We note that
\[v_{3,1}^{(-1)}(z)= \ifrac{1}{w_{1,1}^{(1)}(z)},\qquad v_{2,2}^{(-1)}(z)= \ifrac{1}{u_{1,2}^{(1)}(z)} .\]

Applying the \bt\ $\mathcal{R}_{3}$ \eqref{bt:R3}, with $(a,b,c)=(5,1,-1)$, to the solutions \eqref{sol:v31v22} gives the sequences of rational solutions
\begin{align*}
&{v}_{3,1}^{(-1)} \arr{\mathcal{R}_{3}} u_{2,3}^{(6)} \arr{\mathcal{R}_{3}}\frac{1}{{w}_{3,2}^{(-5)}}
\arr{\mathcal{R}_{3}}{v}_{3,2}^{(0)} \arr{\mathcal{R}_{3}} u_{3,3}^{(7)} \arr{\mathcal{R}_{3}}\frac{1}{{w}_{3,3}^{(-6)}}
\arr{\mathcal{R}_{3}}{v}_{3,3}^{(1)} \\
&{v}_{2,2}^{(-1)} \arr{\mathcal{R}_{3}}\frac{1}{{w}_{3,2}^{(-6)}} \arr{\mathcal{R}_{3}} u_{3,2}^{(5)} 
\arr{\mathcal{R}_{3}}{v}_{2,3}^{(0)} \arr{\mathcal{R}_{3}}\frac{1}{{w}_{3,3}^{(-7)}} \arr{\mathcal{R}_{3}} u_{4,2}^{(6)} 
\arr{\mathcal{R}_{3}}{v}_{2,4}^{(1)} 
\end{align*}
Hence we obtain the solution hierarchies
\begin{subequations}\label{eq3:xhier} \begin{align}
&x_{3n}(z)= -1+\frac{n+3}{z}-\dln\Tmn{0,n}{2-n}(z)= -1+\frac{n+3}{z}-\dln L_n^{(-n-3)} (-z),\\
&x_{3n+1}(z)=-1+\Ln{\Tmn{0,n+1}{1-n}(z)}{\Tmn{0,n}{3-n}(z)}={-1+\Ln{L_{n+1}^{(-n-4)}(-z)}{L_{n}^{(-n-4)}(-z)}},\\
&x_{3n+2}(z)=\dln\Tmn{0,n+1}{2-n}(z)= {\dln L_{n+1}^{(-n-5)}},
\end{align}\end{subequations}
and
\begin{subequations}\label{eq3:yhier} \begin{align} 
&y_{3n}(z)=\frac{n+3}{z}-\dln\Tmn{1,1}{-n-4}(z)=\frac{n+3}{z}-\dln L_2^{(-n-3)}(z),\\
&y_{3n+1}(z)=-1+\Ln{\Tmn{1,1}{-n-5}(z)}{\Tmn{2,1}{-n-5}(z)}=-1+\Ln{L_2^{(-n-4)}(z)}{L_3^{(-n-4)}(z)},\\
&y_{3n+2}(z)=-1+\dln\Tmn{2,1}{-n-6}(z)=-1+\dln L_3^{(-n-5)}(z),
\end{align}\end{subequations}
where $x_{n}$ and $y_{n}$ both satisfy equation \eqref{eq2:x} for the parameters $\pms{\hf(n+1)^2}{-\tfrac{9}{2}}{n+3}$,
$x_{3n+1}$ and $ y_{3n+1}$ for the parameters $\pms{\hf(n+4)^2}{-\hf}{n-2}$ and
$x_{3n+2}$ and $y_{3n+2}$ for the parameters $\pms{\hf(n+1)^2}{-\tfrac{9}{2}}{n+5}$.
We remark that whilst $x_{n}$ is only defined for $n\geq-1$, $y_{n}$ is defined for all $n\in\Integer$. Also the solutions $x_{n}$ only involve the generalised Laguerre polynomials $\Tmn{0,n}{\mu}(z)=(-1)^{\lceil n/2 \rceil} L_n^{(-\mu-2n-1)}(-z)$,
whereas the solutions $y_{n}$ only involve the generalised Laguerre polynomials $\Tmn{1,1}{\mu}(z)=L_2^{(\mu+1)}(z)$ and $\Tmn{2,1}{\mu}(z)=L_3^{(\mu+1)}(z)$.
Both $x_{n}$ and $y_{n}$ satisfy the ternary \dPI\ equation
\[x_n(x_{n+1}+x_{n-1}+1)+\frac{a_{n}}{z}=0,\]
and the equation
\[ \frac{a_{n+1}}{x_{n+2}+x_{n}+1}+\frac{a_{n-1}}{x_{n}+x_{n-2}+1}=z+\frac{a_{n}}{x_{n}},\]
with
\[a_n=\tfrac{1}{3}n+\tfrac{5}{3}-\tfrac{2}{3}\cos(\tfrac{2}{3}\pi n)+\tfrac{10}{9}\sqrt{3}\,\sin(\tfrac{2}{3}\pi n).\]
Further if $(X_{n},Y_{n},Z_{n})=(x_{3n},x_{3n+1},x_{3n+2})$ or 
$(X_{n},Y_{n},Z_{n})=(y_{3n},y_{3n+1},y_{3n+2})$, then we obtain solutions of the discrete system
\begin{align*}
&X_{n}\big(Y_{n}+Z_{n-1}+1\big)+\frac{n+1}{z}=0,\quad
Y_{n}\big(Z_{n}+X_{n}+1\big)+\frac{n+4}{z}=0,\quad
Z_{n}\big(X_{n+1}+Y_{n}+1\big)+\frac{n+1}{z}=0.
\end{align*}
The first few solutions in the hierarchies \eqref{eq3:xhier} and \eqref{eq3:yhier} are given by
\[\begin{split}
x_{0}&=-1+\frac{3}{z}, \\
x_{3}&=-1+\frac{4}{z}-\frac{1}{z-3},\\
x_{6}&=-1+\frac{5}{z}-\frac{2(z-3)}{z^2-6z+12},
\end{split}\qquad\begin{split}
x_{1}&=-1+\frac{1}{z-3},\\
x_{4}&=-1+\Ln{z^2-6z+12}{z-4},\\
x_{7}&=-1+\Ln{z^3-9z^2+36z-60}{z^2-8z+20},
\end{split}\qquad\begin{split}
x_{2}&=\frac{1}{z-4},\\ x_{5}&=\frac{2(z-4)}{z^2-8z+20},\\ x_{8}&= \frac{3(z^2-8z+20)}{z^3-12z^2+60z-120},
\end{split}\]
and
\[\begin{split}
y_{0}&=\frac{3}{z}-\frac{2(z+1)}{z^2+2z+2},\\
y_{3}&=\frac{4}{z}-\frac{2(z+2)}{z^2+4z+6},\\
y_{6}&=\frac{5}{z}-\frac{2(z+3)}{z^2+6z+12}, 
\end{split}\qquad\begin{split}
 y_{1}&=-1+\Ln{z^{2}+4 z +6}{z^{3}+3 z^{2}+6 z +6},\\
 y_{4}&=-1+\Ln{z^2+6z+12}{z^3+6z^2+18z+24},\\
 y_{7}&=-1+\Ln{z^2+8z+20}{z^3+9z^2+36z+60},
\end{split}\qquad\begin{split}
y_{2}&=-1+\frac{3(z^2+4z+6)}{z^3+6z^2+18z+24},\\
y_{5}&=-1+\frac{3(z^2+6z+12)}{z^3+9z^2+36z+60},\\
y_{8}&=-1+\frac{3(z^2+8z+20)}{z^3+12z^2+60z+120}.
\end{split}\]

\section{\label{sec:dis}Discussion}
In this paper we have derived discrete equations from the \bts\ of \PV\ \eqref{eq:pv} and obtained hierarchies of rational solutions of the discrete equations. 
Whilst the asymmetric \dPII\ \eqref{eq:adPII}, the discrete equation \eqref{eq:dPn} and the ternary \dPI\ \eqref{eq:tdPI} were derived by 
Tokihiro, Grammaticos and Ramani \cite{refTGR02} from \bts\ of \PV\ \eqref{eq:pv}, our approach gives the explicit relationship between solutions of the discrete equations and solutions of \PV. This enabled us to derive the hierarchies of rational solutions of the discrete equations.

Tokihiro, Grammaticos and Ramani \cite{refTGR02} also derive the discrete system
\beq\label{TGReq324} x_{n}+x_{n-1}=\frac{1}{y_{n}}+\frac{n+\la+\mu}{1-y_{n}},\qquad y_{n+1}y_{n}=\frac{x_{n}-n-\la}{x_{n}^2-\ph^2},\eeq 
with $\la$, $\mu$ and $\ph$ constants,
which was derived earlier by Fokas, Grammaticos and Ramani \cite{refFGR} from the Schlesinger transformations of \PV\ \eqref{eq:pv}. 
Tokihiro, Grammaticos and Ramani \cite{refTGR02} state that whilst $y_{n}$ satisfies an equation which is equivalent to \PV, $x_{n}$ satisfies a second-order, fourth-degree equation. Schlesinger transformations of \PV\ \eqref{eq:pv} are transformations which are quadratic in $\rmd w/\rmd z$. 
For example, consider the Schlesinger transformations $R_{1}$ and $R_{2}$ of \PV\ in \cite{refFMA,refMF}, which Fokas, Grammaticos and Ramani \cite{refFGR} use to derive the discrete system \eqref{TGReq324}. 
Suppose that $w_{0}=w(z;a_{0},b_{0},c_{0})$ satisfies \PV\ for the parameters $\pms{\hf a_{0}^2}{-\hf b_{0}^2}{c_{0}}$, then the Schlesinger transformations $R_{1}$ and $R_{2}$ are given by
\begin{align*}
w_{1}&=R_{1}(w_{0})=\frac{\big\{zw_{0}'+a_{0}w_{0}^2-(a_{0}+b_{0}+z)w_{0}+b_{0}\big\}\big\{zw_{0}'+a_{0}w_{0}^2-(a_{0}-b_{0}+z)w_{0}-b_{0}\big\}}{z^{2}\big(w_{0}'\big)^{\!2} -2 a_{0} z w_{0}(w_{0} -1)w_{0}'+(a_{0}^2w_{0}^2-b_{0}^2)(w_{0}-1)^2-2z(a_{0}-c_{0}+1)w_{0}(w_{0}-1)-z^2w_{0}^2},\\
w_{2}&=R_{2}(w_{0})=\frac{\big\{zw_{0}'-a_{0}w_{0}^2+(a_{0}+b_{0}+z)w_{0}-b_{0}\big\}\big\{zw_{0}'-a_{0}w_{0}^2+(a_{0}-b_{0}+z)w_{0}+b_{0}\big\}}{z^{2}\big(w_{0}'\big)^{\!2}+2 a_{0} z w_{0}(w_{0} -1)w_{0}'+(a_{0}^2w_{0}^2-b_{0}^2)(w_{0}-1)^2-2z(a_{0}-c_{0}-1)w_{0}(w_{0}-1)-z^2w_{0}^2},
\end{align*}
with $'=\rmd/\rmd z$, and the effect on the parameters is
\[R_{1}\big(a_{0},b_{0},c_{0}\big)=\big(a_{0}+1,b_{0},c_{0}-1\big),\qquad R_{2}\big(a_{0},b_{0},c_{0}\big)=\big(a_{0}-1,b_{0},c_{0}+1\big).
\]
Hence $w_{1}$ and $w_{2}$ respectively satisfy \PV\ \eqref{eq:pv} for the parameters $\pms{\hf (a_{0}+1)^2}{-\hf b_{0}^2}{c_{0}-1}$ and $\pms{\hf (a_{0}-1)^2}{-\hf b_{0}^2}{c_{0}+1}$, so that $R_{1}R_{2}=R_{2}R_{1}=\mathcal{I}$, the identity transformation. In terms of the \bts\ $\T_{\ep_{1},\ep_{2},\ep_{3}}$ defined in \S\ref{sec:bts}, it is straightforward to show that
\[ R_{1}\equiv\T_{1,1,1}^2\circ\T_{-1,-1,-1},\qquad R_{2}\equiv\T_{-1,-1,-1}^2\circ\T_{1,1,1} .\]
All the Schlesinger transformations of \PV\ \eqref{eq:pv} given in \cite{refFMA,refMF} can be written as a sequence of the \bts\ $\T_{\ep_{1},\ep_{2},\ep_{3}}$. The Schlesinger transformations in \cite{refFMA,refMF} are written in terms of the monodromy parameters $\th_{0}$, $\th_{1}$ and $\th_{\infty}$ which are related to the parameters $a_{0}$, $b_{0}$ and $c_{0}$ through \eqref{SchlesParams}, i.e.
\[\th_{0} = \hf(a_{0}+b_{0}-c_{0}-1),\qquad \th_{1} =- \hf(a_{0}+b_{0}+c_{0}-1),\qquad\th_{\infty} = a_{0}-b_{0},\]
so that
\[a_{0}=\hf(\th_{0}-\th_{1}+\th_{\infty}),\qquad b_{0}=\hf(\th_{0}-\th_{1}-\th_{\infty}),\qquad c_{0}=1-\th_{0}-\th_{1}.\]

Exact solutions of the discrete equations \eqref{eq:adPII}, \eqref{eq:dPn} and \eqref{eq:tdPI} in terms of Kummer functions and modified Bessel functions, which are a special case of Kummer functions, cf.~\cite[\S13.6(iii)]{refDLMF}, using the explicit relationship between these discrete equations and \PV\ \eqref{eq:pv}, is currently under investigation, and we do not pursue this further here.

\subsection*{Acknowledgements}
We thank Andy Hone, Edoardo Peroni and Alexander Stokes for stimulating discussions.
BM is grateful to EPSRC for the support of Doctoral Training Partnerships grant EP/W524050/1. {We also thank the reviewers for their careful reading of the manuscript and their comments which have improved it.}

{\section*{Conflict of interest statement}
The authors declare no conflicts of interest.
\section*{Data availability statement}
No data were used for the research described in the article.}
\subsection*{ORCID iDs} 
\begin{tabular}{ll}
Peter A.\ Clarkson & 0000-0002-8777-5284 \\
Clare Dunning & 0000-0003-0535-9891
\end{tabular}
\appendix
\section{Properties of the generalised Laguerre polynomials and the generalised Umemura polynomials}
For convenience, in this appendix we list various differential-difference equations and discrete equations satisfied by the generalised Laguerre polynomials $\Tmn{m,n}{\mu}(z)$ and the generalised Umemura polynomials $\Umn{m,n}{\k}(z)$.
\subsection{Properties of the generalised Laguerre polynomials}
Clarkson and Dunning \cite[\S3]{refCD24} prove the following differential-difference equations for the generalised Laguerre polynomials
\begin{subequations}\label{eq324}\begin{align}
&\Hir{\Tmn{m,n-1}{\mu+1}}{\Tmn{m,n}{\mu}}=\Tmn{m+1,n-1}{\mu}\,\Tmn{m-1,n}{\mu+1},\label{iden:324a}\\
&\Hir{\Tmn{m,n-1}{\mu+1}}{\Tmn{m+1,n}{\mu-1}}=\Tmn{m+1,n-1}{\mu}\,\Tmn{m,n}{\mu},\label{iden:324b}\\
&\Hir{\Tmn{m,n-1}{\mu+1}}{\Tmn{m,n}{\mu-1}}=\Tmn{m+1,n-1}{\mu}\,\Tmn{m-1,n}{\mu},\label{iden:324c}\\
&\Hir{\Tmn{m+1,n}{\mu}}{\Tmn{m,n}{\mu+1}}=\Tmn{m+1,n-1}{\mu+1}\,\Tmn{m,n+1}{\mu},\label{iden:324d}\\
&\Hir{\Tmn{m,n}{\mu}}{\Tmn{m,n}{\mu+1}}=\Tmn{m+1,n-1}{\mu+1}\,\Tmn{m-1,n+1}{\mu},\label{iden:324e}\\
&\Hir{\Tmn{m+1,n}{\mu}}{\Tmn{m,n}{\mu}}=\Tmn{m+1,n-1}{\mu+1}\,\Tmn{m,n+1}{\mu-1},\label{iden:324f}
\end{align}\end{subequations}
where $\D_{z}(f\cdot g)$ is the Hirota operator \eqref{Hirop},
and the discrete equations
\begin{subequations}\begin{align}\Tmn{m,n+1}{\mu-1}\, \Tmn{m,n-1}{\mu+1}&=\Tmn{m+1,n}{\mu-1}\Tmn{m-1,n}{\mu+1}-\big(\Tmn{m,n}{\mu}\big)^{\!2},\label{iden:Tmn1}\\
 {\Tmn{m+1,n}{\mu}\,\Tmn{m,n}{\mu}} &={\Tmn{m+1,n}{\mu-1}\,\Tmn{m,n}{\mu+1}}-\Tmn{m,n+1}{\mu-1} \,\Tmn{m+1,n-1}{\mu+1}, \label{iden:Tmn2}\\ 
 \Tmn{m,n}{\mu-1}\,\Tmn{m,n-1}{\mu+1}&=\Tmn{m,n}{\mu}\,\Tmn{m,n-1}{\mu}-\Tmn{m-1,n}{\mu}\,\Tmn{m+1,n-1}{\mu}.\label{iden:Tmn3}
 \end{align}\end{subequations}

The following differential-difference equations arise from the definitions of $w_{m,n}^{(\mu)}$, $v_{m,n}^{(\mu)}$ and $u_{m,n}^{(\mu)}$ in Theorem \ref{thm:genLag} and Remarks \ref{rem:45}
\begin{subequations}\label{eq:A3}\begin{align}
z\Hir{\Tmn{m-1,n}{\mu}}{\Tmn{m+1,n-1}{\mu}}&=(m+1)\Tmn{m,n-1}{\mu}\,\Tmn{m,n}{\mu}+(z-m-2n-\mu)\Tmn{m-1,n}{\mu}\Tmn{m+1,n-1}{\mu},\label{eq:A3a}\\
z\Hir{\Tmn{m-1,n-1}{\mu+1}}{\Tmn{m,n}{\mu-1}}&=(m+n+\mu)\Tmn{m,n-1}{\mu+1}\,\Tmn{m-1,n}{\mu-1}-(m+n)\Tmn{m,n}{\mu-1}\,\Tmn{m-1,n-1}{\mu+1},\label{eq:A3b}\\
z\Hir{\Tmn{m,n-1}{\mu+1}}{\Tmn{m-1,n+1}{\mu-1}}&=n\Tmn{m-1,n}{\mu+1}\,\Tmn{m,n}{\mu-1}-(z-n-\mu)\Tmn{m,n-1}{\mu+1}\,\Tmn{m-1,n+1}{\mu-1},\label{eq:A3c}
\end{align}\end{subequations}
and so from \eqref{iden:324b} and \eqref{eq:A3b}
\beq z\Tmn{m,n-1}{\mu}\,\Tmn{m-1,n}{\mu} =(m+n+\mu)\Tmn{m,n-1}{\mu+1}\,\Tmn{m-1,n}{\mu-1}-(m+n)\Tmn{m,n}{\mu-1}\,\Tmn{m-1,n-1}{\mu+1}.\label{eq:A4} \eeq

The generalised Laguerre polynomials $\Tmn{m,n}{\mu}(z)$ satisfy the differential-difference equations 
\begin{subequations}\begin{align}
z\D_{z}^2\Big(\Tmn{m,n}{\mu}\cdot\Tmn{m,n}{\mu}\Big)
+ 2\Tmn{m,n}{\mu}\deriv{\Tmn{m,n}{\mu}}{z}
&=2(\mu+m+n+1)\Tmn{m+1,n-1}{\mu+1}\Tmn{m-1,n+1}{\mu-1},\\
z\D_{z}^2\Big(\Tmn{m,n}{\mu}\cdot\Tmn{m,n}{\mu}\Big)
+ 2\Tmn{m,n}{\mu}\deriv{\Tmn{m,n}{\mu}}{z}
&=2(\mu+m+n+1)\left\{ \Tmn{m,n}{\mu+1}\Tmn{m,n}{\mu-1} -\left(\Tmn{m,n}{\mu}\right)^{\!2} \right\},
\end{align}\end{subequations}
where
\[ \D_z^2\big(f\cdot g\big)=g\deriv[2]{f}{z}-2\deriv{f}{z}\deriv{g}{z}+f\deriv{g}{z},\]
from which we obtain the identity
\beq \Tmn{m+1,n-1}{\mu+1}\Tmn{m-1,n+1}{\mu-1} = \Tmn{m,n}{\mu+1}\Tmn{m,n}{\mu-1} -\left(\Tmn{m,n}{\mu}\right)^{\!2} .\eeq

\subsection{Properties of generalised Umemura polynomials}
Masuda, Ohta and Kajiwara \cite{refMOK}, see also \cite{refMasuda25}, prove the following discrete equations for the generalised Umemura polynomials
\begin{subequations}\begin{align}
\Umn{m+1,n}{\k-1}\,\Umn{m-1,n+1}{\k+1} &=\Umn{m+1,n+1}{\k-1}\,\Umn{m-1,n}{\k+1}+2(2m+1)\,\Umn{m,n+1}{\k-1}\,\Umn{m,n}{\k+1}, \\ 
\Umn{m,n+1}{\k-1}\,\Umn{m+1,n-1}{\k+1} &=\Umn{m+1,n+1}{\k-1}\,\Umn{m,n-1}{\k+1}-2(2n+1)\,\Umn{m+1,n}{\k-1}\,\Umn{m,n}{\k+1}, \\ 
\Umn{m-1,n}{\k}\,\Umn{m+1,n-1}{\k} &=z\,\Umn{m,n}{\k-2}\,\Umn{m,n-1}{\k+2}-2(2m+\k)\,\Umn{m,n}{\k}\,\Umn{m,n-1}{\k},\\
\Umn{m,n-1}{\k}\,\Umn{m-1,n+1}{\k} &=z\,\Umn{m,n}{\k-2}\,\Umn{m-1,n}{\k+2}+2(2n+\k)\,\Umn{m,n}{\k}\,\Umn{m-1,n}{\k},\\
\Umn{m+1,n}{\k-2}\,\Umn{m-1,n-1}{\k+2} &=z\,\Umn{m,n}{\k-2}\,\Umn{m,n-1}{\k+2}-2(\k-1)\,\Umn{m,n}{\k}\,\Umn{m,n-1}{\k},\\
\Umn{m,n+1}{\k-2}\,\Umn{m-1,n-1}{\k+2} &=z\,\Umn{m,n}{\k-2}\,\Umn{m-1,n}{\k+2}+2(\k-1)\,\Umn{m,n}{\k}\,\Umn{m-1,n}{\k},\\
\Umn{m-1,n+1}{\k-1}\,\Umn{m,n-1}{\k+1} &=z\,\Umn{m-1,n}{\k-1}\,\Umn{m,n}{\k+1}+2(2m+\k-1)\,\Umn{m,n}{\k-1}\,\Umn{m-1,n}{\k+1},\\
\Umn{m+1,n-1}{\k-1}\,\Umn{m-1,n}{\k+1} &=z\,\Umn{m,n-1}{\k-1}\,\Umn{m,n}{\k+1}-2(2n+\k-1)\,\Umn{m,n}{\k-1}\,\Umn{m,n-1}{\k+1},\\
\Umn{m+1,n}{\k-1}\,\Umn{m-1,n-1}{\k+1} &=z\,\Umn{m,n-1}{\k-1}\,\Umn{m,n}{\k+1}-2(2m+2n+\k)\,\Umn{m,n}{\k-1}\,\Umn{m,n-1}{\k+1}, \\ 
\Umn{m,n+1}{\k-1}\,\Umn{m-1,n-1}{\k+1} &=z\,\Umn{m-1,n}{\k-1}\,\Umn{m,n}{\k+1}+2(2m+2n+\k)\,\Umn{m,n}{\k-1}\,\Umn{m-1,n}{\k+1}, 
\end{align}\end{subequations}
where $\D_{z}(f\cdot g)$ is the Hirota operator \eqref{Hirop},
see equation (4.16) and Proposition 4.1 in \cite{refMOK}.
From equations (A.12) and (A.13) in \cite[Appendix A]{refMasuda25} one can derive the following differential-difference equations for the generalised Umemura polynomials
\begin{subequations}\begin{align}
2\Hir{\Umn{m,n-1}{\k}}{\Umn{m-1,n}{\k}} &=\Umn{m,n-1}{\k}\,\Umn{m-1,n}{\k}-\Umn{m,n}{\k}\,\Umn{m-1,n-1}{\k}, \label{eq:Masuda1}\\
2\Hir{\Umn{m,n-1}{\k}}{\Umn{m-1,n}{\k}} &=\Umn{m,n}{\k-2}\,\Umn{m-1,n-1}{\k+2}-\Umn{m,n-1}{\k}\,\Umn{m-1,n}{\k}, \label{eq:Masuda2} \\
2\Hir{\Umn{m,n}{\k-1}}{\Umn{m-1,n-1}{\k+1}} &=\Umn{m,n}{\k-1}\,\Umn{m-1,n-1}{\k+1}-\Umn{m,n-1}{\k+1}\,\Umn{m-1,n}{\k-1}, \label{eq:Masuda3}\\
2\Hir{\Umn{m,n}{\k-1}}{\Umn{m-1,n-1}{\k+1}} &=\Umn{m,n-1}{\k-1}\,\Umn{m-1,n}{\k+1}-\Umn{m,n}{\k-1}\,\Umn{m-1,n-1}{\k+1}.\label{eq:Masuda4}
\end{align}\end{subequations}
By adding and subtracting these equations we obtain the differential-difference equations
\begin{subequations}\begin{align}
4\Hir{\Umn{m,n-1}{\k}}{\Umn{m-1,n}{\k}} &=\Umn{m,n}{\k-2}\,\Umn{m-1,n-1}{\k+2}-\Umn{m,n}{\k}\,\Umn{m-1,n-1}{\k}, \label{eq:Masuda12a}\\
4\Hir{\Umn{m,n}{\k-1}}{\Umn{m-1,n-1}{\k+1}} &=\Umn{m,n-1}{\k-1}\,\Umn{m-1,n}{\k+1}-\Umn{m,n-1}{\k+1}\,\Umn{m-1,n}{\k-1}, \label{eq:Masuda34a}
\end{align}\end{subequations}
and the discrete equations
\begin{subequations}\begin{align}
2\Umn{m,n-1}{\k}\,\Umn{m-1,n}{\k} &=\Umn{m,n}{\k}\,\Umn{m-1,n-1}{\k}+\Umn{m,n}{\k-2}\,\Umn{m-1,n-1}{\k+2},\label{eq:Masuda12b}\\
2\Umn{m,n}{\k-1}\,\Umn{m-1,n-1}{\k+1}&=\Umn{m,n-1}{\k+1}\,\Umn{m-1,n}{\k-1}+\Umn{m,n-1}{\k-1}\,\Umn{m-1,n}{\k+1}.\label{eq:Masuda34b}
\end{align}\end{subequations}

Further differential-difference equations are given by 
\begin{subequations}\begin{align}
\Hir{\Umn{m+1,n}{\k-1}}{\Umn{m-1,n}{\k+1}} &=(2m+1)\Umn{m,n}{\k-1}\,\Umn{m,n}{\k+1},\label{eq:B5a}\\
\Hir{\Umn{m,n+1}{\k-1}}{\Umn{m,n-1}{\k+1}} &=(2n+1)\Umn{m,n}{\k-1}\,\Umn{m,n}{\k+1}.\label{eq:B5b}
\end{align}\end{subequations}

The following differential-difference equations arise from the definitions of $\wh{w}_{m,n}^{(\k)}$, $\wh{v}_{m,n}^{(\k)}$ and $\wh{u}_{m,n}^{(\k)}$ in Theorem \ref{thm:genUm} and Remark \ref{rem:410}
\begin{subequations}\label{eq:A12}\begin{align}
z\Hir{\Umn{m-1,n}{2\k}}{\Umn{m,n-1}{2\k+2}}&=(n+\k)\Umn{m-1,n}{2\k}\Umn{m,n-1}{2\k+2}-(m+\k)\Umn{m,n-1}{2\k}\,\Umn{m-1,n}{2\k+2},\\
z\Hir{\Umn{m,n-1}{2\k}}{\Umn{m-1,n}{2\k+2}}&=(m+\k)\Umn{m,n-1}{2\k}\Umn{m-1,n}{2\k+2}-(n+\k)\Umn{m-1,n}{2\k}\,\Umn{m,n-1}{2\k+2},\\
z\Hir{\Umn{m-1,n-1}{2\k+1}}{\Umn{m,n}{2\k+1}}&=\k\Umn{m-1,n-1}{2\k+1}\,\Umn{m,n}{2\k+1}- (m+n+\k)\Umn{m-1,n-1}{2\k+3}\,\Umn{m,n}{2\k-1},\\2z\Hir{\Umn{m+1,n}{\k-1}}{\Umn{m-1,n}{\k+1}}&=(2m+1)\left\{\Umn{m-1,n}{\k+1}\,\Umn{m+1,n}{\k-1}+\Umn{m,n-1}{\k+1}\,\Umn{m,n+1}{\k-1}\right\},\label{eq:B6d}\\
2z\Hir{\Umn{m,n+1}{\k-1}}{\Umn{m,n-1}{\k+1}}&=(2n+1)\left\{\Umn{m-1,n}{\k+1}\,\Umn{m+1,n}{\k-1}+\Umn{m,n-1}{\k+1}\,\Umn{m,n+1}{\k-1}\right\}.\label{eq:B6e}
\end{align}\end{subequations}
Consequently from \eqref{eq:B5a} and \eqref{eq:B6d}, or from \eqref{eq:B5b} and \eqref{eq:B6e}, we obtain
\beq 2z\Umn{m,n}{\k-1}\,\Umn{m,n}{\k+1}=\Umn{m-1,n}{\k+1}\,\Umn{m+1,n}{\k-1}+\Umn{m,n-1}{\k+1}\,\Umn{m,n+1}{\k-1}.\label{eq:B7}\eeq
The generalised Umemura polynomials $\Umn{m,n}{\k}(z)$ satisfy the differential-difference equations 
\begin{subequations}\label{sys41}\begin{align}
\Umn{m+1,n}{\k}\,\Umn{m-1,n}{\k}&=4z\D_{z}^2\Big(\Umn{m,n}{\k}\cdot\Umn{m,n}{\k}\Big)
+8\Umn{m,n}{\k}\deriv{\Umn{m,n}{\k}}{z} 
+(z-6m-2n-2\k-2)\left(\Umn{m,n}{\k}\right)^{\!2}, \label{eq:Umn1}\\
\Umn{m,n+1}{\k}\,\Umn{m,n-1}{\k}&=4z\D_{z}^2\Big(\Umn{m,n}{\k}\cdot\Umn{m,n}{\k}\Big)
+8\Umn{m,n}{\k}\deriv{\Umn{m,n}{\k}}{z} 
+(z+2m+6n+2\k+2)\left(\Umn{m,n}{\k}\right)^{\!2}, \label{eq:Umn2}
\end{align}\end{subequations}
and
\begin{subequations}\label{sys44}\begin{align}
\Umn{m+1,n}{\k-2}\,\Umn{m-1,n}{\k+2}&=4z\D_{z}^2\Big(\Umn{m,n}{\k}\cdot\Umn{m,n}{\k}\Big)
+8\Umn{m,n}{\k}\deriv{\Umn{m,n}{\k}}{z} 
+(z+2m-2n-2\k+2)\left(\Umn{m,n}{\k}\right)^{\!2},\\
\Umn{m,n+1}{\k-2}\,\Umn{m,n-1}{\k+2}&=4z\D_{z}^2\Big(\Umn{m,n}{\k}\cdot\Umn{m,n}{\k}\Big)
+8\Umn{m,n}{\k}\deriv{\Umn{m,n}{\k}}{z} 
+(z+2m-2n+2\k-2)\left(\Umn{m,n}{\k}\right)^{\!2},
\end{align}\end{subequations}
with $\Umn{-1,-1}{(\k)}(z)=\Umn{-1,0}{(\k)}(z)=\Umn{0,-1}{(\k)}(z)=\Umn{0,0}{(\k)}(z)=1$.
Eliminating the derivative terms in \eqref{sys41} and \eqref{sys44} gives
\begin{subequations}\begin{align}
\Umn{m,n+1}{\k}\,\Umn{m,n-1}{\k}-\Umn{m+1,n}{\k}\,\Umn{m-1,n}{\k}&=4(2m+2n+\k+1)\left(\Umn{m,n}{\k}\right)^{\!2},\\
\Umn{m,n+1}{\k-2}\,\Umn{m,n-1}{\k+2}-\Umn{m+1,n}{\k-2}\,\Umn{m-1,n}{\k+2}&=4(\k-1)\left(\Umn{m,n}{\k}\right)^{\!2}.
\end{align}\end{subequations}
Also from \eqref{sys41} and \eqref{sys44} we obtain the difference equations
\begin{subequations}\begin{align}
\Umn{m,n+1}{\k-2}\,\Umn{m,n-1}{\k+2}-\Umn{m+1,n}{\k}\,\Umn{m-1,n}{\k}&=4(2m+\k)\left(\Umn{m,n}{\k}\right)^{\!2},\\
\Umn{m,n+1}{\k}\,\Umn{m,n-1}{\k}-\Umn{m+1,n}{\k-2}\,\Umn{m-1,n}{\k+2}&=4(2n+\k)\left(\Umn{m,n}{\k}\right)^{\!2},\\
\Umn{m+1,n}{\k-2}\,\Umn{m-1,n}{\k+2}-\Umn{m+1,n}{\k}\,\Umn{m-1,n}{\k}&=4(2m+1)\left(\Umn{m,n}{\k}\right)^{\!2},\\
\Umn{m,n+1}{\k}\,\Umn{m,n-1}{\k}-\Umn{m,n+1}{\k-2}\,\Umn{m,n-1}{\k+2}&=4(2n+1)\left(\Umn{m,n}{\k}\right)^{\!2}.
\end{align}\end{subequations}

\def\refjl#1#2#3#4#5#6#7{\vspace{-0.2cm}
\bibitem{#1}\textrm{\frenchspacing#2}, \textrm{#3},
\textit{\frenchspacing#4}, \textbf{#5}\ (#6) #7.}

\def\refbk#1#2#3#4#5{\vspace{-0.2cm}
\bibitem{#1} \textrm{\frenchspacing#2}, \textit{#3}, #4 (#5).}

\def\refcf#1#2#3#4#5#6#7{\vspace{-0.2cm}
\bibitem{#1} \textrm{\frenchspacing#2}, \textrm{#3},
in: \textit{#4}, {\frenchspacing#5}, #6 (#7).}

\def\refpp#1#2#3#4{\vspace{-0.2cm}
\bibitem{#1} \textrm{\frenchspacing#2}, \textrm{#3}, #4.}

\def\refirnm#1#2#3#4#5{\vspace{-0.2cm}
\bibitem{#1}\textrm{\frenchspacing#2}, \textrm{#3},
\textit{\frenchspacing Int. Math. Res. Not. IMRN} (#4)\ #5.}
\def\NMJ{Nagoya Math. J.}

\def\refjsm#1#2#3#4#5#6#7{\vspace{-0.2cm}
\bibitem{#1}\textrm{\frenchspacing#2}, \textrm{#3},
\textit{\frenchspacing#4} (#6)\ #7.}

\def\fit#1{\textit{\frenchspacing#1}}
{}

\end{document}